\documentclass[12pt]{ociamthesis}  
\usepackage[answerdelayed]{exercise}
\usepackage{docmute}
\usepackage[hyperindex]{hyperref}
\usepackage{enumitem,multirow}
\usepackage{amsmath,amsfonts,amssymb,mathtools,xspace}
\usepackage{amsthm}
\usepackage{tikz}
\usetikzlibrary{decorations.pathreplacing}
\usetikzlibrary{decorations.markings}
\usetikzlibrary{calc}
\usepackage{etoolbox}
\usepackage{datetime}
\tikzstyle{dc}   = [circle, minimum width=8pt, draw, inner sep=0pt, path picture={\draw (path picture bounding box.south east) -- (path picture bounding box.north west) (path picture bounding box.south west) -- (path picture bounding box.north east);}]
\tikzstyle{dp}   = [circle, minimum width=8pt, draw, inner sep=0pt, path picture={\draw (path picture bounding box.west) -- (path picture bounding box.north) (path picture bounding box.south west) -- (path picture bounding box.north) (path picture bounding box.south west) -- (path picture bounding box.north east) (path picture bounding box.north east) -- (path picture bounding box.south) (path picture bounding box.south) -- (path picture bounding box.east);}]
\tikzstyle{ls} = [dc,scale=0.65]
\tikzstyle{ls'} = [dp,scale=1.3]
\tikzstyle{dpt}  = [circle, minimum width=8pt, draw, inner sep=0pt, path picture={\draw (path picture bounding box.south) -- (path picture bounding box.north) (path picture bounding box.west) -- (path picture bounding box.east);}]
\tikzstyle{agg} = [dpt,scale=0.65]

\numberwithin{Exercise}{section}

\newenvironment{pic}[1][]
{\begin{aligned}\begin{tikzpicture}[#1]}
{\end{tikzpicture}\end{aligned}}
\newcommand{\edges}[1][]%
{
}

\makeatletter
\def\calign@preamble{%
   &\hfil\strut@
    \setboxz@h{\@lign$\m@th\displaystyle{##}$}%
    \ifmeasuring@\savefieldlength@\fi
    \set@field
    \hfil
    \tabskip\alignsep@
}
\let\cmeasure@\measure@
\patchcmd\cmeasure@{\divide\@tempcntb\tw@}{}{}{}
\patchcmd\cmeasure@{\divide\@tempcntb\tw@}{}{}{}
\patchcmd\cmeasure@{\ifodd\maxfields@
  \global\advance\maxfields@\@ne
  \fi}{}{}{}    
  
\makeatother


\renewcommand\matrix[1]
{\hspace{-1.5pt}\left( \hspace{-2pt} \raisebox{-0pt}{$\begin{array}{@{\extracolsep{-5.5pt}\,}cccc}#1\end{array}$}  \hspace{0pt} \right)\hspace{-1.5pt}}

\newcommand\tinymatrix[1]
{\left( \hspace{-2pt} \renewcommand\thickspace{\kern2pt} \scriptstyle\begin{smallmatrix} #1 \end{smallmatrix} \hspace{-2pt} \right)}

\newcommand\ignore[1]{}

\pgfdeclarelayer{foreground}
\pgfdeclarelayer{background}
\pgfsetlayers{main,foreground,background}

\makeatletter
\pgfkeys{%
  /tikz/on layer/.code={
    \pgfonlayer{#1}\begingroup
    \aftergroup\endpgfonlayer
    \aftergroup\endgroup
  },
  /tikz/node on layer/.code={
    \gdef\node@@on@layer{%
      \setbox\tikz@tempbox=\hbox\bgroup\pgfonlayer{#1}\unhbox\tikz@tempbox\endpgfonlayer\egroup}
    \aftergroup\node@on@layer
  },
  /tikz/end node on layer/.code={
    \endpgfonlayer\endgroup\endgroup
  }
}
\def\node@on@layer{\aftergroup\node@@on@layer}
\makeatother\def\thickness{0.7pt}
\tikzstyle{oldmorphism}=[minimum width=30pt, minimum height=16pt, draw, font=\small, inner sep=0pt, fill=white, line width=\thickness]
\tikzstyle{cross}=[preaction={draw=white, -, line width=10pt}]
\tikzstyle{braid}=[double=black, line width=3*\thickness, double distance=\thickness, white]
\tikzstyle{string}=[line width=\thickness]
\tikzstyle{scalar}=[circle, inner sep=0pt, minimum width=15pt, draw, line width=\thickness]
\tikzstyle{dot}=[circle, draw=black, fill=black!25, inner sep=.4ex, line width=\thickness, node on layer=foreground]
\tikzstyle{blackdot}=[circle, draw=black, fill=black!35, inner sep=.4ex, line width=\thickness, node on layer=foreground]
\tikzstyle{whitedot}=[circle, draw=black, fill=white, inner sep=.4ex, line width=\thickness, node on layer=foreground]
\tikzstyle{mixedmorphism}=[morphism, minimum width=30pt, minimum height=16pt, draw, font=\small, inner sep=0pt, fill=white, line width=\thickness,rounded corners=1ex]
\tikzstyle{thick}=[line width=\thickness]
\tikzstyle{tiny}=[font=\tiny]

\tikzset{arrow/.style={decoration={
    markings,
    mark=at position #1 with \arrow{thickarrow}},
    postaction=decorate}
}
\tikzset{reverse arrow/.style={decoration={
    markings,
    mark=at position #1 with \arrow{reversethickarrow}},
    postaction=decorate}
}

\newif\ifblack\pgfkeys{/tikz/black/.is if=black}
\newif\ifwedge\pgfkeys{/tikz/wedge/.is if=wedge}
\newif\ifvflip\pgfkeys{/tikz/vflip/.is if=vflip}
\newif\ifhflip\pgfkeys{/tikz/hflip/.is if=hflip}
\newif\ifhvflip\pgfkeys{/tikz/hvflip/.is if=hvflip}
\newif\ifconnectnw\pgfkeys{/tikz/connect nw/.is if=connectnw}
\newif\ifconnectne\pgfkeys{/tikz/connect ne/.is if=connectne}
\newif\ifconnectsw\pgfkeys{/tikz/connect sw/.is if=connectsw}
\newif\ifconnectse\pgfkeys{/tikz/connect se/.is if=connectse}
\newif\ifconnectn\pgfkeys{/tikz/connect n/.is if=connectn}
\newif\ifconnects\pgfkeys{/tikz/connect s/.is if=connects}
\newif\ifconnectnwf\pgfkeys{/tikz/connect nw >/.is if=connectnwf}
\newif\ifconnectnef\pgfkeys{/tikz/connect ne >/.is if=connectnef}
\newif\ifconnectswf\pgfkeys{/tikz/connect sw >/.is if=connectswf}
\newif\ifconnectsef\pgfkeys{/tikz/connect se >/.is if=connectsef}
\newif\ifconnectnf\pgfkeys{/tikz/connect n >/.is if=connectnf}
\newif\ifconnectsf\pgfkeys{/tikz/connect s >/.is if=connectsf}
\newif\ifconnectnwr\pgfkeys{/tikz/connect nw </.is if=connectnwr}
\newif\ifconnectner\pgfkeys{/tikz/connect ne </.is if=connectner}
\newif\ifconnectswr\pgfkeys{/tikz/connect sw </.is if=connectswr}
\newif\ifconnectser\pgfkeys{/tikz/connect se </.is if=connectser}
\newif\ifconnectnr\pgfkeys{/tikz/connect n </.is if=connectnr}
\newif\ifconnectsr\pgfkeys{/tikz/connect s </.is if=connectsr}
\tikzset{keylengthnw/.initial=\connectheight}
\tikzset{keylengthn/.initial =\connectheight}
\tikzset{keylengthne/.initial=\connectheight}
\tikzset{keylengthsw/.initial=\connectheight}
\tikzset{keylengths/.initial =\connectheight}
\tikzset{keylengthse/.initial=\connectheight}
\tikzset{connect nw length/.style={connect nw=true, keylengthnw={#1}}}
\tikzset{connect n length/.style ={connect n =true, keylengthn ={#1}}}
\tikzset{connect ne length/.style={connect ne=true, keylengthne={#1}}}
\tikzset{connect sw length/.style={connect sw=true, keylengthsw={#1}}}
\tikzset{connect s length/.style ={connect s =true, keylengths ={#1}}}
\tikzset{connect se length/.style={connect se=true, keylengthse={#1}}}
\tikzset{connect nw < length/.style={connect nw <=true, keylengthnw={#1}}}
\tikzset{connect n < length/.style ={connect n <=true,  keylengthn ={#1}}}
\tikzset{connect ne < length/.style={connect ne <=true, keylengthne={#1}}}
\tikzset{connect sw < length/.style={connect sw <=true, keylengthnw={#1}}}
\tikzset{connect s < length/.style ={connect s <=true,  keylengths ={#1}}}
\tikzset{connect se < length/.style={connect se <=true, keylengthse={#1}}}
\tikzset{connect nw > length/.style={connect nw >=true, keylengthnw={#1}}}
\tikzset{connect n > length/.style ={connect n >=true,  keylengthn ={#1}}}
\tikzset{connect ne > length/.style={connect ne >=true, keylengthne={#1}}}
\tikzset{connect sw > length/.style={connect sw >=true, keylengthsw={#1}}}
\tikzset{connect s > length/.style ={connect s >=true,  keylengths ={#1}}}
\tikzset{connect se > length/.style={connect se >=true, keylengthse={#1}}}

\newlength\morphismheight
\newlength\minimummorphismwidth
\newlength\stateheight
\newlength\minimumstatewidth
\newlength\connectheight
\tikzset{width/.initial=\minimummorphismwidth}

\makeatletter
\pgfarrowsdeclare{thickarrow}{thickarrow}
{
  \pgfutil@tempdima=-0.84pt%
  \advance\pgfutil@tempdima by-1.3\pgflinewidth%
  \pgfutil@tempdimb=-1.7pt%
  \advance\pgfutil@tempdimb by.625\pgflinewidth%
  \pgfarrowsleftextend{+\pgfutil@tempdima}
  \pgfarrowsrightextend{+\pgfutil@tempdimb}
}
{
  \pgfmathparse{\pgfgetarrowoptions{thickarrow}}%
  \pgfsetlinewidth{1.25 pt}
  \pgfutil@tempdima=0.28pt%
  \advance\pgfutil@tempdima by.3\pgflinewidth%
  \pgfsetlinewidth{0.8\pgflinewidth}
  \pgfsetdash{}{+0pt}
  \pgfsetroundcap
  \pgfsetroundjoin
  \pgfpathmoveto{\pgfqpoint{-3\pgfutil@tempdima}{4\pgfutil@tempdima}}
  \pgfpathcurveto
  {\pgfqpoint{-2.75\pgfutil@tempdima}{2.5\pgfutil@tempdima}}
  {\pgfqpoint{0pt}{0.25\pgfutil@tempdima}}
  {\pgfqpoint{0.75\pgfutil@tempdima}{0pt}}
  \pgfpathcurveto
  {\pgfqpoint{0pt}{-0.25\pgfutil@tempdima}}
  {\pgfqpoint{-2.75\pgfutil@tempdima}{-2.5\pgfutil@tempdima}}
  {\pgfqpoint{-3\pgfutil@tempdima}{-4\pgfutil@tempdima}}
  \pgfusepathqstroke
}
\pgfarrowsdeclare{reversethickarrow}{reversethickarrow}
{
  \pgfutil@tempdima=-0.84pt%
  \advance\pgfutil@tempdima by-1.3\pgflinewidth%
  \pgfutil@tempdimb=0.2pt%
  \advance\pgfutil@tempdimb by.625\pgflinewidth%
  \pgfarrowsleftextend{+\pgfutil@tempdima}
  \pgfarrowsrightextend{+\pgfutil@tempdimb}
}
{
  \pgftransformxscale{-1}
  \pgfmathparse{\pgfgetarrowoptions{thickarrow}}%
  \ifpgfmathunitsdeclared%
    \pgfmathparse{\pgfmathresult pt}%
  \else%
    \pgfmathparse{\pgfmathresult*\pgflinewidth}%
  \fi%
  \let\thickness=\pgfmathresult
  \pgfsetlinewidth{1.25 pt}
  \pgfutil@tempdima=0.28pt%
  \advance\pgfutil@tempdima by.3\pgflinewidth%
  \pgfsetlinewidth{0.8\pgflinewidth}
  \pgfsetdash{}{+0pt}
  \pgfsetroundcap
  \pgfsetroundjoin
  \pgfpathmoveto{\pgfqpoint{-3\pgfutil@tempdima}{4\pgfutil@tempdima}}
  \pgfpathcurveto
  {\pgfqpoint{-2.75\pgfutil@tempdima}{2.5\pgfutil@tempdima}}
  {\pgfqpoint{0pt}{0.25\pgfutil@tempdima}}
  {\pgfqpoint{0.75\pgfutil@tempdima}{0pt}}
  \pgfpathcurveto
  {\pgfqpoint{0pt}{-0.25\pgfutil@tempdima}}
  {\pgfqpoint{-2.75\pgfutil@tempdima}{-2.5\pgfutil@tempdima}}
  {\pgfqpoint{-3\pgfutil@tempdima}{-4\pgfutil@tempdima}}
  \pgfusepathqstroke
}
\makeatother

\makeatletter
\pgfdeclareshape{ground}
{
    \savedanchor\centerpoint
    {
        \pgf@x=0pt
        \pgf@y=0pt
    }
    \anchor{center}{\centerpoint}
    \anchorborder{\centerpoint}
    \anchor{north}
    {
        \pgf@x=0pt
        \pgf@y=0.5*0.33*\stateheight
    }
    \anchor{south}
    {
        \pgf@x=0pt
        \pgf@y=0pt
    }
    \saveddimen\overallwidth
    {
        \pgfkeysgetvalue{/pgf/minimum width}{\minwidth}
        \pgf@x=\minimumstatewidth
        \ifdim\pgf@x<\minwidth
            \pgf@x=\minwidth
        \fi
    }
    \backgroundpath
    {
        \begin{pgfonlayer}{foreground}
        \pgfsetstrokecolor{black}
        \pgfsetlinewidth{1.25pt}
        \ifhflip
            \pgftransformyscale{-1}
        \fi
        \pgftransformscale{0.5}
        \pgfpathmoveto{\pgfpoint{-0.5*\overallwidth}{0}}
        \pgfpathlineto{\pgfpoint{0.5*\overallwidth}{0}}
        \pgfpathmoveto{\pgfpoint{-0.33*\overallwidth}{0.33*\stateheight}}
        \pgfpathlineto{\pgfpoint{0.33*\overallwidth}{0.33*\stateheight}}
        \pgfpathmoveto{\pgfpoint{-0.16*\overallwidth}{0.66*\stateheight}}
        \pgfpathlineto{\pgfpoint{0.16*\overallwidth}{0.66*\stateheight}}
        \pgfpathmoveto{\pgfpoint{-0.02*\overallwidth}{\stateheight}}
        \pgfpathlineto{\pgfpoint{0.02*\overallwidth}{\stateheight}}
        \pgfusepath{stroke}
        \end{pgfonlayer}
    }
}
\tikzset{forward arrow style/.style={every to/.style, decoration={
    markings,
    mark=at position 0.5 with \arrow{thickarrow}},
    postaction=decorate}}
\tikzset{reverse arrow style/.style={every to/.style, decoration={
    markings,
    mark=at position 0.5 with \arrow{reversethickarrow}},
    postaction=decorate}}
\pgfdeclareshape{morphism}
{
    \savedanchor\centerpoint
    {
        \pgf@x=0pt
        \pgf@y=0pt
    }
    \anchor{center}{\centerpoint}
    \anchorborder{\centerpoint}
    \saveddimen\savedlengthnw
    {
        \pgfkeysgetvalue{/tikz/keylengthnw}{\len}
        \pgf@x=\len
    }
    \saveddimen\savedlengthn
    {
        \pgfkeysgetvalue{/tikz/keylengthn}{\len}
        \pgf@x=\len
    }
    \saveddimen\savedlengthne
    {
        \pgfkeysgetvalue{/tikz/keylengthne}{\len}
        \pgf@x=\len
    }
    \saveddimen\savedlengthsw
    {
        \pgfkeysgetvalue{/tikz/keylengthsw}{\len}
        \pgf@x=\len
    }
    \saveddimen\savedlengths
    {
        \pgfkeysgetvalue{/tikz/keylengths}{\len}
        \pgf@x=\len
    }
    \saveddimen\savedlengthse
    {
        \pgfkeysgetvalue{/tikz/keylengthse}{\len}
        \pgf@x=\len
    }
    \saveddimen\overallwidth
    {
        \pgfkeysgetvalue{/tikz/width}{\minwidth}
        \pgf@x=\wd\pgfnodeparttextbox
        \ifdim\pgf@x<\minwidth
            \pgf@x=\minwidth
        \fi
    }
    \savedanchor{\upperrightcorner}
    {
        \pgf@y=.5\ht\pgfnodeparttextbox
        \advance\pgf@y by -.5\dp\pgfnodeparttextbox
        \pgf@x=.5\wd\pgfnodeparttextbox
    }
    \anchor{north}
    {
        \pgf@x=0pt
        \pgf@y=0.5\morphismheight
    }
    \anchor{north east}
    {
        \pgf@x=\overallwidth
        \multiply \pgf@x by 2
        \divide \pgf@x by 5
        \pgf@y=0.5\morphismheight
    }
    \anchor{east}
    {
        \pgf@x=\overallwidth
        \divide \pgf@x by 2
        \advance \pgf@x by 5pt
        \pgf@y=0pt
    }
    \anchor{west}
    {
        \pgf@x=-\overallwidth
        \divide \pgf@x by 2
        \advance \pgf@x by -5pt
        \pgf@y=0pt
    }
    \anchor{north west}
    {
        \pgf@x=-\overallwidth
        \multiply \pgf@x by 2
        \divide \pgf@x by 5
        \pgf@y=0.5\morphismheight
    }
    \anchor{connect nw}
    {
        \pgf@x=-\overallwidth
        \multiply \pgf@x by 2
        \divide \pgf@x by 5
        \pgf@y=0.5\morphismheight
        \advance\pgf@y by \savedlengthnw
    }
    \anchor{connect ne}
    {
        \pgf@x=\overallwidth
        \multiply \pgf@x by 2
        \divide \pgf@x by 5
        \pgf@y=0.5\morphismheight
        \advance\pgf@y by \savedlengthne
    }
    \anchor{connect sw}
    {
        \pgf@x=-\overallwidth
        \multiply \pgf@x by 2
        \divide \pgf@x by 5
        \pgf@y=-0.5\morphismheight
        \advance\pgf@y by -\savedlengthsw
    }
    \anchor{connect se}
    {
        \pgf@x=\overallwidth
        \multiply \pgf@x by 2
        \divide \pgf@x by 5
        \pgf@y=-0.5\morphismheight
        \advance\pgf@y by -\savedlengthse
    }
    \anchor{connect n}
    {
        \pgf@x=0pt
        \pgf@y=0.5\morphismheight
        \advance\pgf@y by \savedlengthn
    }
    \anchor{connect s}
    {
        \pgf@x=0pt
        \pgf@y=-0.5\morphismheight
        \advance\pgf@y by -\savedlengths
    }
    \anchor{south east}
    {
        \pgf@x=\overallwidth
        \multiply \pgf@x by 2
        \divide \pgf@x by 5
        \pgf@y=-0.5\morphismheight
    }
    \anchor{south west}
    {
        \pgf@x=-\overallwidth
        \multiply \pgf@x by 2
        \divide \pgf@x by 5
        \pgf@y=-0.5\morphismheight
    }
    \anchor{south}
    {
        \pgf@x=0pt
        \pgf@y=-0.5\morphismheight
    }
    \anchor{text}
    {
        \upperrightcorner
        \pgf@x=-\pgf@x
        \pgf@y=-\pgf@y
    }
    \backgroundpath
    {
        \pgfsetstrokecolor{black}
        \pgfsetlinewidth{\thickness}
        \begin{scope}
                \ifhflip
                    \pgftransformyscale{-1}
                \fi
                \ifvflip
                    \pgftransformxscale{-1}
                \fi
                \ifhvflip
                    \pgftransformxscale{-1}
                    \pgftransformyscale{-1}
                \fi
                \pgfpathmoveto{\pgfpoint
                    {-0.5*\overallwidth-5pt}
                    {0.5*\morphismheight}}
                \pgfpathlineto{\pgfpoint
                    {0.5*\overallwidth+5pt}
                    {0.5*\morphismheight}}
                \ifwedge
                    \pgfpathlineto{\pgfpoint
                        {0.5*\overallwidth + 15pt}
                        {-0.5*\morphismheight}}
                \else
                    \pgfpathlineto{\pgfpoint
                        {0.5*\overallwidth + 5pt}
                        {-0.5*\morphismheight}}
                \fi
                \pgfpathlineto{\pgfpoint
                    {-0.5*\overallwidth-5pt}
                    {-0.5*\morphismheight}}
                \pgfpathclose
                \pgfusepath{stroke}
        \end{scope}
        \ifconnectnw
            \pgfpathmoveto{\pgfpoint
                {-0.4*\overallwidth}
                {0.5*\morphismheight}}
            \pgfpathlineto{\pgfpoint
                {-0.4*\overallwidth}
                {0.5*\morphismheight+\savedlengthnw}}
            \pgfusepath{stroke}
        \fi
        \ifconnectne
            \pgfpathmoveto{\pgfpoint
                {0.4*\overallwidth}
                {0.5*\morphismheight}}
            \pgfpathlineto{\pgfpoint
                {0.4*\overallwidth}
                {0.5*\morphismheight+\savedlengthne}}
            \pgfusepath{stroke}
        \fi
        \ifconnectsw
            \pgfpathmoveto{\pgfpoint
                {-0.4*\overallwidth}
                {-0.5*\morphismheight}}
            \pgfpathlineto{\pgfpoint
                {-0.4*\overallwidth}
                {-0.5*\morphismheight-\savedlengthsw}}
            \pgfusepath{stroke}
        \fi
        \ifconnectse
            \pgfpathmoveto{\pgfpoint
                {0.4*\overallwidth}
                {-0.5*\morphismheight}}
            \pgfpathlineto{\pgfpoint
                {0.4*\overallwidth}
                {-0.5*\morphismheight-\savedlengthse}}
            \pgfusepath{stroke}
        \fi
        \ifconnectn
            \pgfpathmoveto{\pgfpoint
                {0pt}
                {0.5*\morphismheight}}
            \pgfpathlineto{\pgfpoint
                {0pt}
                {0.5*\morphismheight+\savedlengthn}}
            \pgfusepath{stroke}
        \fi
        \ifconnects
            \pgfpathmoveto{\pgfpoint
                {0pt}
                {-0.5*\morphismheight}}
            \pgfpathlineto{\pgfpoint
                {0pt}
                {-0.5*\morphismheight-\savedlengths}}
            \pgfusepath{stroke}
        \fi
        \ifconnectnwf
            \draw [forward arrow style] (-0.4*\overallwidth,0.5*\morphismheight)
                to (-0.4*\overallwidth,0.5*\morphismheight+\savedlengthnw);
        \fi
        \ifconnectnef
            \draw [forward arrow style] (0.4*\overallwidth,0.5*\morphismheight)
                to (0.4*\overallwidth,0.5*\morphismheight+\savedlengthne);
        \fi
        \ifconnectswf
            \draw [forward arrow style] (-0.4*\overallwidth,-0.5*\morphismheight-\savedlengthsw)
                to (-0.4*\overallwidth,-0.5*\morphismheight);
        \fi
        \ifconnectsef
            \draw [forward arrow style] (0.4*\overallwidth,-0.5*\morphismheight-\savedlengthse)
                to (0.4*\overallwidth,-0.5*\morphismheight);
        \fi
        \ifconnectnf
            \draw [forward arrow style] (0,0.5*\morphismheight)
                to (0,0.5*\morphismheight+\savedlengthn);
        \fi
        \ifconnectsf
            \draw [forward arrow style] (0,-0.5*\morphismheight-\savedlengths)
                to (0,-0.5*\morphismheight);
        \fi
        \ifconnectnwr
            \draw [reverse arrow style] (-0.4*\overallwidth,0.5*\morphismheight)
                to (-0.4*\overallwidth,0.5*\morphismheight+\savedlengthnw);
        \fi
        \ifconnectner
            \draw [reverse arrow style] (0.4*\overallwidth,0.5*\morphismheight)
                to (0.4*\overallwidth,0.5*\morphismheight+\savedlengthne);
        \fi
        \ifconnectswr
            \draw [reverse arrow style] (-0.4*\overallwidth,-0.5*\morphismheight-\savedlengthsw)
                to (-0.4*\overallwidth,-0.5*\morphismheight);
        \fi
        \ifconnectser
            \draw [reverse arrow style] (0.4*\overallwidth,-0.5*\morphismheight-\savedlengthse)
                to (0.4*\overallwidth,-0.5*\morphismheight);
        \fi
        \ifconnectnr
            \draw [reverse arrow style] (0,0.5*\morphismheight)
                to (0,0.5*\morphismheight+\savedlengthn);
        \fi
        \ifconnectsr
            \draw [reverse arrow style] (0,-0.5*\morphismheight-\savedlengths)
                to (0,-0.5*\morphismheight);
        \fi
    }
}
\pgfdeclareshape{swish right}
{
    \savedanchor\centerpoint
    {
        \pgf@x=0pt
        \pgf@y=0pt
    }
    \anchor{center}{\centerpoint}
    \anchorborder{\centerpoint}
    \anchor{north}
    {
        \pgf@x=\minimummorphismwidth
        \divide\pgf@x by 5
        \pgf@y=\morphismheight
        \divide\pgf@y by 2
        \advance\pgf@y by \connectheight
    }
    \anchor{south}
    {
        \pgf@x=-\minimummorphismwidth
        \divide\pgf@x by 5
        \pgf@y=-\morphismheight
        \divide\pgf@y by 2
        \advance\pgf@y by -\connectheight
    }
    \backgroundpath
    {
        \pgfsetstrokecolor{black}
        \pgfsetlinewidth{\thickness}
        \pgfpathmoveto{\pgfpoint
            {-0.2*\minimummorphismwidth}
            {-0.5*\morphismheight-\connectheight}}
        \pgfpathcurveto
            {\pgfpoint{-0.2*\minimummorphismwidth}{0pt}}
            {\pgfpoint{0.2*\minimummorphismwidth}{0pt}}
            {\pgfpoint
                {0.2*\minimummorphismwidth}
                {0.5*\morphismheight+\connectheight}}
        \pgfusepath{stroke}
    }
}
\pgfdeclareshape{swish left}
{
    \savedanchor\centerpoint
    {
        \pgf@x=0pt
        \pgf@y=0pt
    }
    \anchor{center}{\centerpoint}
    \anchorborder{\centerpoint}
    \anchor{north}
    {
        \pgf@x=-\minimummorphismwidth
        \divide\pgf@x by 5
        \pgf@y=\morphismheight
        \divide\pgf@y by 2
        \advance\pgf@y by \connectheight
    }
    \anchor{south}
    {
        \pgf@x=\minimummorphismwidth
        \divide\pgf@x by 5
        \pgf@y=-\morphismheight
        \divide\pgf@y by 2
        \advance\pgf@y by -\connectheight
    }
    \backgroundpath
    {
        \pgfsetstrokecolor{black}
        \pgfsetlinewidth{\thickness}
        \pgfpathmoveto{\pgfpoint
            {0.2*\minimummorphismwidth}
            {-0.5*\morphismheight-\connectheight}}
        \pgfpathcurveto
            {\pgfpoint{0.2*\minimummorphismwidth}{0pt}}
            {\pgfpoint{-0.2*\minimummorphismwidth}{0pt}}
            {\pgfpoint
                {-0.2*\minimummorphismwidth}
                {0.5*\morphismheight+\connectheight}}
        \pgfusepath{stroke}
    }
}
\pgfdeclareshape{state}
{
    \savedanchor\centerpoint
    {
        \pgf@x=0pt
        \pgf@y=0pt
    }
    \anchor{center}{\centerpoint}
    \anchorborder{\centerpoint}
    \saveddimen\overallwidth
    {
        \pgf@x=3\wd\pgfnodeparttextbox
        \ifdim\pgf@x<\minimumstatewidth
            \pgf@x=\minimumstatewidth
        \fi
    }
    \savedanchor{\upperrightcorner}
    {
        \pgf@x=.5\wd\pgfnodeparttextbox
        \pgf@y=.5\ht\pgfnodeparttextbox
        \advance\pgf@y by -.5\dp\pgfnodeparttextbox
    }
    \anchor{A}
    {
        \pgf@x=-\overallwidth
        \divide\pgf@x by 4
        \pgf@y=0pt
    }
    \anchor{B}
    {
        \pgf@x=\overallwidth
        \divide\pgf@x by 4
        \pgf@y=0pt
    }
    \anchor{text}
    {
        \upperrightcorner
        \pgf@x=-\pgf@x
        \ifhflip
            \pgf@y=-\pgf@y
            \advance\pgf@y by 0.4\stateheight
        \else
            \pgf@y=-\pgf@y
            \advance\pgf@y by -0.4\stateheight
        \fi
    }
    \backgroundpath
    {
        \pgfsetstrokecolor{black}
        \pgfsetlinewidth{\thickness}
        \pgfpathmoveto{\pgfpoint{-0.5*\overallwidth}{0}}
        \pgfpathlineto{\pgfpoint{0.5*\overallwidth}{0}}
        \ifhflip
            \pgfpathlineto{\pgfpoint{0}{\stateheight}}
        \else
            \pgfpathlineto{\pgfpoint{0}{-\stateheight}}
        \fi
        \pgfpathclose
        \ifblack
            \pgfsetfillcolor{black!35}
            \pgfusepath{fill,stroke}
        \else
            \pgfusepath{stroke}
        \fi
    }
}
\makeatother

\newcommand{\tinycomult}[1][dot]{
\smash{\raisebox{-2pt}{\hspace{-5pt}\ensuremath{\begin{pic}[scale=0.4,string]
    \node (0) at (0,0) {};
    \node[#1, inner sep=1.5pt] (1) at (0,0.55) {};
    \node (2) at (-0.5,1) {};
    \node (3) at (0.5,1) {};
    \draw (0.center) to (1.center);
    \draw (1.center) to [out=left, in=down, out looseness=1.5] (2.center);
    \draw (1.center) to [out=right, in=down, out looseness=1.5] (3.center);
\end{pic}
}\hspace{-3pt}}}}
\newcommand{\tinycomultls}[1][whitedot]{
\smash{\raisebox{-2pt}{\hspace{-5pt}\ensuremath{\begin{pic}[scale=0.4,string]
    \node (0) at (0,0) {};
    \node[ls,scale=0.9,#1, inner sep=1.5pt] (1) at (0,0.55) {};
    \node (2) at (-0.5,1) {};
    \node (3) at (0.5,1) {};
    \draw (0.center) to (1.center);
    \draw (1.center) to [out=left, in=down, out looseness=1.5] (2.center);
    \draw (1.center) to [out=right, in=down, out looseness=1.5] (3.center);
\end{pic}
}\hspace{-3pt}}}}

\newcommand{\tinycomultagg}[1][whitedot]{
\smash{\raisebox{-2pt}{\hspace{-5pt}\ensuremath{\begin{pic}[scale=0.4,string]
    \node (0) at (0,0) {};
    \node[agg,scale=0.9,#1, inner sep=1.5pt] (1) at (0,0.55) {};
    \node (2) at (-0.5,1) {};
    \node (3) at (0.5,1) {};
    \draw (0.center) to (1.center);
    \draw (1.center) to [out=left, in=down, out looseness=1.5] (2.center);
    \draw (1.center) to [out=right, in=down, out looseness=1.5] (3.center);
\end{pic}
}\hspace{-3pt}}}}

\newcommand{\tinycounit}[1][dot]{
\smash{\raisebox{-2pt}{\ensuremath{\hspace{-3pt}\begin{pic}[scale=0.4,string]
        \node (0) at (0,0) {};
        \node (1) at (0,1) {};
        \node[#1, inner sep=1.5pt] (d) at (0,0.55) {};
        \draw (0.center) to (d.center);
    \end{pic}
    \hspace{-1pt}}}}}
    \newcommand{\tinycounitls}[1][whitedot]{
\smash{\raisebox{-2pt}{\ensuremath{\hspace{-3pt}\begin{pic}[scale=0.4,string]
        \node (0) at (0,0) {};
        \node (1) at (0,1) {};
        \node[ls,scale=0.9,#1, inner sep=1.5pt] (d) at (0,0.55) {};
        \draw (0.center) to (d.center);
    \end{pic}
    \hspace{-1pt}}}}}
\newcommand{\tinymult}[1][dot]{
\smash{\raisebox{-2pt}{\hspace{-5pt}\ensuremath{\begin{pic}[scale=0.4,string,yscale=-1]
    \node (0) at (0,0) {};
    \node[#1, inner sep=1.5pt] (1) at (0,0.55) {};
    \node (2) at (-0.5,1) {};
    \node (3) at (0.5,1) {};
    \draw (0.center) to (1.center);
    \draw (1.center) to [out=left, in=down, out looseness=1.5] (2.center);
    \draw (1.center) to [out=right, in=down, out looseness=1.5] (3.center);
\end{pic}
}\hspace{-3pt}}}}

\newcommand{\tinymultk}[1][whitedot]{
\smash{\raisebox{-2pt}{\hspace{-5pt}\ensuremath{\begin{pic}[scale=0.4,string,yscale=-1]
    \node (0) at (0,0) {};
    \node[#1, inner sep=1.5pt,label={[xshift=-0.1cm,yshift=-0.1cm]0:{\tiny$k$}}] (1) at (0,0.55) {};
    \node (2) at (-0.5,1) {};
    \node (3) at (0.5,1) {};
    \draw (0.center) to (1.center);
    \draw (1.center) to [out=left, in=down, out looseness=1.5] (2.center);
    \draw (1.center) to [out=right, in=down, out looseness=1.5] (3.center);
\end{pic}
}\hspace{-3pt}}}}

\newcommand{\tinymultls}[1][whitedot]{
\smash{\raisebox{-2pt}{\hspace{-5pt}\ensuremath{\begin{pic}[scale=0.4,string,yscale=-1]
    \node (0) at (0,0) {};
    \node[ls,scale=0.9,#1, inner sep=1.5pt] (1) at (0,0.55) {};
    \node (2) at (-0.5,1) {};
    \node (3) at (0.5,1) {};
    \draw (0.center) to (1.center);
    \draw (1.center) to [out=left, in=down, out looseness=1.5] (2.center);
    \draw (1.center) to [out=right, in=down, out looseness=1.5] (3.center);
\end{pic}
}\hspace{-1pt}}}}

\newcommand{\tinymultagg}[1][whitedot]{
\smash{\raisebox{-2pt}{\hspace{-5pt}\ensuremath{\begin{pic}[scale=0.4,string,yscale=-1]
    \node (0) at (0,0) {};
    \node[agg,scale=0.9,#1, inner sep=1.5pt] (1) at (0,0.55) {};
    \node (2) at (-0.5,1) {};
    \node (3) at (0.5,1) {};
    \draw (0.center) to (1.center);
    \draw (1.center) to [out=left, in=down, out looseness=1.5] (2.center);
    \draw (1.center) to [out=right, in=down, out looseness=1.5] (3.center);
\end{pic}
}\hspace{-1pt}}}}

\newcommand{\tinyunit}[1][dot]{
\smash{\raisebox{-2pt}{\ensuremath{\hspace{-3pt}\begin{pic}[scale=0.4,string,yscale=-1]
        \node (0) at (0,0) {};
        \node (1) at (0,1) {};
        \node[#1, inner sep=1.5pt] (d) at (0,0.55) {};
        \draw (0.center) to (d.center);
    \end{pic}
    \hspace{-1pt}}}}}
    \newcommand{\tinyunitls}[1][whitedot]{
\smash{\raisebox{-2pt}{\ensuremath{\hspace{-3pt}\begin{pic}[scale=0.4,string,yscale=-1]
        \node (0) at (0,0) {};
        \node (1) at (0,1) {};
        \node[ls,scale=0.9,#1, inner sep=1.5pt] (d) at (0,0.55) {};
        \draw (0.center) to (d.center);
    \end{pic}
    \hspace{-1pt}}}}}
    
    \newcommand{\tinyunitagg}[1][whitedot]{
\smash{\raisebox{-2pt}{\ensuremath{\hspace{-3pt}\begin{pic}[scale=0.4,string,yscale=-1]
        \node (0) at (0,0) {};
        \node (1) at (0,1) {};
        \node[agg,scale=0.9,#1, inner sep=1.5pt] (d) at (0,0.55) {};
        \draw (0.center) to (d.center);
    \end{pic}
    \hspace{-1pt}}}}}

\newcommand{\tinyhandle}[1][dot]{\raisebox{-2pt}{\ensuremath{\hspace{-3pt}\begin{pic}[scale=0.4,string]
        \node (0) at (0,0) {};
        \node[dot, inner sep=1.0pt] (1) at (0,0.3) {};
        \node[dot, inner sep=1.0pt] (2) at (0,0.7) {};
        \node (3) at (0,1) {};
        \draw (0.center) to (1.center);
        \draw (2.center) to (3.center);
        \draw[in=180, out=180, looseness=2] (1.center) to (2.center);
        \draw[in=0, out=0, looseness=2] (1.center) to (2.center);
\end{pic}\hspace{-1pt}}}}

\newcommand{\tinycomultdb}[1][dot]{
\smash{\raisebox{-2pt}{\hspace{-5pt}\ensuremath{\begin{pic}[scale=0.4,string]
    \node (0) at (0,0) {};
    \node[#1, inner sep=1.5pt] (1) at (0,0.55) {};
    \node (2) at (-0.5,1) {};
    \node (3) at (0.5,1) {};
    \draw (0.center) to (1.center);
    \draw (1.center) to [out=left, in=down, out looseness=1.5] (2.center);
    \draw (1.center) to [out=right, in=down, out looseness=1.5] (3.center);
    \node (0') at (0.75,0) {};
    \node[#1, inner sep=1.5pt] (1') at (0.75,0.55) {};
    \node (2') at (0.25,1) {};
    \node (3') at (1.25,1) {};
    \draw (0'.center) to (1'.center);
    \draw (1'.center) to [out=left, in=down, out looseness=1.5] (2'.center);
    \draw (1'.center) to [out=right, in=down, out looseness=1.5] (3'.center);
\end{pic}
}\hspace{-3pt}}}}

\newcommand{\tinymultlss}[1][whitedot]{
\smash{\raisebox{-2pt}{\hspace{-5pt}\ensuremath{\begin{pic}[scale=0.4,string,yscale=-1]
    \node (0) at (0,0) {};
    \node[ls',scale=0.65,#1, inner sep=1.5pt] (1) at (0,0.55) {};
    \node (2) at (-0.5,1) {};
    \node (3) at (0.5,1) {};
    \draw (0.center) to (1.center);
    \draw (1.center) to [out=left, in=down, out looseness=1.5] (2.center);
    \draw (1.center) to [out=right, in=down, out looseness=1.5] (3.center);
\end{pic}
}\hspace{-1pt}}}}

\newcommand{\tinycomultlss}[1][whitedot]{
\smash{\raisebox{-2pt}{\hspace{-5pt}\ensuremath{\begin{pic}[scale=0.4,string]
    \node (0) at (0,0) {};
    \node[ls',scale=0.65,#1, inner sep=1.5pt] (1) at (0,0.55) {};
    \node (2) at (-0.5,1) {};
    \node (3) at (0.5,1) {};
    \draw (0.center) to (1.center);
    \draw (1.center) to [out=left, in=down, out looseness=1.5] (2.center);
    \draw (1.center) to [out=right, in=down, out looseness=1.5] (3.center);
\end{pic}
}\hspace{-3pt}}}}

\newcommand\cat[1]{\ensuremath{\mathbf{#1}}}
\renewcommand\dag{\ensuremath{\dagger}}

\newcommand\id[1][]{\ensuremath{\mathrm{id}_{#1}}}

\DeclareMathOperator{\Ob}{Ob}
\newcommand\ket[1]{\ensuremath{| #1 \rangle}}

\newcommand\C{\ensuremath{\mathbb{C}}}
\newcommand{\inprod}[2]{\ensuremath{\langle #1\hspace{0.5pt}|\hspace{0.5pt}#2 \rangle}}

\usepackage{color}

\makeindex
\usepackage{amssymb}

\theoremstyle{plain}
\newtheorem{theorem}{Theorem}[chapter] 

\newtheorem{proposition}[theorem]{Proposition}

\theoremstyle{lemma}
\newtheorem{lemma}[theorem]{Lemma}

\theoremstyle{definition} 
\newtheorem{definition}{Definition}[chapter] 
 
\newtheorem*{conj*}{Conjecture}

\newcommand{\tr}{\text{tr}}
\newcommand{\ob}{\text{Ob}}

\title{Exploring Quantum Teleportation \\[1ex]     
        through Unitary Error Bases}   

\author{Benjamin Musto}             
\college{Kellogg College}  

\renewcommand{\submittedtext}{A dissertation submitted for the degree of}
\degree{Master of Science: Mathematics and the Foundations of Computer Science}     
\degreedate{Trinity 2014}         

\begin{document}

\baselineskip=18pt plus1pt

\setcounter{secnumdepth}{3}
\setcounter{tocdepth}{3}

\maketitle                  
\begin{abstract}
\textit{Unitary error bases have a great number of applications across
quantum information and quantum computation, and are fundamentally
linked to quantum teleportation, dense coding and quantum error
correction. Werner's combinatorial construction builds a unitary error
basis from a family of Hadamard matrices and a latin square. In this
dissertation, I give a new categorical axiomatisation of latin
squares, and use this to give a fully graphical presentation and proof
of the correctness of Werner's construction. The categorical approach
makes clear that some of the latin square axioms are unnecessary for
the construction to go through, and I propose a generalised
construction scheme with the potential to create new classes of
unitary error bases.}
\end{abstract}
\begin{romanpages}          
\tableofcontents            
\end{romanpages}            
\setcounter{chapter}{-1}


\label{chap:introduction}

\setlength\minimummorphismwidth{7mm}
\tikzset{wedge=false}
\chapter{Introduction}
Unitary error bases came to prominence in the mid-1990s through papers by  Knill ~\cite{knill}, Steane ~\cite{steane} and others working on  quantum error correction. Quantum teleportation was a relatively late discovery in quantum theory, not conceived of until 1993 in this breakthrough paper ~\cite{teleportation}. It has since been physically realised using a range of methods. It utilises quantum entanglement, a phenomenon with no classical analogue, to communicate quantum information across potentially great distances.

Werner showed in 2001 that all `tight', quantum teleportation protocols, as well as dense coding protocols, can be precisely mathematically characterised as unitary error bases. Werner defined `tight' to mean protocols that are optimally realised from minimal resources in terms of Hilbert Space dimensions and classical bits \cite{werner2001all}. The proof in that paper was fairly involved and long. More recently the categorical quantum mechanics research programme has rendered the correspondence more apparent by use of the graphical calculus. 

New constructions for UEBs thus give us new ways to perform quantum teleportation and dense coding algorithms, and are thus highly sought after. My main result here is a generalisation of the combinatorial construction.

 So unitary error bases have now become even more important objects of research in quantum information, and various attempts have been made to classify them, leading  to two main constructions, giving two types of unitary error bases. The algebraic construction due to Knill ~\cite{knill}, giving a nice error basis, and the combinatorial construction due to Werner \cite{werner2001all} giving a shift and multiply basis. Another construction has come from the categorical quantum mechanics research programme \cite{toy}  . Given a pair of mutually unbiased bases, a quantum teleportation protocol, and thus a unitary error basis can be derived. I will call this the MUB construction, and the resulting unitary error basis an MUB error basis. 

In this dissertation I have achieved an explicit categorical presentation of the combinatorial construction. I have done so using the graphical calculus of categorical quantum mechanics. I have then proven, entirely graphically the correctness of the combinatorial construction. I have used this categorical presentation to generalise the construction, and presented a specific model that
does not produce trivially equivalent UEBs. I have also proven the correctness of this model, again completely graphically. 

As an essential step towards deriving a categorical presentation of the combinatorial construction, I have categorically axiomatised a latin square. This in itself, was a non-trivial undertaking, and has the potential to be applied beyond what I have done here.
\vspace{5mm}\\
\textit{Outline}\\
In Chapter 1, I will define a unitary error basis as well as the notion of equivalence between bases. I will then summarise the necessary background material on the graphical calculus of categorical quantum mechanics. \\In Chapter 2, I will define the combinatorial and MUB constructions as well as what I will call the `minimal combinatorial construction'. I will prove that MUB bases are unitary error bases. I will explicitly formulate both MUB and minimal combinatorial constructions graphically and prove that they are equivalent, thus showing that MUB error bases are isomorphic to a subset of shift and multiply bases. \\In Chapter 3, I will characterise latin squares in the graphical calculus and prove that only those that are associative obey the Frobenius law. \\In Chapter 4, I will explicitly formulate shift and multiply bases graphically, and present a purely graphical proof that they are unitary error bases.  \\In Chapter 5, I will generalise shift and multiply bases and give a specific model that is not trivially equivalent to a shift and multiply basis. I will end with a specific example of a generalised latin square (that is not a latin square) giving rise to a UEB.
\\Chapter 6 is the conclusion and in Chapter 7, I will outline some possible directions for further research.\\
In the appendix there is a glossary of the graphical axioms which I will refer to during graphical proofs.


\setlength\minimummorphismwidth{7mm}
\tikzset{wedge=false}
\chapter{Background}
\section{Unitary Error Bases}
I will start by formally defining a unitary error basis.
\begin{definition}[Unitary Error Basis]
A unitary error basis on a $d$ dimensional Hilbert space is a family of $d$ unitary matrices $U_i$, each of size $d \times d$, such that:
\begin{equation}
 \tr (U^{\dag}_i \circ U_{i'})=\delta_{ii'}d 
\end{equation}
\end{definition}
As mentioned in the introduction there is a notion of equivalence between UEBs. Let $A=A_i \, ,\, \, 0 \leq i< d$ and $B=B_i \, ,\, \, 0 \leq i<  d$ be unitary error bases. Then $A$ is equivalent to $B$ if $\exists$ unitary matrices $U$ and $V$ and $c_i \in \mathbb{C}$ such that $A_i=c_iU  B_i V $ for all $0 \leq i<  d $ \cite{klapp}.
\section{The Pauli Matrices} 
 The most famous example of a unitary error basis is the Pauli matrices:

\begin{equation}
 \left \{ \begin{pmatrix}
    1 & \, \,0\\
    0 & \, 1
  \end{pmatrix},
  \begin{pmatrix}
    1 & \,0\\
    0 & -1
  \end{pmatrix},
  \begin{pmatrix}
    0 & e^{i\theta}\\
    1 & 0
  \end{pmatrix},
  \begin{pmatrix}
    0 & -e^{i\theta}\\
    1 & 0
  \end{pmatrix}
 \right \}\end{equation} 

where $d=2$. In fact, up to equivalence, the Pauli matrices are the only unitary error basis  for $d=2$. 

The following is a proof is due to Klappenecker and R\"otteler \cite{klapp}. However their proof is extremely terse being only eight lines long and having several sizeable gaps. I have worked out how to fill in the gaps myself, and here present a fully explicit proof.\begin{lemma}
All unitary error bases with $d=2$ are equivalent to the Pauli matrices. \end{lemma}

\begin{proof}
 
Let 
\begin{equation}
E
\quad:=\quad
\{ A_1,A_2,A_3,A_4 \}
\end{equation}
be a unitary error basis. If we now compose each of our unitaries by $A_1^{\dag}$ on the left then we obtain the equivalent UEB, $\{ \mathbb{I}_2,A'_2,A'_3,A'_4 \}$ for some unitaries $A'_2,A'_3,A'_4$.

Now $A'_2$ is a unitary matrix and thus a normal matrix. By the spectral theorem $\exists$ a unitary matrix $P$ $s.t$ $P^{\dag}A'_2 P = D$ for some $D$ that is diagonal. \\
So now take our UEB and conjugate with $P$ to obtain $\{ \mathbb{I}_2,D,B_1,B_2 \}$ with $D$ diagonal. 
\begin{equation*}
\text{Let } D= \begin{pmatrix}
    a & 0\\
    0 & b
  \end{pmatrix} \text{for some } a,b \in \C \end{equation*}
Then  
$
 \tr(D^{\dag}D)
=
2 
\Rightarrow
|a|^2+|b|^2
=
2
$
and
$
 \tr(D \mathbb{I}_2)
=
0 
\Rightarrow
a+b
=
0
$.
Thus $| a|^2+|a|^2=2$. So $a= e^{i\theta}$ and $b=-e^{i\theta}$ for some $\theta \in \mathbb{R}$. Since we can multiply our matrix by a complex scalar we can assume: 
\begin{equation*}D= \begin{pmatrix}
    1 & 0\\
    0 & -1
  \end{pmatrix}\end{equation*}
Let the diagonal elements of $B_1$ be $b_1$ and $b_2$.
Now $\tr(D^{\dag}B_1 )=0=b_1-b_2\Rightarrow b_1=b_2$ \\
and $\tr(\mathbb{I}_2^\dag B_1)=b_1+b_2=2b_1=2b_2=0$.\\
$\Rightarrow b_1=b_2=0.$ Similarly the diagonal elements of $B_2$ are zero.

Since we can multiply by any complex scalar we can assume that: 
\begin{equation*}
B_1=
\begin{pmatrix}
    0 & c_1\\
    1 & 0
\end{pmatrix} 
\text{and } B_2=
\begin{pmatrix}
    0 & c_2\\
    1 & 0
\end{pmatrix} 
\text{for } c_1,c_2 \in \mathbb{C}
\end{equation*}
$\tr (B_1 B_1^\dag)=|c_1|+1=2 \Rightarrow |c_1|=1.$  Similarly, $|c_2|=1$ \\
$\tr (B_1^\dag B_2)=\bar c_1 c_2+1=0 \Rightarrow \bar c_1 c_2=-1  \Rightarrow \bar c_1 c_1 c_2 =-c_1\Rightarrow c_2=-c_1$

So since $c_1$ is a phase we can write $c_1$ as $e^{i\phi}$ for some $\phi \in \mathbb{R} $ .\\
Now we have:
\begin{equation}
E \equiv E'=\left \{ \begin{pmatrix}
    1 & \, \,0\\
    0 & \, 1
  \end{pmatrix},
  \begin{pmatrix}
    1 & \,0\\
    0 & -1
  \end{pmatrix},
  \begin{pmatrix}
    0 & e^{i\phi}\\
    1 & 0
  \end{pmatrix},
  \begin{pmatrix}
    0 & -e^{i\phi}\\
    1 & 0
  \end{pmatrix}
 \right \}\end{equation} \\
 \newpage
Now let 
\begin{equation*}
T=
\begin{pmatrix}
    1 & \,0\\
    0 & e^{i \phi /2}
\end{pmatrix}
\end{equation*} 
We now compose all our basis matrices on the left by $T$ and on the right by $T^\dag$. The first two are both diagonal and thus commute with $T$, thus by unitarity of $T$ they remain unchanged.

The third becomes:

\begin{equation*}
\begin{pmatrix}
    0 & e^{i\phi/2}\\
    e^{i\phi/2} & 0
\end{pmatrix}
\end{equation*}
 and if we now multiply by the phase $e^{-i\phi/2}$ we obtain 
\begin{equation*}
\begin{pmatrix}
    0 & \, \, 1\\
    1 & \, \, 0
 \end{pmatrix}
 \end{equation*}

  The fourth matrix becomes:
 \begin{equation*}
 \begin{pmatrix}
    0 & -e^{i\theta/2}\\
    e^{i\theta/2} & 0
  \end{pmatrix}
  \end{equation*}
   and if we now multiply by the phase $e^{i(\pi/2 -\theta/2)}$ we have
  \begin{equation*}
  \begin{pmatrix}
    0 & -i\\
    i & 0
  \end{pmatrix}
  \end{equation*}

  Thus \begin{equation}
  E \equiv \left\{ \begin{pmatrix}
    1 &\, \, 0\\
    0 & \, 1
  \end{pmatrix}, \begin{pmatrix}
    1 & \, 0\\
    0 & -1
  \end{pmatrix},\begin{pmatrix}
    0 &\, \, 1\\
    1 & \, \,0
  \end{pmatrix},\begin{pmatrix}
    0 & -i\\
    i & \, 0
  \end{pmatrix} 
\right \}\end{equation}

which is the unitary error basis made up of the Pauli matrices as required.

\end{proof}

\section{Categorical Quantum Mechanics}
I will assume basic knowledge of category theory. 
\subsection{The Graphical Calculus and Symmetric Monoidal Categories}
Many of the diagrams and definitions in this section are from the book `An Introduction to Categorical Quantum Mechanics' ~\cite{cqm2014}.
Given a category, $\cat{C}$ the identity morphism, $\id{}$ of an object $A \in Ob(\cat{C})$ is represented by a wire:\\
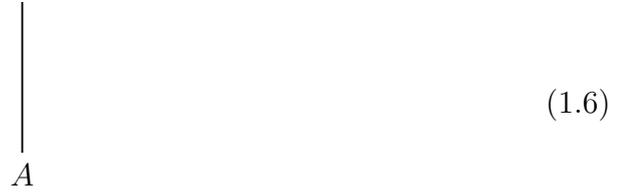
\begin{equation}
\begin{aligned}
\begin{tikzpicture}
    \draw[string] (0,0) node [below] {$A$} to
    node [auto,swap] {\phantom{$A$}}
    (0,2);
\end{tikzpicture}
\end{aligned}
\end{equation}
A morphism $f:A \rightarrow B$ is represented as follows:\\
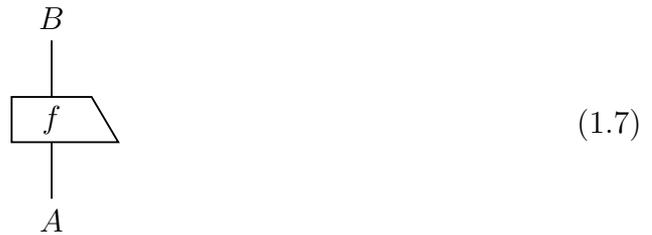
\begin{equation}
\begin{aligned}
\begin{tikzpicture}
  \begin{pgfonlayer}{background}
    \node (f) [morphism,wedge] at (0,0) {$f$};
  \end{pgfonlayer}
  \begin{pgfonlayer}{foreground}
    \draw[string] (f.north) to
        +(0,0.75) node [above] {$B$};
    \draw[string] (f.south) to
        +(0,-0.75) node [below] {$A$};
  \end{pgfonlayer}
\end{tikzpicture}
\end{aligned}
\end{equation}
Note that these diagrams are read bottom to top, this is not essential and is sometimes reversed in certain papers. I will stick to the bottom to top orientation in this dissertation.
Composition of morphisms is represented as connecting diagrams in series as follows. Given $f$ as above and \\
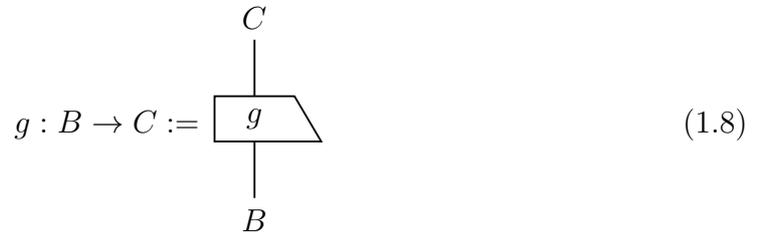
\begin{equation}
g:B \rightarrow C:=
\begin{aligned}
\begin{tikzpicture}
  \begin{pgfonlayer}{background}
    \node (f) [morphism,wedge] at (0,0) {$g$};
  \end{pgfonlayer}
  \begin{pgfonlayer}{foreground}
    \draw[string] (f.north) to
        +(0,0.75) node [above] {$C$};
    \draw[string] (f.south) to
        +(0,-0.75) node [below] {$B$};
  \end{pgfonlayer}
\end{tikzpicture}
\end{aligned}
\end{equation}
then
\begin{equation}
g \circ f:=
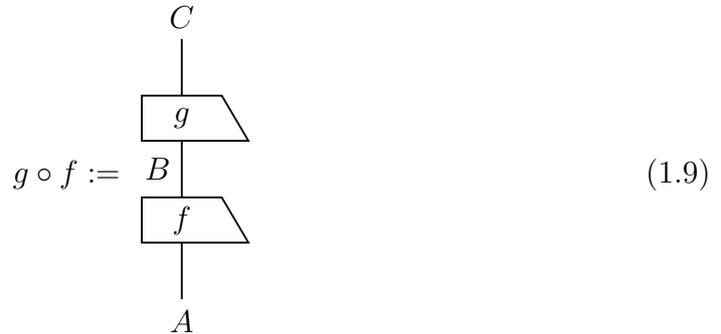
\begin{aligned}
\begin{tikzpicture}[yscale=1]
  \begin{pgfonlayer}{background}
\node (f) [morphism,wedge] at (0,0) {$f$};
\node (g) [morphism,wedge, anchor=south] at ([yshift=0.75cm] f.north) {$g$};
\end{pgfonlayer}
  \begin{pgfonlayer}{foreground}
\draw[string] (g.north) to
    +(0,0.75) node [above] {$C$};
\draw[string] (f.south) to
    +(0,-0.75) node [below] {$A$};
\draw[string] (f.north) to node [auto] {$B$} (g.south) {};
\end{pgfonlayer}
\end{tikzpicture}
\end{aligned}
\end{equation}
\begin{definition}[Monoidal Category]A monoidal category is a category equipped with the following additional structure:
\begin{itemize}
\item a monoidal product, defined to be a bifunctor $\otimes \colon \cat C \times \cat C \to \cat C$;
\item a unit object, defined to be a special object $I \in \Ob(\cat C )$ which is both a left and right unit for the monoidal product via natural isomorphisms $\lambda$ and $\rho$ with components, $\lambda_A:I \otimes A \rightarrow A$ and $\rho_A:A \otimes I \rightarrow A$; 
\item the monoidal product is associative via the natural isomorphism $\alpha$, with components $\alpha_{A,B,C}:(A \otimes B) \otimes C
  \rightarrow A \otimes (B \otimes C)$; 
\end{itemize}
There is also a coherence property that ensures that equations built from $\circ$, $\otimes$,
$\id{}$, $\alpha$, $\alpha ^{-1}$, $\lambda$, $\lambda ^{-1}$, $\rho$ and $\rho ^{-1}$ are well-defined~\cite{cqm2014}.\\
\textit{Coherence property}\\
The following diagrams commute~\cite{cqm2014}:\\
\begin{equation}
\label{eq:triangle}
\hspace{3mm}\begin{aligned}
\begin{tikzpicture}[xscale=6,yscale=2]
\node (A) at (0,0) {$(A \otimes I) \otimes B$};
\node (B) at (1,0) {$A \otimes (I \otimes B)$};
\node (C) at (0.5,-1) {$A \otimes B$};
\draw [->] (A) to node [auto,swap, pos=0.4] {$\rho_A \otimes \id[B]$} (C);
\draw [->] (B) to node [auto, pos=0.4] {$\id[A] \otimes \lambda_B$} (C);
\draw [->] (A) to node [above] {$\alpha_{A,I,B}$} (B);
\end{tikzpicture}
\end{aligned}
\end{equation}
\begin{equation}
\label{eq:pentagon}
\hspace{6.4mm}\begin{aligned}
\begin{tikzpicture}[xscale=2.3,yscale=2]
\node (A) at (0,0)
    {$\big( (A \otimes B) \otimes C\big) \otimes D$};
\node (B) at (0.8,1)
    {$\big( A \otimes (B \otimes C) \big) \otimes D$};
\node (C) at (3.2,1)
    {$A \otimes \big( ( B \otimes C) \otimes D\big)$};
\node (D) at (4,0)
    {$A \otimes \big( B \otimes ( C \otimes D ) \big)$};
\node (E) at (2,-1)
    {$(A \otimes B) \otimes (C \otimes D)$};
\draw [->] (A)
    to node [auto, pos=0.3]
    {$\alpha _{A,B,C} \otimes \id[D]$} (B);
\draw [->] (B)
    to node [auto]
    {$\alpha_{A,B \otimes C,D}$} (C);
\draw [->] (C)
    to node [auto, pos=0.7]
    {$\id[A] \otimes \alpha_{B,C,D}$} (D);
\draw [->] (A)
    to node [auto,swap]
    {$\alpha_{A \otimes B, C, D}$} (E);
\draw [->] (E)
    to node [auto,swap]
    {$\alpha_{A,B,C\otimes D}$} (D);
\end{tikzpicture}
\end{aligned}
\end{equation}

\end{definition}

\begin{definition}[Symmetric Monoidal Category]A symmetric monoidal category, $\cat{C}$, is a monoidal category with a natural isomorphism $\sigma$ whose components are \mbox{$\sigma_{A,B}: A \otimes B \rightarrow B \otimes A$} and $\sigma_{B,A}: B \otimes A \rightarrow A \otimes B$ which are mutual inverses $\forall A,B \in \ob (\cat{C})$ The components of $\sigma$ are often called swap maps~\cite{cqm2014}.
\end{definition}
Graphically, the tensor product is represented as parallel composition. For $f$ and $g$ as above:
\begin{equation}
\begin{aligned}
\begin{tikzpicture}
\begin{pgfonlayer}{background}
\node (f) [morphism,wedge] at (0,0) {$f$};
\node (g) [morphism,wedge] at (1.5,0) {$g$};
\end{pgfonlayer}
\begin{pgfonlayer}{foreground}
\draw[string] (f.north) to node [auto] {$B$} +(0,0.75);
\draw[string] (f.south) to [out=down, in=up] node [auto, swap] {$A$} +(0,-0.75);
\draw[string] (g.north) to node [auto] {$D$} +(0,0.75);
\draw[string] (g.south) to [out=down, in=up] node [auto, swap] {$C$} +(0,-0.75);
\end{pgfonlayer}
\end{tikzpicture}
\end{aligned}
\end{equation}
The unit object is not drawn by convention. Morphisms $I \rightarrow A$ are called states and are represented:
\begin{equation}
\begin{aligned}
\begin{tikzpicture}
\node[state] (0) at (0,0) {$a$};
\node (1) at (0,1) [above] {$A$};
\draw[string] (0) to (1);
\end{tikzpicture}
\end{aligned}
\end{equation}
Morphisms $A \rightarrow I$ are called effects and are represented:
\begin{equation}
\begin{aligned}
\begin{tikzpicture}[yscale=-1]
\node[state,hflip] (0) at (0,0) {$b$};
\node (1) at (0,1.5) [above] {$A$};
\draw[string] (0) to (1);
\end{tikzpicture}
\end{aligned}
\end{equation}
I will be working almost exclusively within the symmetric monoidal category \cat{FHilb} in this dissertation. Objects are finite dimensional Hilbert spaces, morphisms are linear maps, the monoidal product is the standard linear algebraic tensor product, and the unit object is the the one-dimensional Hilbert space, $\mathbb{C}$. States correspond to column vectors, and effects correspond to row vectors. 

Transposition is represented by vertical reflection, conjugation by rotation by $\pi$ radians and adjoints by horizontal reflection. So the adjoint of a linear map is its conjugate transpose as expected~\cite{cqm2014}. For example:
\begin{equation}
\left(
\begin{aligned}
\begin{tikzpicture}
  \begin{pgfonlayer}{background}
    \node (f) [morphism,wedge] at (0,0) {$f$};
  \end{pgfonlayer}
  \begin{pgfonlayer}{foreground}
    \draw[string] (f.north) to
        +(0,0.75) node [above] {$B$};
    \draw[string] (f.south) to
        +(0,-0.75) node [below] {$A$};
  \end{pgfonlayer}
\end{tikzpicture}
\end{aligned}
\right)^{{\Huge \dag}}
\quad=\quad
\begin{aligned}
\begin{tikzpicture}
  \begin{pgfonlayer}{background}
    \node (f) [morphism,hflip,wedge] at (0,0) {$f$};
  \end{pgfonlayer}
  \begin{pgfonlayer}{foreground}
    \draw[string] (f.north) to
        +(0,0.75) node [above] {$A$};
    \draw[string] (f.south) to
        +(0,-0.75) node [below] {$B$};
  \end{pgfonlayer}
\end{tikzpicture}
\end{aligned}
\end{equation}
Hence the asymmetry of the morphisms.
\subsection{Classical Structures}
\begin{definition}[Monoid]
Given a linear map \tinymult[blackdot] and a state \tinyunit \,, where all wires are the same object $H$, $(H,\tinymult,\tinyunit)$ is a monoid if the following equations are satisfied~\cite{cqm2014}:
\begin{equation}
\begin{aligned}\begin{tikzpicture}
          \node (0a) at (-1,0.25) {};
          \node (0b) at (-0.5,0.25) {};
          \node (0c) at (0.5,0.25) {};

          \node[dot] (1) at (0,1) {};
          \node[dot] (2) at (-0.5,1.5) {};

          \node (5a) at (-0.5,2) {};

          \draw[string, out=90, in =180] (0a) to (2);
          \draw[string, out=0, in=90] (2) to (1);
          \draw[string, out=0, in=90] (1) to (0c);
          \draw[string, out=180, in=90] (1) to (0b);

          \draw[string] (2) to (5a.center);
         
          \end{tikzpicture}\end{aligned}   
  \quad   = \quad
\begin{aligned}\begin{tikzpicture}
          \node (0a) at (0.75,0.25) {};
          \node (0b) at (0.25,0.25) {};
          \node (0c) at (-0.75,0.25) {};

          \node[dot] (1) at (-0.25,1) {};
          \node[dot] (2) at (0.25,1.5) {};

          \node (5a) at (0.25,2) {};

          \draw[string, out=90, in =00] (0a) to (2);
          \draw[string, out=180, in=90] (2) to (1);
          \draw[string, out=180, in=90] (1) to (0c);
          \draw[string, out=0, in=90] (1) to (0b);

          \draw[string] (2) to (5a.center);
\end{tikzpicture}\end{aligned} 
\end{equation}
and
\begin{equation}
\begin{aligned}\begin{tikzpicture}
          
          \node (0a) at (-1,0) {};
          \node (0b) at (-0.5,0) {};
          \node[dot] (0c) at (0,0) {};
          \node[dot] (7) at (-0.5,-0.5) {};       
          \node (10a) at (-0.5,-1) {};       
          \draw[string, out=270, in =180] (0a.center) to (7);
          \draw[string, out=0, in=270] (7) to (0c.center);
          \draw[string] (7) to (10a.center);
         
          \end{tikzpicture}\end{aligned}   
  \quad   = \quad
\begin{aligned}\begin{tikzpicture}
           \node (10a) at (-0.5,0) {}; 
           \node (5a) at (-0.5,1) {}; 
           \draw[string] (5a.center) to (10a.center);
\end{tikzpicture}\end{aligned}
\quad=\quad
\begin{aligned}\begin{tikzpicture}

          \node[dot] (0a) at (-1,0) {};
          \node (0b) at (-0.5,0) {};
          \node (0c) at (0,0) {};
          \node[dot] (7) at (-0.5,-0.5) {};       
          \node (10a) at (-0.5,-1) {};       
          \draw[string, out=270, in =180] (0a.center) to (7);
          \draw[string, out=0, in=270] (7) to (0c.center);
          \draw[string] (7) to (10a.center);

\end{tikzpicture}\end{aligned}  
  \end{equation}
\end{definition}
\begin{definition}[Comonoid]
Given a linear map \tinycomult[blackdot] and an effect \tinycounit \,, where all wires are the same object $H$, $(H,\tinycomult,\tinycounit)$ is a comonoid if the following equations are satisfied~\cite{cqm2014}:
\begin{equation}
\begin{aligned}\begin{tikzpicture}[yscale=-1]
          \node (0a) at (-1,0.25) {};
          \node (0b) at (-0.5,0.25) {};
          \node (0c) at (0.5,0.25) {};

          \node[dot] (1) at (0,1) {};
          \node[dot] (2) at (-0.5,1.5) {};

          \node (5a) at (-0.5,2) {};

          \draw[string, out=90, in =180] (0a.center) to (2);
          \draw[string, out=0, in=90] (2) to (1);
          \draw[string, out=0, in=90] (1) to (0c.center);
          \draw[string, out=180, in=90] (1) to (0b.center);

          \draw[string] (2) to (5a.center);
         
          \end{tikzpicture}\end{aligned}   
  \quad   = \quad
\begin{aligned}\begin{tikzpicture}[yscale=-1]
          \node (0a) at (0.75,0.25) {};
          \node (0b) at (0.25,0.25) {};
          \node (0c) at (-0.75,0.25) {};

          \node[dot] (1) at (-0.25,1) {};
          \node[dot] (2) at (0.25,1.5) {};

          \node (5a) at (0.25,2) {};

          \draw[string, out=90, in =00] (0a.center) to (2);
          \draw[string, out=180, in=90] (2) to (1);
          \draw[string, out=180, in=90] (1) to (0c.center);
          \draw[string, out=0, in=90] (1) to (0b.center);

          \draw[string] (2) to (5a.center);
\end{tikzpicture}\end{aligned}  
  \end{equation}
and
\begin{equation}
\begin{aligned}\begin{tikzpicture}
          
          \node (0a) at (-1,-0.5) {};
          \node (0b) at (-0.5,-0.5) {};
          \node[dot] (0c) at (0,-0.5) {};
          \node[dot] (7) at (-0.5,0) {};       
          \node (10a) at (-0.5,0.5) {};       
          \draw[string, out=90, in =180] (0a.center) to (7);
          \draw[string, out=0, in=90] (7) to (0c.center);
          \draw[string] (7) to (10a.center);
         
          \end{tikzpicture}\end{aligned}   
  \quad   = \quad
\begin{aligned}\begin{tikzpicture}
           \node (10a) at (-0.5,0) {}; 
           \node (5a) at (-0.5,1) {}; 
           \draw[string] (5a.center) to (10a.center);
\end{tikzpicture}\end{aligned}
\quad=\quad
\begin{aligned}\begin{tikzpicture}

      \node[dot] (0a) at (-1,-0.5) {};
          \node (0b) at (-0.5,-0.5) {};
          \node (0c) at (0,-0.5) {};
          \node[dot] (7) at (-0.5,0) {};       
          \node (10a) at (-0.5,0.5) {};       
          \draw[string, out=90, in =180] (0a.center) to (7);
          \draw[string, out=0, in=90] (7) to (0c.center);
          \draw[string] (7) to (10a.center);

\end{tikzpicture}\end{aligned}  
  \end{equation}
  \end{definition}
  \begin{definition}[Comonoid Homomorphism]
  A comonoid homomorphism\\ $(H,\tinycomult[blackdot],\tinycounit[blackdot])$ $\rightarrow (H',\tinycomult[whitedot],\tinycounit[whitedot])$ is a morphism $f:H \rightarrow H'$ such that~\cite{cqm2014}:
\begin{align}
\begin{aligned}
\begin{tikzpicture}[string]
\node (f1) [morphism, wedge, connect n, connect s, width=0cm] at (1,-2) {$f$};
\node (m) [whitedot] at (f1.connect n) {};
\draw (m.center) to [out=right, in=down] +(0.75,0.75);
\draw (m.center) to [out=left, in=down] +(-0.75,0.75);
\end{tikzpicture}
\end{aligned}
\quad&=\quad
\begin{aligned}
\begin{tikzpicture}[string]
\node (f1) [morphism, wedge, connect n, width=0cm] at (0.25,-0.2) {$f$};
\node (f2) [morphism, wedge, connect n, width=0cm] at (1.75,-0.2) {$f$};
\node (m) [blackdot] at (1,-1) {};
\draw (m.center) to [out=left, in=down] (f1.south);
\draw (m.center) to [out=right, in=down] (f2.south);
\draw (m.center) to +(0,-0.75);
\end{tikzpicture}
\end{aligned}
\\
\begin{aligned}
\begin{tikzpicture}[string]
\draw [white] (-0.75,0) to (0.75,0);
\node (f1) [morphism, wedge, connect n, width=0cm, connect s] at (0,-0.2) {$f$};
\node (e) [whitedot] at (f1.connect n) {};
\draw [white] (-0.5,1.25) to (0.5,1.25);
\end{tikzpicture}
\end{aligned}
\quad&=\quad
\begin{aligned}
\begin{tikzpicture}[string]
\node (e) [blackdot] at (0,0) {};
\draw (e.center) to +(0,-1);
\draw [white] (-0.3,1.25) to (0.3,1.25);
\end{tikzpicture}
\end{aligned}
\end{align}
  \end{definition}
  \begin{definition}[Classical Structure]
  Given a monoid $(H,\tinymult[blackdot],\tinyunit[blackdot])$ and comonoid $(H,\tinycomult[blackdot],\tinycounit[blackdot])$ on the same object $H$, they form a classical structure, $(H,\tinymult[blackdot],\tinyunit[blackdot],\tinycomult[blackdot],\tinycounit[blackdot])$  if the following holds~\cite{cqm2014}:
 
\hspace{55mm} \textit{Frobenius law}\begin{equation}
    \begin{aligned}\begin{tikzpicture}[yscale=0.75]
          \node (0) at (0,0) {};
          \node (0a) at (0,1) {};
          \node[blackdot] (1) at (0.5,2) {};
          \node[blackdot] (2) at (1.5,1) {};
          \node (3) at (1.5,0) {};
          \node (4) at (2,3) {};
          \node (4a) at (2,2) {};
          \node (5) at (0.5,3) {};
          \draw[string] (0) to (0a.center);
          \draw[string,out=90,in=180] (0a.center) to (1.center);
          \draw[string,out=0,in=180] (1.center) to (2.center);
          \draw[string,out=0,in=270] (2.center) to (4a.center);
          \draw[string] (4a.center) to (4);
          \draw[string] (2.center) to (3);
          \draw[string] (1.center) to (5);
      \end{tikzpicture}\end{aligned}
    \quad = \quad
    \begin{aligned}\begin{tikzpicture}[yscale=0.75]
          \node (0a) at (-0.5,0) {};
          \node (0b) at (0.5,0) {};
          \node[blackdot] (1) at (0,1) {};
          \node[blackdot] (2) at (0,2) {};
          \node (3a) at (-0.5,3) {};
          \node (3b) at (0.5,3) {};
          \draw[string,out=90,in=180] (0a) to (1.center);
          \draw[string,out=90,in=0] (0b) to (1.center);
          \draw[string] (1.center) to (2.center);
          \draw[string,out=180,in=270] (2.center) to (3a);
          \draw[string,out=0,in=270] (2.center) to (3b);
      \end{tikzpicture}\end{aligned}
    \quad = \quad
    \begin{aligned}\begin{tikzpicture}[yscale=0.75,xscale=-1]
          \node (0) at (0,0) {}; 
          \node (0a) at (0,1) {};
          \node[blackdot] (1) at (0.5,2) {};
          \node[blackdot] (2) at (1.5,1) {};
          \node (3) at (1.5,0) {};
          \node (4) at (2,3) {};
          \node (4a) at (2,2) {};
          \node (5) at (0.5,3) {};
          \draw[string] (0) to (0a.center);
          \draw[string,out=90,in=180] (0a.center) to (1.center);
          \draw[string,out=0,in=180] (1.center) to (2);
          \draw[string,out=0,in=270] (2.center) to (4a.center);
          \draw[string] (4a.center) to (4);
          \draw[string] (2.center) to (3);
          \draw[string] (1.center) to (5);
      \end{tikzpicture}\end{aligned}
  \end{equation} 
  This is the Frobenius law and makes $(H,\tinymult[blackdot],\tinyunit[blackdot],\tinycomult[blackdot],\tinycounit[blackdot])$ into a Frobenius algebra. 

\hspace{60mm} \textit{specialness}
\begin{equation}
\begin{aligned}\begin{tikzpicture}
          \node (0a) at (-1,0) {};
          \node (0b) at (-0.5,0) {};
          \node (0c) at (0,0) {};

          \node[dot] (2) at (-0.5,0.5) {};
          \node (5a) at (-0.5,1) {};
 
          \node[dot] (7) at (-0.5,-0.5) {};       
          \node (10a) at (-0.5,-1) {};

          \draw[string, out=90, in =180] (0a.center) to (2);
          \draw[string, out=0, in=90] (2) to (0c.center);

          \draw[string] (2) to (5a.center);
          
          \draw[string, out=270, in =180] (0a.center) to (7);
          \draw[string, out=0, in=270] (7) to (0c.center);

          \draw[string] (7) to (10a.center);
         
          \end{tikzpicture}\end{aligned}   
  \quad   = \quad
\begin{aligned}\begin{tikzpicture}
           \node (10a) at (-0.5,-1) {}; 
           \node (5a) at (-0.5,1) {}; 
           \draw[string] (5a.center) to (10a.center);
\end{tikzpicture}\end{aligned}  
  \end{equation}
 \hspace{63mm} \textit{commutativity}
\begin{equation}
\begin{aligned}\begin{tikzpicture}
          \node (0a) at (-1,0.25) {};
          \node (0b) at (-0.5,0.25) {};
          \node (0c) at (0,0.25) {};

          \node[dot] (2) at (-0.5,1.5) {};

          \node (5a) at (-0.5,2) {};

          \draw[string, out=90, in =180] (0a) to (2);
          \draw[string, out=0, in=90] (2) to (0c);

          \draw[string] (2) to (5a.center);
         
          \end{tikzpicture}\end{aligned}   
  \quad   = \quad
\begin{aligned}\begin{tikzpicture}
           \node (0a) at (-1,0.25) {};
          \node (0b) at (-0.5,0.25) {};
          \node (0c) at (0,0.25) {};

          \node[dot] (2) at (-0.5,1.5) {};

          \node (5a) at (-0.5,2) {};

          \draw[string, out=90, in =0] (0a) to (2);
          \draw[string, out=180, in=90] (2) to (0c);
          \draw[string] (2) to (5a.center);
\end{tikzpicture}\end{aligned}  
  \end{equation}  
and $(\tinymult[blackdot])^{\dag}=\tinycomult[blackdot]$  
\end{definition}

In $\cat{FHilb}$ every classical structure takes the following form. Given $H \in \ob (\cat{FHilb})$ with orthonormal basis $\ket{i}\,, \, \,0 \leq i < d$ and $c \in \mathbb{C}$~\cite{cqm2014}: \begin{equation}\label{cls}
\tinycomult[blackdot](\ket{i})=\ket{i} \otimes \ket{i},\, \, \,
\tinymult[blackdot](\ket{i} \otimes \ket{j})=\delta_{ij} \ket{i}, \, \, \,
\tinycounit[blackdot](\ket{i})=1 \text{ and }
\tinyunit[blackdot](c)=c \sum^{d-1}_{i=0} \ket{i}
\end{equation}
Given a classical structure on a Hilbert space $H$ the following theorem gives us the concept of a spider:
\begin{theorem}
 Any connected morphism from $H^{\otimes m}$ to $H^{\otimes n}$ built using the multiplication, comultiplication, unit and counit maps of a classical structure, as well as the identity on $H$ and the switch map, $\sigma_{H,H}$, can be re-written in the following normal form~\cite{cqm2014}:

\begin{equation}\label{eq:normalform}
    \overbrace{\underbrace{
    \begin{pic}[xscale=1.33, yscale=0.5, string]
      \node (0) at (0,0) {};
      \node (1) at (2,0) {};
      \node (2) at (3,0) {};
      \node[dot] (3) at (2.5, 1) {};
      \node (4) at (2.25, 1.75) {};
      \node (6) at (1.25, 2.75) {};
      \node[dot] (7) at (1,3) {};
      \node[dot] (8) at (1,4) {};
      \node (9) at (1.25,4.25) {};
      \node (11) at (1.75,4.75) {};
      \node[dot] (12) at (2,5) {};
      \node[dot] (13) at (2.5,6) {};
      \node (14) at (0,7) {};
      \node (15) at (1.5,7) {};
      \node (16) at (2,7) {};
      \node (17) at (3,7) {};
      \edges[string]
      \draw[out=90, in=180, looseness=0.75] (0.center) to (7.center);
      \draw[out=90, in=180] (1.center) to (3.center);
      \draw[out=90, in=0] (2.center) to (3.center);
      \draw (7) to (8);
      \draw[out=270, in=0] (17) to (13);
      \draw[out=270, in=180] (16) to (13);
      \draw[out=270, in=180, looseness=0.66] (15) to (12);
      \draw[out=270, in=0] (13) to (12);
      \draw[out=270, in=180, looseness=0.75] (14) to (8);
      \draw[dashed, out=0, in=270] (8) to (12);
      \begin{scope}
      \clip (12) circle (7mm);
      \draw[string, in=0, out=270] (12) to (8);
      \end{scope}
      \begin{scope}
      \clip (8) circle (3mm);
      \draw[string, out=0, in=270] (8) to (12);
      \end{scope}
      \draw[out=0, in=90, dashed] (7) to (3);
      \begin{scope}
      \clip (7) circle (3mm);
      \draw[string, out=0, in=90] (7) to (3);
      \end{scope}
      \begin{scope}
      \clip (3) circle (7mm);
      \draw[string, out=0, in=90] (7) to (3);
      \end{scope}
    \end{pic}}_{m}}^{n}
  \end{equation}
\end{theorem}

Thus any two connected diagrams of black dots with the same numbers of inputs and outputs are equal and the following is well-defined:

\begin{definition}[Spider]
 Any connected morphism from $H^{\otimes m}$ to $H^{\otimes n}$ built using the multiplication, comultiplication, unit and counit maps of a classical structure, as well as the identity and the switch map, $\sigma_{H,H}$ is called a spider and is denoted~\cite{qcs}:
\begin{equation}
    \overbrace{\underbrace{
    \begin{pic}
     \node[blackdot,scale=2] (A){};
     \draw[string,out=190,in=90] (A) to (-2,-1);
     \draw[string,out=350,in=90] (A) to (2,-1);
     \draw[string,out=190,in=90] (A) to (-1.7,-1);
     \draw[string,out=350,in=90] (A) to (1.7,-1);  
     \draw[string,out=170,in=270] (A) to (-2,1);
     \draw[string,out=10,in=270] (A) to (2,1);
     \draw[string,out=170,in=270] (A) to (-1.7,1);
     \draw[string,out=10,in=270] (A) to (1.7,1); 
     \draw[loosely dotted] (-1.65,0.9) to (1.65,0.9); 
     \draw[loosely dotted] (-1.65,-0.9) to (1.65,-0.9);      
    \end{pic}}_{m}}^{n}
  \end{equation}
\end{definition}
By the above it is clear that any connected spiders can be merged into one as long as the total number of inputs and outputs are preserved.

Given a classical structure there are two spiders which play a special role, the one with no inputs and two outputs and the one with two inputs and no outputs:
\begin{equation}
 \begin{aligned}\begin{tikzpicture}[xscale=1.35]        
        \node[blackdot] (1) at (0,1) {};
        \node (2) at (-0.5,2) {$H$};
        \node (3) at (0.5,2) {$H$};
        \draw[string,out=180,in=down] (1.center) to (2.south);
        \draw[string,out=0,in=down] (1.center) to (3.south);
    \end{tikzpicture}\end{aligned}
    ,
 \begin{aligned}\begin{tikzpicture}[xscale=1.35,yscale=-1]        
        \node[blackdot] (1) at (0,1) {};
        \node (2) at (-0.5,2) {$H$};
        \node (3) at (0.5,2) {$H$};
        \draw[string,out=180,in=down] (1.center) to (2.north);
        \draw[string,out=0,in=down] (1.center) to (3.north);
    \end{tikzpicture}\end{aligned}   
\end{equation}
These are sometimes referred to as cups and caps for obvious reasons. These morphisms enact a self-duality of $H$, which means that the following equations hold~\cite{cqm2014}:\\
\textit{(SN) Snake equation}
\begin{equation}
 \begin{aligned}\begin{tikzpicture}

          \node (2a) at (1.25,0.5) {};
          \node (2b)[blackdot] at (0.75,1.75) {};
          \node (3a)[blackdot] at (0,0.75){};
          \node (3b) at (-0.5,2){};

          \draw[string,out=180,in=0] (2b) to (3a);
          \draw[string,out=180,in=270] (3a) to (3b);

          \draw[string,out=90,in=00] (2a) to (2b);
      \end{tikzpicture}\end{aligned} 
      \quad= \quad
      \begin{aligned}\begin{tikzpicture}
      \node (A) at (0,0) {};
      \node(B) at (0,3) {};
      \draw[string] (A) to (B);
      \end{tikzpicture}\end{aligned}  
  \quad=\quad
    \begin{aligned}\begin{tikzpicture}

          \node (2a) at (1.75,1.5) {};
          \node (2d)[blackdot] at (1.25,0.25) {};
          \node (3d)[blackdot] at (0.5,1.25){};
          \node (3e) at (0,0){};
          
          \draw[string,out=180,in=0] (2d) to (3d);
          \draw[string,out=180,in=90] (3d) to (3e);
         \draw[string,out=270,in=00] (2a) to (2d);
          
      \end{tikzpicture}\end{aligned}    
      \end{equation}  
The trace of a linear operator $f$ is given graphically by~\cite{cqm2014}:
\begin{equation}\label{eq:trace}
     \tr (f):= \begin{pic}
  \begin{pgfonlayer}{background}
    \node (f) [morphism,wedge] at (0,0) {$f$};
  \end{pgfonlayer}
  \begin{pgfonlayer}{foreground}
    \draw[string] (f.north) to
        +(0,0.75) node [above] {};
    \draw[string] (f.south) to
        +(0,-0.75) node [below] {};
  \end{pgfonlayer}
  \node[blackdot] (A) at (1, 1.75){};
  \node[blackdot] (B) at (1, -1.75){};
  \draw[string] (2,-1) to (2,1);
  \draw[string,out=270,in=180] (0,-1) to (B);
  \draw[string,out=270,in=0] (2,-1) to (B);
  \draw[string,out=90,in=180] (0,1) to (A);
  \draw[string,out=90,in=0] (2,1) to (A);
  \end{pic}
    \end{equation}
Where in this case the cup and cap map are given by a classical structure. A swap map is necessary above the cup here in general as the duality may not be of this form. However, since classical structures are commutative I have elided it.

The comultiplication map \tinycomult[blackdot], of our classical structure copies the basis states of an orthonormal basis. Given a classical structure, the states that the comultiplication map copies always form an orthonormal basis. Given an orthonormal basis a classical structure can be defined as above by equations (1.25). So there is a one-to-one correspondence between orthonormal bases and classical structures. If the basis is not normalised then the corresponding structure will not be special~\cite{cqm2014}. In quantum theory pairs of mutually unbiased bases of Hilbert spaces are a well studied and important phenomenon. The classical structure analogue is pairs of complementary classical structures. 
\begin{definition}[Complementary Bases]
Given two orthonormal bases $\ket {a_i}$ and $\ket {b_j}$ for $H \in \ob (\cat{FHilb})$ of dimension $d$, they are mutually unbiased when:
\begin{equation}
|\inprod{a_i}{b_j}|^2 \, = \, \frac{1}{d}
\end{equation}
$\forall i,j \,, \, \, 0 \leq i,j < d-1$~\cite{cqm2014}.
\end{definition}
For the corresponding classical structures \tinymult[blackdot] and \tinymult[whitedot] this means that they satisfy the following equation~\cite{cqm2014}:
\begin{equation}
d
\begin{pic}[string]
\draw (-0.5,0.25) to (-0.5,1) node [blackdot] {} to [out=left, in=right] (-1,2) node [blackdot] {} to [out=left, in=right] (-1.5,1.5) node [whitedot] {} to [out=left, in=down] (-2,2) to [out=up, in=left] (-0.75,3) node (a) [whitedot] {} to [out=right, in=right] (-0.5,1);
\draw (a.center) to +(0,0.75);
\end{pic}
\quad=\quad
\begin{pic}[string]
\draw (0,0.25) to (0,1) node [blackdot] {};
\draw (0,3) node [whitedot] {} to (0,3.75);
\end{pic}
\end{equation}
This is equivalent to the following morphism being unitary~\cite{cqm2014}:
\begin{equation}
\sqrt{d}\,\,
\begin{aligned}
\begin{tikzpicture}[yscale=0.75,string]
\node (b) [blackdot] at (0,0) {};
\node (w) [whitedot] at (1,1) {};
\draw (-0.75,2) to [out=down, in=left] (b.center);
\draw (b.center) to [out=right, in=left] (w.center);
\draw (w.center) to (1,2);
\draw (b.center) to (0,-1);
\draw (w.center) to [out=right, in=up] (1.75,-1);
\end{tikzpicture}
\end{aligned}
\end{equation} 
And since our structures are commutative here we can interchange the roles of black and white in the above equation.

Given the definition of our classical structure maps, spiders can also be written in terms of sums of their copyable basis states and the corresponding  effects (i.e. the adjoints of the basis states) as follows~\cite{qcs}:
 \begin{equation}\label{eq:normalform}
    \overbrace{\underbrace{
    \begin{pic}
     \node[blackdot,scale=2] (A){};
     \draw[string,out=190,in=90] (A) to (-2,-1);
     \draw[string,out=350,in=90] (A) to (2,-1);
     \draw[string,out=190,in=90] (A) to (-1.7,-1);
     \draw[string,out=350,in=90] (A) to (1.7,-1);  
     \draw[string,out=170,in=270] (A) to (-2,1);
     \draw[string,out=10,in=270] (A) to (2,1);
     \draw[string,out=170,in=270] (A) to (-1.7,1);
     \draw[string,out=10,in=270] (A) to (1.7,1); 
     \draw[loosely dotted] (-1.65,0.9) to (1.65,0.9); 
     \draw[loosely dotted] (-1.65,-0.9) to (1.65,-0.9);          
    \end{pic}}_{m}}^{n}
    \quad=\quad
    \sum^{d-1}_{i=0}
    \overbrace{\underbrace{
    \begin{pic}
     \node[state,black,scale=0.5] (A)at (-1.75,0.3){$i$};
     \node[state,hflip,black,scale=0.5] (B)at (-2,-0.3){$i$};
     \node[state,black,scale=0.5] (C)at (1.75,0.3){$i$};
     \node[state,hflip,black,scale=0.5] (D)at (2,-0.3){$i$};
     \node[state,black,scale=0.5] (A')at (-1.25,0.3){$i$};
     \node[state,hflip,black,scale=0.5] (B')at (-1.5,-0.3){$i$};
     \node[state,black,scale=0.5] (C')at (1.25,0.3){$i$};
     \node[state,hflip,black,scale=0.5] (D')at (1.5,-0.3){$i$};
     \draw[string] (A) to (-1.75,1);
     \draw[string] (B) to (-2,-1);
     \draw[string] (C) to (1.75,1);
     \draw[string] (D) to (2,-1);  
     \draw[string] (A') to (-1.25,1);
     \draw[string] (B') to (-1.5,-1);
     \draw[string] (C') to (1.25,1);
     \draw[string] (D') to (1.5,-1); 
     \draw[loosely dotted] (-1,0.1) to (1,0.1); 
     \draw[loosely dotted] (-1.25,-0.1) to (1.25,-0.1);          
    \end{pic}}_{m}}^{n}
  \end{equation}
In the case with one input and one output we have the spectral decomposition of the identity~\cite{qcs}:
\begin{equation*}
\begin{aligned}\begin{pic}
\node[blackdot] (A) {};
\draw[string] (0,-1) to (A);
\draw[string] (0,1) to (A);
\end{pic}\end{aligned}
\quad=\quad
\sum^{d-1}_{i=0}
\begin{aligned}\begin{pic}
\node[state,black,scale=0.5] (A) at (0,0.3) {$i$};
\node[state,hflip,black,scale=0.5] (B) at (0,-0.3) {$i$};
\draw[string] (0,-1) to (B);
\draw[string] (0,1) to (A);
\end{pic}\end{aligned}
\quad=\quad
\begin{aligned}\begin{pic}
\draw[string] (0,-1) to (0,1);
\end{pic}\end{aligned}
\end{equation*}


\chapter{Constructions of Unitary Error Bases}
There are two main constructions of unitary error bases in the literature, the so-called `algebraic construction' and the `combinatorial construction'. In this paper, I will be interested in the latter, as well as a construction that has arisen out of the categorical quantum mechanics programme of research, which I will refer to as the `MUB construction'. 

\section{The Combinatorial Construction}
The combinatorial construction is due to Werner ~\cite{werner2001all}, and produces a unitary error basis called a shift and multiply basis. However, the construction presented here is more along the lines of that found in ~\cite{klapp} rather than in  Werner's paper, which I find somewhat unintuitive. The construction of a shift and multiply basis consisting of $d^2$  unitary matrices of size $d\times d$ requires one latin square of order $d$, and $d$ Hadamard matrices of size $d\times d$. In general there need be no relationship between the Hadamard matrices and the latin square. 

\begin{definition}[Latin Square]
A latin square of order $d$ is a $d\times d$ array such that each row and each column is a permutation of the elements of the cyclic group $\mathbb{Z}_d$~\cite{latinsquare}.  \end{definition}
\begin{definition}[Hadamard Matrix]
A Hadamard matrix of order $d$ is a $d\times d$ matrix, $H$ such that each entry $|H_{ij}|=1$ and $H \circ H^{\dag}=d \mathbb{I}_d$~\cite{Hadamard}.  
\end{definition} 
 
Given a latin square, $L$ and a family of $d$ Hadamard matrices $H^j: 0 \leq j <d$ all of order $d$, define $P_j$ and $H_{\text{diag}(i)}$ as follows. $P_j$ is the $d\times d$ permutation matrix representing the $j^{th}$ row of $L$ and $H^j_{\text{diag}(i)}$ is the $d\times d $ matrix with the $i^{th}$ row of $H^j$ along the diagonal and zeros elsewhere. 

Then a shift and multiply basis $E_{ij}$ is obtained as follows:
\begin{equation}
E_{ij}=\{ P_j \circ H^j_{\text{diag}(i)}: 0\leq i,j < d \}
\end{equation}
There are proofs that a shift and multiply basis is a unitary error basis in both of the papers mentioned above. I will not repeat them here. However once I have developed the requisite machinery, I will produce a purely diagramatic proof of my own in Chapter 5. This proof gives more insight into how a shift and multiply basis produces a UEB.

Latin squares and Hadamard matrices each have a notion of equivalence. Two latin squares are said to be isotopic if one can be obtained from the other by permuting the symbols, the columns or the rows. Two Hadamard matrices are also equivalent if one is obtained from the other by permuting the rows or columns. In addition to when the rows or columns differ only by a phase factor. Isotopic latin squares and equivalent Hadamard matrices produce equivalent UEBs. However, every Hadamard must be modified in exactly the same way to ensure equivalence of UEBs~\cite{werner2001all}.

\section{The MUB Construction}

The following diagram represents abstract quantum teleportation in a symmetric monoidal category, where $(A,\tinymult[blackdot],\tinyunit[blackdot])$ and $(A,\tinymult[whitedot],\tinyunit[whitedot])$ are Frobenius algebras ~\cite{cqm2014}:\\
%
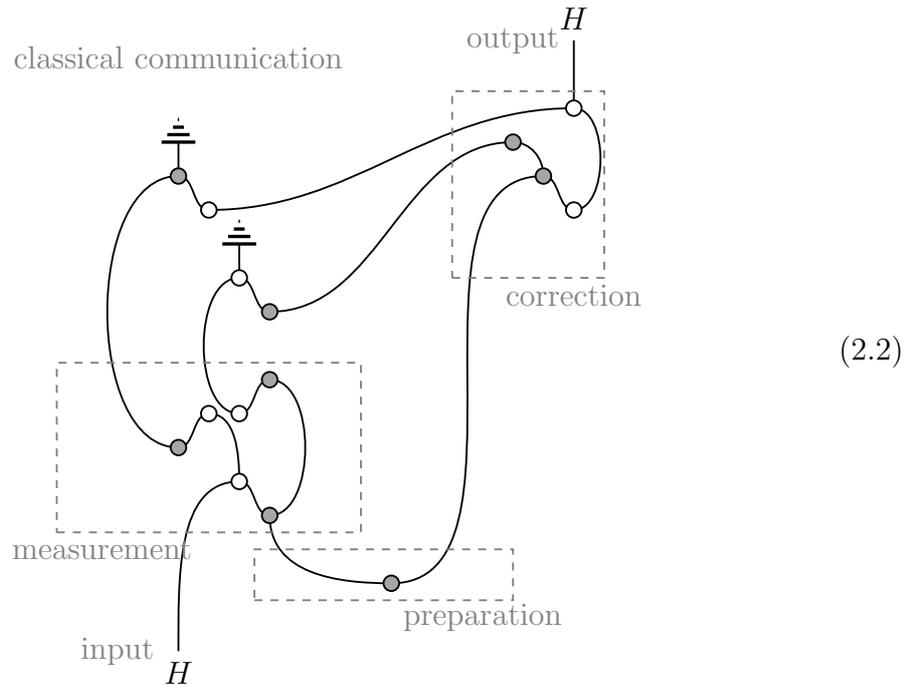
\begin{equation}\begin{pic}[string,yscale=.9,xscale=.8]
      \draw (-1,0) node [below] {$H$}
      to [out=up, in=left] (0,2.5) node [whitedot] {} 
      to [out=up, in=right] (-.5,3.5) node [whitedot] {}
      to [out=left, in=right] (-1,3) node [blackdot] {}
      to [out=left, in=left] (-1,7) node [blackdot] {}
      to (-1,7.5) node [ground] {};
      \draw (0,2.5) 
      to [out=right, in=left] (.5,2) node [blackdot] {}
      to [out=right, in=right] (.5,4) node [blackdot] {}
      to [out=left, in=right] (0,3.5) node [whitedot] {}
      to [out=left, in=left] (0,5.5) node [whitedot] {}
      to (0,6) node [ground] {};
      \draw (.5,2)
      to [out=down, in=left] (2.5,1) node [blackdot] {}
      to [out=right, in=left] (5,7) node [blackdot] {}
      to [out=right, in=left] (5.5,6.5) node [whitedot] {}
      to [out=right, in=right] (5.5,8) node [whitedot] {}
      to (5.5,8);
      \draw (-1,7)
      to [out=right, in=left] (-.5,6.5) node [whitedot] {}
      to [out=right, in=left] (5.5,8) 
      to (5.5,9) node [above] {$H$};
      \draw (0,5.5)
      to [out=right, in=left] (.5,5) node [blackdot] {}
      to [out=right, in=left] (4.5,7.5) node [blackdot] {}
      to [out=right, in=up] (5,7);
      \draw [gray, dashed] (-3,1.75) rectangle (2,4.25);
      \draw [gray, dashed] (.25,1.5) rectangle (4.5,0.75);
      \draw [gray, dashed] (3.5,5.5) rectangle (6,8.25);
      \node [gray] at (-2,0) {input};
      \node [gray] at (4.5,9) {output};
      \node [gray] at (4,.5) {preparation};
      \node [gray] at (-2.25,1.5) {measurement};
      \node [gray] at (5.5,5.25) {correction};
      \node [gray] at (-1,8.75) {classical communication};
    \end{pic} \end{equation}
This is more general than is necessary for our purposes. We make the choice of $\cat{FHilb}$ for our symmetric monoidal category and choose  dagger commutative special Frobenius algebras i.e classical structures as our Frobenius algebras. 
The set of all $d^2$ corrections form a unitary error basis which characterises a particular quantum teleportation protocol.

\subsection{Construction of a Pair of Complementary Classical Structures from an Abelian Group}
\begin{theorem}
Any Hadamard matrix order $d$ is the change of basis matrix between a pair of MUBs on a $d$ dimensional Hilbert space~\cite{mubs}.
\end{theorem}
Given a finite abelian group $G$ of order $d$, a pair of classical structures can be canonically derived.

Take $H$ to be the $d$ dimensional Hilbert space with ONB given by the elements of $G$. Let us denote that ONB: 
\begin{equation*}\ket{a_i} \,,  \, \, 0\leq i <d\text{ where }a_i \in G\end{equation*}
\begin{definition}[Main Classical Structure]
Let us denote the classical structure for which $\ket{a_i}$ form a complete ONB of copyable states, $(H,\tinymult[blackdot],\tinyunit[blackdot])$ and call it the main classical structure for the group $G$.
\end{definition}
Let $F$ be the Fourier transform matrix of $G$. $F$ will be a Hadamard matrix of order $d$. Thus the vectors $F(\ket{a_i}) \,,  \, \, 0\leq i <d$ form a basis which is mutually unbiased to $\ket{a_i} \,,  \, \, 0\leq i <d$. Let us denote:
\begin{equation*}
\ket{c_i}:= F(\ket{a_i}) \,,  \, \, 0\leq i <d
\end{equation*}
\begin{definition}[$G$-Frobenius Algebra]
Let us denote the dagger Frobenius algebra for which $\ket{c_i}$ as above form a complete basis of copyable states as $(H,\tinymultagg,\tinyunitagg)$ and call it a $G$-Frobenius algebra.
 This Frobenius algebra has multiplication given by the linear extension of the binary operation of $G$.
\end{definition}
Please note that since $F$ is Hadamard, $F^{\dag}F=d\mathbb{I}$ and so $F$ is not a unitary transformation and in particular does not preserve norm.
Thus $\ket{c_i}$ is  not orthonormal and our $G$-Frobenius algebra $(H,\tinymultagg,\tinyunitagg)$, is not special. $G$-Frobenius algebras are commutative dagger Frobenius algebras (the commutativity coming from the abelian group's operation).

It is however, easy to see how $\ket{c_i}$ can be normalised. Note that $\frac{1}{\sqrt{d}}F$ is unitary.\\
Let:
\begin{equation*}
\ket{b_i}:=\frac{1}{\sqrt{d}}F(\ket{a_i}) \,,  \, \, 0\leq i <d
\end{equation*}
Since $\ket{b_i}$ is obtained from $\ket{a_i}$ which is orthonormal, via a unitary transformation, {\ket{b_i} is an ONB. In fact $\ket{b_i}$ is simply $\ket{c_i}$ normalised so $\ket{b_i}$ is also mutually unbiased to $\ket{a_i}$. \begin{definition}[Normalised $G$-Classical Structure]
Let us denote the classical structure for which $\ket{b_i}$ form a complete ONB of copyable states, $(H,\tinymult[whitedot],\tinyunit[whitedot])$ and call it a normalised $G$-classical structure.
\end{definition}
So to recap from our abelian group $G$, we have a pair of mutually unbiased ONB $\ket{a_i}$ and $\ket{b_i}$ (where $a_i \in G$) with corresponding classical structures, the main classical structure for $G$, $(H,\tinymult[blackdot],\tinyunit[blackdot])$ and the normalised $G$-classical structure $(H,\tinymult[whitedot],\tinyunit[whitedot])$ respectively.

In addition we have a non-normalised basis $\ket{c_i}$ which is mutually unbiased with $\ket{a_i}$ and corresponding (non-special) $G$-Frobenius algebra $(H,\tinymultagg,\tinyunitagg),$ complementary with our main classical structure. We also know that $\tinymultagg$ is the linear extension of the binary operation on $G$.
The following two Lemmas establish a relationship between $(H,\tinymultagg,\tinyunitagg)$ and $(H,\tinymult[whitedot],\tinyunit[whitedot])$.

\begin{lemma}
With $(H,\tinymultagg ,\tinyunitagg )$ and $(H,\tinymult[whitedot],\tinyunit[whitedot])$ as above the following equation holds:
\begin{equation}
\tinymultagg \quad = \quad \sqrt{d} \tinymult[whitedot]
\end{equation}
\end{lemma}
\begin{proof}
Let $\ket{a_i} \, , \, \, 0 \leq i < d$ be the white ONB. The non-normalised basis corresponding to $(H,\tinymultagg ,\tinyunitagg )$ is thus given by $\sqrt{d} \ket{a_i} \, , \, \, 0 \leq i < d$. These states are copied by $\tinycomultagg$, so we have: $\forall i \, , \, \, 0 \leq i < d$
\begin{equation*}
 \tinycomultagg(\sqrt{d}\ket{a_i})=\sqrt{d}\ket{a_i} \otimes \sqrt{d}\ket{a_i}=d\ket{a_i}\otimes\ket{a_i} \end{equation*}
Now $\tinycomult[whitedot]$ copies the states $\ket{a_i}$ so: $\forall i \, , \, \, 0 \leq i < d$ 
\begin{equation*}
\sqrt{d}\tinycomult[whitedot](\sqrt{d}\ket{a_i})=d\tinycomult[whitedot](\ket{a_i})=d\ket{a_i}\otimes\ket{a_i}
\end{equation*}
Hence $\sqrt{d}\tinycomult[whitedot]$ and $\tinycomultagg$ are equal on all elements of the basis $\sqrt{d} \ket{b_i} \, , \, \, 0 \leq i < d$ and are thus equal as linear maps.
\end{proof}
\begin{lemma}
\begin{equation}
\tinyunitagg=\frac{1}{\sqrt{d}}\tinyunit[whitedot]
\end{equation}
\end{lemma}
\begin{proof}
The first equation is by unitarity of white. The implication follows by the Lemma 2.2 above:
\begin{equation*}\left[
\begin{aligned}\begin{tikzpicture}

      \node[whitedot] (0a) at (-1,-0.5) {};
          \node (0b) at (-0.5,-0.5) {};
          \node (0c) at (0,-0.5) {};
          \node[whitedot] (7) at (-0.5,0) {};       
          \node (10a) at (-0.5,0.5) {};       
          \draw[string, out=90, in =180] (0a) to (7);
          \draw[string, out=0, in=90] (7) to (0c.center);
          \draw[string] (7) to (10a.center);

\end{tikzpicture}\end{aligned} 
\quad=\quad 
\begin{aligned}\begin{tikzpicture}
           \node (10a) at (-0.5,0) {}; 
           \node (5a) at (-0.5,1) {}; 
           \draw[string] (5a.center) to (10a.center);
\end{tikzpicture}\end{aligned}\right]
\quad \Rightarrow \quad 
\left[
\frac{1}{\sqrt{d}}\begin{aligned}\begin{tikzpicture}

      \node[whitedot] (0a) at (-1,-0.5) {};
          \node (0b) at (-0.5,-0.5) {};
          \node (0c) at (0,-0.5) {};
          \node[agg] (7) at (-0.5,0) {};       
          \node (10a) at (-0.5,0.5) {};       
          \draw[string, out=90, in =180] (0a) to (7);
          \draw[string, out=0, in=90] (7) to (0c.center);
          \draw[string] (7) to (10a.center);

\end{tikzpicture}\end{aligned} 
\quad=\quad 
\begin{aligned}\begin{tikzpicture}
           \node (10a) at (-0.5,0) {}; 
           \node (5a) at (-0.5,1) {}; 
           \draw[string] (5a.center) to (10a.center);
\end{tikzpicture}\end{aligned}\right]
\end{equation*}
So $\frac{1}{\sqrt{d}} \tinyunit[whitedot]$ is the left unit for \tinymultagg and by similar reasoning the right unit as well. Thus $\tinyunitagg=\frac{1}{\sqrt{d}}\tinyunit[whitedot]$
\end{proof}
 \begin{definition}[MUB Construction]
Given an abelian group $G$of  order $d$, let $(H,\tinymult,\tinyunit)$ and $(H,\tinymult[whitedot],\tinyunit[whitedot])$  be the main classical structure and normalised $G$-classical structure canonically derived as above. 

Then 
\begin{equation}
  \label{eq:}
  M_{ij}
   \quad := \quad d
\begin{aligned}\begin{tikzpicture}[yscale=0.9375]
          \node (0a) at (0,0) {};
          \node (i)[state,black,scale=0.5] at (-1,1) {$j$};
          \node (j)[state,scale=0.5] at (-0.5,1) {$i$};
          \node (2a)[blackdot] at (0,2) {};
          \node (2b)[blackdot] at (0.5,1.5) {};
          \node (3a)[whitedot] at (1,1){};
          \node (3b)[whitedot] at (1,2.5){};
          \node (0b) at (1,3.5){};
          \draw[string,out=90,in=180] (0a) to (2b);
          \draw[string,out=0,in=180] (2b) to (3a);
          \draw[string,out=0,in=0] (3a) to (3b);
          \draw[string] (3b) to (0b);
          \draw[string,out=180,in=90] (3b) to(i);
          \draw[string,out=90,in=180] (j) to (2a);
          \draw[string,out=0,in=90] (2a) to (2b);
         
      \end{tikzpicture}\end{aligned}
      \end{equation}
 is an \textit{MUB error basis}. The construction of an MUB error basis from an abelian group is the MUB construction. 
\end{definition}
\textit{Note: this definition is my own and will not (yet) be found in the literature.}\\

\begin{theorem}
An MUB error basis is a unitary error basis.   
\end{theorem}

\begin{proof}
$M_{ij}$ can be simplified as follows:
\\
\begin{equation} d
\begin{aligned}\begin{tikzpicture}[yscale=0.75,xscale=0.65]
          \node (0a) at (0,0) {};
          \node (i)[state,black,scale=0.5] at (-1,1) {$j$};
          \node (j)[state,scale=0.5] at (-0.5,1) {$i$};
          \node (2a)[blackdot] at (0,2) {};
          \node (2b)[blackdot] at (0.5,1.5) {};
          \node (3a)[whitedot] at (1,1){};
          \node (3b)[whitedot] at (1,2.5){};
          \node (0b) at (1,3.5){};
          \draw[string,out=90,in=180] (0a) to (2b);
          \draw[string,out=0,in=180] (2b) to (3a);
          \draw[string,out=0,in=0] (3a) to (3b);
          \draw[string] (3b) to (0b);
          \draw[string,out=180,in=90] (3b) to(i);
          \draw[string,out=90,in=180] (j) to (2a);
          \draw[string,out=0,in=90] (2a) to (2b);
         
      \end{tikzpicture}\end{aligned}
      \quad \overset{(\text{Wh,C) and (Bl,C}) }{=}\quad
     d \begin{aligned}\begin{tikzpicture}[yscale=0.75,xscale=0.65]
          \node (0a) at (1,-0.5) {};
          \node (0c) at (1,1){};
          \node (i)[state,black,scale=0.5] at (-1.5,1.5) {$j$};
          \node (j)[state,scale=0.5] at(1.5,1.5) {$i$};
          \node (2a)[blackdot] at (1,2) {};
          \node (2b)[blackdot] at (0.5,1.5) {};
          \node (3a)[whitedot] at (0,1){};
          \node (3b)[whitedot] at (-1,2){};
          \node (0b) at (-1,3){};
          \draw[string,out=90,in=0] (0c.center) to (2b);
          \draw[string] (0a) to (0c.center);
          \draw[string,out=180,in=0] (2b) to (3a);
          \draw[string,out=180,in=0] (3a) to (3b);
          \draw[string] (3b) to (0b);
          \draw[string,out=180,in=90] (3b) to(i);
          \draw[string,out=90,in=0] (j) to (2a);
          \draw[string,out=180,in=90] (2a) to (2b);
         
      \end{tikzpicture}\end{aligned}
       \quad \overset{(\text{Bl,A})}{=} \quad
    d  \begin{aligned}\begin{tikzpicture}[yscale=0.75,xscale=0.65]
          \node (0a) at (0.5,-1) {};
          \node (i)[state,black,scale=0.5] at (-1.25,1.5) {$j$};
          \node (j)[state,scale=0.5] at(2,0) {$i$};
          \node (2a)[blackdot] at (1.5,0.5) {};
          \node (2b)[blackdot] at (0.75,1.5) {};
          \node (3a)[whitedot] at (0,1){};
          \node (3b)[whitedot] at (-0.75,2){};
          \node (0b) at (-0.75,3){};
          \draw[string,out=90,in=180] (0a) to (2a);
          \draw[string,out=180,in=0] (2b) to (3a);
          \draw[string,out=180,in=0] (3a) to (3b);
          \draw[string] (3b) to (0b);
          \draw[string,out=180,in=90] (3b) to(i);
          \draw[string,out=90,in=0] (j) to (2a);
          \draw[string,out=90,in=00] (2a) to (2b);
         
      \end{tikzpicture}\end{aligned}
      \end{equation}\\  
      So 
   \begin{equation}
  M_{ij}
   \quad =\quad
   d \begin{aligned}\begin{tikzpicture}[yscale=0.75,xscale=0.65]
          \node (0a) at (0.5,-1) {};
          \node (i)[state,black,scale=0.5] at (-1.25,1.5) {$j$};
          \node (j)[state,scale=0.5] at(2,0) {$i$};
          \node (2a)[blackdot] at (1.5,0.5) {};
          \node (2b)[blackdot] at (0.75,1.5) {};
          \node (3a)[whitedot] at (0,1){};
          \node (3b)[whitedot] at (-0.75,2){};
          \node (0b) at (-0.75,3){};
          \draw[string,out=90,in=180] (0a) to (2a);
          \draw[string,out=180,in=0] (2b) to (3a);
          \draw[string,out=180,in=0] (3a) to (3b);
          \draw[string] (3b) to (0b);
          \draw[string,out=180,in=90] (3b) to(i);
          \draw[string,out=90,in=0] (j) to (2a);
          \draw[string,out=90,in=00] (2a) to (2b);
         
      \end{tikzpicture}\end{aligned}
      \end{equation}
This is the form of $M_{ij}$ that I will favour henceforth.\\
\bigskip\\
\textit{$M_{ij}$ is a unitary} \\
$M_{ij}$ is the composition of the following three matrices:
\\
 \begin{equation}\sqrt{d} \begin{aligned}\begin{tikzpicture}
\node (0a) at (0.75,-1) {};
         
          \node (j)[state,scale=0.5] at(2,0) {$i$};
          \node (2a)[blackdot] at (1.5,0.5) {};
          \node (2b) at (1.5,1.5) {};
          
          \draw[string,out=90,in=180] (0a) to (2a);

          \draw[string,out=90,in=0] (j) to (2a);
          \draw[string,out=90,in=270] (2a) to (2b);
         
      \end{tikzpicture}\end{aligned}
      \quad,\quad
       \begin{aligned}\begin{tikzpicture}

          \node (2a) at (1.25,0.5) {};
          \node (2b)[blackdot] at (0.75,1.75) {};
          \node (3a)[whitedot] at (0,0.75){};
          \node (3b) at (-0.5,2){};

          \draw[string,out=180,in=0] (2b) to (3a);
          \draw[string,out=180,in=270] (3a) to (3b);

          \draw[string,out=90,in=00] (2a) to (2b);
      \end{tikzpicture}\end{aligned}
      \quad,\quad
      \sqrt{d} \begin{aligned}\begin{tikzpicture}
\node (0a) at (0.75,-1) {};
         
          \node (i)[state,black,scale=0.5] at(-0.5,0) {$j$};
          \node (2a)[whitedot] at (0,0.5) {};
          \node (2b) at (0,1.5) {};
          
          \draw[string,out=90,in=0] (0a) to (2a);

          \draw[string,out=90,in=180] (i) to (2a);
          \draw[string,out=90,in=270] (2a) to (2b);
         
      \end{tikzpicture}\end{aligned}\end{equation}
   Let us label them $U_1$, $U_2$ and $U_3$, respectively.\\
\textit{$U_1$ is  unitary}\\
The first expression here comes from the fact that the black basis states form a complete copyable basis for the main classical structure, the first equivalence follows by notational simplification and the third is by equation (1.31).
\begin{equation}
\forall a \left[
\sqrt{d} \sqrt{d} \begin{aligned}\begin{pic}[yscale=0.75]
\node[blackdot] (1) at (0,0) {};
\node[blackdot] (2) at (0,1) {};
\node[state,scale=0.5] (b) at (-0.5,-1) {$i$};
\node[state,hflip,scale=0.5] (a) at (-0.5,2) {$i$};
\node[state,black,scale=0.5] (5) at (0.5,-1) {$a$};
\draw[string] (1) to (2);
\draw[string,out=180,in=90] (1) to (b);
\draw[string,out=180,in=270] (2) to (a);
\draw[string,out=90,in=0] (0.5,-1) to (1);
\draw[string,out=270,in=0] (0.5,2) to (2);

\end{pic}\end{aligned}
\quad=\quad
\sqrt{d} \sqrt{d}\begin{aligned}\begin{pic}[yscale=0.75]
\node[state,black,scale=0.5] (5) at (0.5,-1) {$a$};
\node[state,scale=0.5] (b) at (-0.5,-1) {$i$};
\node[state,black,hflip,scale=0.5] (6) at (-0.5,-1) {$a$};
\node[state,hflip,scale=0.5] (a) at (-0.5,2) {$i$};
\node[state,black,scale=0.5] (7) at (-0.5,2) {$a$};
\draw[string] (5) to (0.5,2);
\end{pic}\end{aligned}
\, \, \right]
\quad \Leftrightarrow \quad
\forall a \left[
d \begin{aligned}\begin{pic}[yscale=0.75]
\node[blackdot] (1) at (0,0) {};
\node[blackdot] (2) at (0,1) {};
\node[state,scale=0.5] (b) at (-0.5,-1) {$i$};
\node[state,hflip,scale=0.5] (a) at (-0.5,2) {$i$};
\node[state,black,scale=0.5] (5) at (0.5,-1) {$a$};
\draw[string] (1) to (2);
\draw[string,out=180,in=90] (1) to (b);
\draw[string,out=180,in=270] (2) to (a);
\draw[string,out=90,in=0] (0.5,-1) to (1);
\draw[string,out=270,in=0] (0.5,2) to (2);

\end{pic}\end{aligned}
\quad=\quad 
d|\inprod{a}{i}|^2 \begin{aligned}\begin{pic}[yscale=0.75]
\node[state,black,scale=0.5] (5) at (0.5,-1) {$a$};

\draw[string] (5) to (0.5,2);
\end{pic}\end{aligned} \right]
\quad $$\\$$ \Leftrightarrow \quad
 \left[
d \begin{aligned}\begin{pic}[yscale=0.75]
\node[blackdot] (1) at (0,0) {};
\node[blackdot] (2) at (0,1) {};
\node[state,scale=0.5] (b) at (-0.5,-1) {$i$};
\node[state,hflip,scale=0.5] (a) at (-0.5,2) {$i$};
\node (5) at (0.5,-1) {};
\draw[string] (1) to (2);
\draw[string,out=180,in=90] (1) to (b);
\draw[string,out=180,in=270] (2) to (a);
\draw[string,out=90,in=0] (0.5,-1) to (1);
\draw[string,out=270,in=0] (0.5,2) to (2);

\end{pic}\end{aligned}
\quad=\quad 
 \begin{aligned}\begin{pic}[yscale=0.75]
\node (5) at (0.5,-1) {};

\draw[string] (5.center) to (0.5,2);
\end{pic}\end{aligned} \right]
\end{equation} 
The other direction is similar. Hence $U_1$ is unitary.\\
\textit{$U_2$ is  unitary}
\\

\begin{equation}
\begin{aligned}\begin{tikzpicture}

          \node (2a) at (0,0) {};
          \node (2b)[blackdot] at (-0.5,1.25) {};
          \node (3a)[whitedot] at (-1.25,0.25){};
          \node (3b) at (-1.75,1.5){};

          \draw[string,out=180,in=0] (2b) to (3a);
          \draw[string,out=180,in=270] (3a) to (3b);
         
          \draw[string,out=90,in=00] (2a.center) to (2b);

          \node (2d)[blackdot] at (-0.5,-1.25) {};
          \node (3d)[whitedot] at (-1.25,-0.25){};
          \node (3e) at (-1.75,-1.5){};
          
          \draw[string,out=180,in=0] (2d) to (3d);
          \draw[string,out=180,in=90] (3d) to (3e);
         
          \draw[string,out=270,in=00] (2a.center) to (2d);
      \end{tikzpicture}\end{aligned}
      \quad\overset{(\text{Bl,C and Wh,C})}{=} \quad
      \begin{aligned}\begin{tikzpicture}

          \node (2a) at (0,0) {};
          \node (2b)[blackdot] at (-0.5,1.25) {};
          \node (3a)[whitedot] at (-1.25,0.25){};
          \node (3b) at (-1.75,1.5){};

          \draw[string,out=180,in=0] (2b) to (3a);
          \draw[string,out=180,in=270] (3a) to (3b);
         
          \draw[string,out=90,in=00] (2a.center) to (2b);

          \node (2d)[blackdot] at (0.5,-1.25) {};
          \node (3d)[whitedot] at (1.25,-0.25){};
          \node (3e) at (1.75,-1.5){};
          
          \draw[string,out=0,in=180] (2d) to (3d);
          \draw[string,out=0,in=90] (3d) to (3e);
         
          \draw[string,out=270,in=180] (2a.center) to (2d);
      \end{tikzpicture}\end{aligned}
      \quad\overset{(\text{Bl,SN})}{=} \quad
        \begin{aligned}\begin{tikzpicture}

          \node (2a) at (1.25,0.5) {};
          \node (2b)[whitedot] at (0.75,1.75) {};
          \node (3a)[whitedot] at (0,0.75){};
          \node (3b) at (-0.5,2){};

          \draw[string,out=180,in=0] (2b) to (3a);
          \draw[string,out=180,in=270] (3a) to (3b);

          \draw[string,out=90,in=00] (2a) to (2b);
      \end{tikzpicture}\end{aligned} 
      \quad\overset{(\text{Wh,SN})}{=} \quad
      \begin{aligned}\begin{tikzpicture}
      \node (A) at (0,0) {};
      \node(B) at (0,3) {};
      \draw[string] (A) to (B);
      \end{tikzpicture}\end{aligned}       
 \end{equation}       
 \begin{equation}
\begin{aligned}\begin{tikzpicture}

          \node (2a) at (1.75,1.5) {};
          \node (2d)[blackdot] at (1.25,0.25) {};
          \node (3d)[whitedot] at (0.5,1.25){};
          \node (3e) at (0,0){};
          
          \draw[string,out=180,in=0] (2d) to (3d);
          \draw[string,out=180,in=90] (3d) to (3e.center);
         \draw[string,out=270,in=00] (2a) to (2d);

          \node (2e) at (1.75,-1.5) {};
          \node (2b)[blackdot] at (1.25,-0.25) {};
          \node (3a)[whitedot] at (0.5,-1.25){};

          \draw[string,out=180,in=0] (2b) to (3a);
          \draw[string,out=180,in=270] (3a) to (3e.center);
         \draw[string,out=90,in=00] (2e) to (2b);         
      \end{tikzpicture}\end{aligned}
      \quad\overset{(\text{Bl,C and Wh,C})}{=} \quad
      \begin{aligned}\begin{tikzpicture}
          
          \node (2a) at (1.75,1.5) {};
          \node (2d)[blackdot] at (1.25,0.25) {};
          \node (3d)[whitedot] at (0.5,1.25){};
          \node (3e) at (0,0){};
          
          \draw[string,out=180,in=0] (2d) to (3d);
          \draw[string,out=180,in=90] (3d) to (3e.center);
         \draw[string,out=270,in=00] (2a) to (2d);

          \node (2e) at (-1.75,-1.5) {};
          \node (2b)[blackdot] at (-1.25,-0.25) {};
          \node (3a)[whitedot] at (-0.5,-1.25){};

          \draw[string,out=0,in=180] (2b) to (3a);
          \draw[string,out=0,in=270] (3a) to (3e.center);
         \draw[string,out=90,in=180] (2e) to (2b);
      \end{tikzpicture}\end{aligned}
      \quad\overset{(\text{Wh,SN})}{=} \quad
        \begin{aligned}\begin{tikzpicture}

          \node (2a) at (1.75,1.5) {};
          \node (2d)[blackdot] at (1.25,0.25) {};
          \node (3d)[blackdot] at (0.5,1.25){};
          \node (3e) at (0,0){};
          
          \draw[string,out=180,in=0] (2d) to (3d);
          \draw[string,out=180,in=90] (3d) to (3e);
         \draw[string,out=270,in=00] (2a) to (2d);
          
      \end{tikzpicture}\end{aligned} 
      \quad\overset{(\text{Bl,SN})}{=} \quad
      \begin{aligned}\begin{tikzpicture}
      \node (A) at (0,0) {};
      \node(B) at (0,3) {};
      \draw[string] (A) to (B);
      \end{tikzpicture}\end{aligned}     
 \end{equation}
Hence $U_2$ is unitary.\\
\textit{$U_3$ is unitary}\\
   Note that the basis states for black are phases for white. Thus $U_3$ represents a phase shift, which is unitary, the factor of $\sqrt{d}$ being necessary because white is special~\cite{cqm2014}. Put another way, by equation (2.3) $U_3$ is equal to:
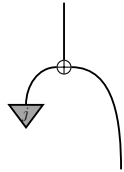
\begin{equation*}\begin{aligned}\begin{tikzpicture}
\node (0a) at (0.75,-1) {};
         
          \node (i)[state,black,scale=0.5] at(-0.5,0) {$j$};
          \node (2a)[agg] at (0,0.5) {};
          \node (2b) at (0,1.5) {};
          
          \draw[string,out=90,in=0] (0a) to (2a);

          \draw[string,out=90,in=180] (i) to (2a);
          \draw[string,out=90,in=270] (2a) to (2b);
         
\end{tikzpicture}\end{aligned}\end{equation*}
$U_3$ thus represents addition by the element $j$ in our abelian group and ${U_3}^{\dag}$ is subtraction by $j$. Composed in either order these clearly give the identity.\\

 \bigskip
 \textit{Orthogonality}\\
 Now we need to show that:
 \begin{equation}
 \tr(M_{ij}\circ {M}^{\dag}_{i'j'})=\delta_{ii'}\delta_{jj'}d   
\end{equation}

\bigskip
\begin{equation}
 \tr(M_{ij}\circ {M}^{\dag}_{i'j'})
   \quad = \quad 
 d^2 \begin{aligned}\begin{tikzpicture}[yscale=0.55,xscale=0.5]
          \node (0a) at (1.25,0) {};
          \node (i)[state,black,scale=0.5] at (-1.5,0.75) {$j$};
          \node (j)[state,scale=0.5] at(2.75,0.625) {$i$};
          \node (2a)[blackdot] at (2.25,1) {};
          \node (2b)[blackdot] at (1.25,4) {};
          \node (3a)[whitedot] at (0,1.5){};
          \node (3b)[whitedot] at (-0.75,2.5){};
          \node (0b)[whitedot] at (-2,3.25){};
          
          \draw[string,out=90,in=180] (0a.center) to (2a);
          \draw[string,out=180,in=0] (2b) to (3a);
          \draw[string,out=180,in=0] (3a) to (3b);
          \draw[string,out=90,in=0] (3b) to (0b);
          \draw[string,out=180,in=90] (3b) to(i);
          \draw[string,out=90,in=0] (j) to (2a);
          \draw[string,out=90,in=00] (2a) to (2b);

          \node (i')[state, hflip,black,scale=0.5] at (-1.5,-0.75) {$j'$};
          \node (j')[state, hflip,scale=0.5] at(2.75,-0.625) {$i'$};
          \node (2a')[blackdot] at (2.25,-1) {};
          \node (2b')[blackdot] at (1.25,-4) {};
          \node (3a')[whitedot] at (0,-1.5){};
          \node (3b')[whitedot] at (-0.75,-2.5){};
          \node (0b')[whitedot] at (-2,-3.25){};
         \node (0) at (-3.5,0){};
          \draw[string,out=270,in=180] (0a.center) to (2a');
          \draw[string,out=180,in=0] (2b') to (3a');
          \draw[string,out=180,in=0] (3a') to (3b');
          \draw[string, out=270,in=0] (3b') to (0b');
          \draw[string,out=180,in=270] (3b') to(i');
          \draw[string,out=270,in=0] (j') to (2a');
          \draw[string,out=270,in=0] (2a') to (2b');
          
          \draw[string,out=180,in=90] (0b) to (0.center);
          \draw[string,out=180,in=270] (0b') to (0.center);
           \draw [gray, dashed] (-4.25,1.25) rectangle (1.25,3.5);
            \draw [gray, dashed] (-4.25,-1.25) rectangle (1.25,-3.5);
         
      \end{tikzpicture}\end{aligned}
      \quad\overset{2\times(\text{Wh,SM})}{=}\quad
 d^2  \begin{aligned}\begin{tikzpicture}[yscale=0.2,xscale=0.25]  
   \node (w1)[blackdot] at (0,0) {};
   \node (b1)[whitedot] at (-2.5,3.25){};
   \node (w2)[blackdot] at (2.5,3.25) {};
   \node (b2)[whitedot] at (-2.5,9.75){};
   \node (w4)[blackdot] at (0,13) {};
   \node (w3)[blackdot] at (2.5,9.75){};
   \node (i)[state,black,scale=0.5] at (-0.5,8.5){$j$};
   \node (i')[state,hflip,black,scale=0.5] at (-0.5,4.5){$j'$};
   \node (j)[state, scale=0.5] at (4.5,8.5) {$i$};
   \node (j')[state,hflip,scale=0.5] at (4.5,4.5){$i'$};
   
   \draw[string,out=180,in=270] (w1) to (b1);
   \draw[string, out=180,in=180] (b1) to (b2);
   \draw[string,out=0,in=270] (w1) to (w2);
   \draw[string, out=180,in=180] (w2) to (w3);   
   \draw[string, out=90,in=180] (b2) to (w4); 
   \draw[string, out=0,in=90] (w4) to (w3); 
   \draw[string, out=0,in=270] (b1) to (i');   
   \draw[string, out=0,in=90] (b2) to (i);
   \draw[string, out=0,in=270] (w2) to (j');
   \draw[string, out=0,in=90] (w3) to (j);
   
    \end{tikzpicture}\end{aligned}
    \quad $$\\$$ \overset{(\text{Wh,C) and (Wh,CC})}{=}\quad
   d^2  \begin{aligned}\begin{tikzpicture}[yscale=0.2,xscale=0.25]  
   \node (w1)[blackdot] at (0,0) {};
   \node (b1)[whitedot] at (-2.5,3.25){};
   \node (w2)[blackdot] at (2.5,3.25) {};
   \node (b2)[whitedot] at (-2.5,9.75){};
   \node (w4)[blackdot] at (0,13) {};
   \node (w3)[blackdot] at (2.5,9.75){};
   \node (i)[state,black,scale=0.5] at (-4.5,8.5){$j$};
   \node (i')[state,hflip,black,scale=0.5] at (-4.5,4.5){$j'$};
   \node (j)[state, scale=0.5] at (4.5,8.5) {$i$};
   \node (j')[state,hflip,scale=0.5] at (4.5,4.5){$i'$};
   
   \draw[string,out=180,in=270] (w1) to (b1);
   \draw[string, out=0,in=0] (b1) to (b2);
   \draw[string,out=0,in=270] (w1) to (w2);
   \draw[string, out=180,in=180] (w2) to (w3);   
   \draw[string, out=90,in=180] (b2) to (w4); 
   \draw[string, out=0,in=90] (w4) to (w3); 
   \draw[string, out=180,in=270] (b1) to (i');   
   \draw[string, out=180,in=90] (b2) to (i);
   \draw[string, out=0,in=270] (w2) to (j');
   \draw[string, out=0,in=90] (w3) to (j);
   
    \end{tikzpicture}\end{aligned}
        \quad  \overset{(\text{Bl,A) and (Bl,CA})}{=}\quad
    d^2 \begin{aligned}\begin{tikzpicture}[yscale=0.2,xscale=0.25]  
   \node (w1)[blackdot] at (-1,1.75) {};
   \node (b1)[whitedot] at (-2.5,3.25){};
   \node (w2)[blackdot] at (0.5,0) {};
   \node (b2)[whitedot] at (-2.5,9.75){};
   \node (w4)[blackdot] at (-1,11.25) {};
   \node (w3)[blackdot] at (0.5,13){};
   \node (i)[state,black,scale=0.5] at (-4.5,8.5){$j$};
   \node (i')[state,hflip,black,scale=0.5] at (-4.5,4.5){$j'$};
   \node (j)[state, scale=0.5] at (2.5,10.75) {$i$};
   \node (j')[state,hflip,scale=0.5] at (2.5,2.25){$i'$};
   
   \draw[string,out=180,in=270] (w1) to (b1);
   \draw[string, out=0,in=0] (b1) to (b2);
   \draw[string,out=270,in=180] (w1) to (w2);
   \draw[string, out=0,in=0] (w1) to (w4);   
   \draw[string, out=90,in=180] (b2) to (w4); 
   \draw[string, out=90,in=180] (w4) to (w3); 
   \draw[string, out=180,in=270] (b1) to (i');   
   \draw[string, out=180,in=90] (b2) to (i);
   \draw[string, out=0,in=270] (w2) to (j');
   \draw[string, out=0,in=90] (w3) to (j);
   
    \end{tikzpicture}\end{aligned}
        \quad $$\\$$ \overset{(\text{Wh,SM})}{=}\quad
    d^2 \begin{aligned}\begin{tikzpicture}[yscale=0.2,xscale=0.25]  
   \node (w1)[blackdot] at (0,3.25) {};
   \node (b1)[whitedot] at (-5,3.25){};
   \node (w2)[blackdot] at (2.5,0) {};
   \node (b2)[whitedot] at (-5,9.75){};
   \node (w4)[blackdot] at (0,9.75) {};
   \node (w3)[blackdot] at (2.5,13){};
   \node (i)[state,black,scale=0.5] at (-7,8.5){$j$};
   \node (i')[state,hflip,black,scale=0.5] at (-7,4.5){$j'$};
   \node (j)[state, scale=0.5] at (4.5,10.75) {$i$};
   \node (j')[state,hflip,scale=0.5] at (4.5,2.25){$i'$};
   \node (b4)[whitedot] at (-2.5,7.75){};
   \node (b3)[whitedot] at (-2.5,5.25){};
   
   \draw[string,out=180,in=0] (w1) to (b3);
   \draw[string] (b3) to (b4);
   \draw[string,out=270,in=180] (w1) to (w2);
   \draw[string, out=0,in=0] (w1) to (w4);   
   \draw[string, out=0,in=180] (b4) to (w4); 
   \draw[string, out=90,in=180] (w4) to (w3); 
   \draw[string, out=180,in=270] (b1) to (i');   
   \draw[string, out=180,in=90] (b2) to (i);
   \draw[string, out=0,in=270] (w2) to (j');
   \draw[string, out=0,in=90] (w3) to (j);
   \draw[string, out=0,in=180] (b1) to (b3);
   \draw[string, out=0,in=180] (b2) to (b4);   
    \end{tikzpicture}\end{aligned}
    \quad  \overset{\text{by equation (1.33)}}{=}\quad
    d
     \begin{aligned}\begin{tikzpicture}[yscale=0.2,xscale=0.25]  
   
   \node (b1)[whitedot] at (0,3.25){};
   \node (w2)[blackdot] at (5,3.25) {};
   \node (b2)[whitedot] at (0,9.75){};
   
   \node (w3)[blackdot] at (5,9.75){};
   \node (i)[state,black,scale=0.5] at (-2,8.5){$j$};
   \node (i')[state,hflip,black,scale=0.5] at (-2,4.5){$j'$};
   \node (j)[state, scale=0.5] at (7,8.5) {$i$};
   \node (j')[state,hflip,scale=0.5] at (7,4.5){$i'$};

   \draw[string, out=180,in=180] (w2) to (w3);

   \draw[string, out=180,in=270] (b1) to (i');   
   \draw[string, out=180,in=90] (b2) to (i);
   \draw[string, out=0,in=270] (w2) to (j');
   \draw[string, out=0,in=90] (w3) to (j);
   \draw[string, out=0,in=0] (b1) to (b2);      
    \end{tikzpicture}\end{aligned}    
            \quad $$\\$$ \overset{(\text{Bl,C) and (Wh,C})}{=} \quad d
     \begin{aligned}\begin{tikzpicture}[yscale=0.2,xscale=0.25]  
   
   \node (b1)[whitedot] at (0,3.25){};
   \node (w2)[blackdot] at (6,3.25) {};
   \node (b2)[whitedot] at (4,9.75){};
   
   \node (w3)[blackdot] at (10,9.75){};
   \node (i)[state,black,scale=0.5] at (6,7.5){$j$};
   \node (i')[state,hflip,black,scale=0.5] at (-2,5.5){$j'$};
   \node (j)[state, scale=0.5] at (12,7.5) {$i$};
   \node (j')[state,hflip,scale=0.5] at (4,5.5){$i'$};

   \draw[string, out=0,in=180] (w2) to (w3);

   \draw[string, out=180,in=270] (b1) to (i');   
   \draw[string, out=0,in=90] (b2) to (i);
   \draw[string, out=180,in=270] (w2) to (j');
   \draw[string, out=0,in=90] (w3) to (j);
   \draw[string, out=0,in=180] (b1) to (b2);
      
    \end{tikzpicture}\end{aligned}
    \quad\overset{(\text{Bl,SN) and (Wh,SN})}{=}\quad d
    \begin{aligned}\begin{tikzpicture}
    \node[state,black,scale=0.5] at (0,0) {$j$};
    \node[state,hflip,black,scale=0.5] at (0,0) {$j'$}; 
    \node[state,scale=0.5] at (0.5,0) {$i$};
    \node[state,hflip,scale=0.5] at (0.5,0) {$i'$};
    \end{tikzpicture}\end{aligned}
\quad \\=\quad
\delta_{ii'}\delta_{jj'}d
\end{equation}

\end{proof}

\subsection{Minimal Shift and Multiply Basis}
I would like to define a minimal shift and multiply basis as follows:
\begin{definition}[Minimal Shift and Multiply Basis]
Given an abelian group, a latin square can be obtained from the group's multiplication table. A Hadamard matrix can be obtained from the matrix of the group's Fourier transform. Using this latin square and Hadamard matrix (in place of the entire family of $d$ Hadamard matrices required) as input, the combinatorial construction gives us a shift and multiply basis. I will refer to a shift and multiply basis obtained from an abelian group in this manner as a minimal shift and multiply basis.
I will refer to this construction as the minimal combinatorial construction. \end{definition}

The following question then naturally arises: If we have a minimal shift and multiply basis and an MUB error basis arising from the same finite abelian group, are they equivalent?
\\
\begin{theorem}
Given an abelian group $G$, the minimal combinatorial construction and the MUB construction produce equivalent UEBs.
\end{theorem}
\begin{proof}
An MUB error basis arising from a finite abelian group, $G$, is represented by:

\begin{equation}
M_{ij}
   \quad =\quad
   d \begin{aligned}\begin{tikzpicture}[yscale=0.75,xscale=0.65]
          \node (0a) at (0.5,-1) {};
          \node (i)[state,black,scale=0.5] at (-1.25,1.5) {$j$};
          \node (j)[state,scale=0.5] at(2,0) {$i$};
          \node (2a)[blackdot] at (1.5,0.5) {};
          \node (2b)[blackdot] at (0.75,1.5) {};
          \node (3a)[whitedot] at (0,1){};
          \node (3b)[whitedot] at (-0.75,2){};
          \node (0b) at (-0.75,3){};
          \draw[string,out=90,in=180] (0a) to (2a);
          \draw[string,out=180,in=0] (2b) to (3a);
          \draw[string,out=180,in=0] (3a) to (3b);
          \draw[string] (3b) to (0b);
          \draw[string,out=180,in=90] (3b) to(i);
          \draw[string,out=90,in=0] (j) to (2a);
          \draw[string,out=90,in=00] (2a) to (2b);
         
      \end{tikzpicture}\end{aligned}
      \end{equation}
Where \begin{tikzpicture}
\node[state,black,scale=0.25] (A)  at (0,0){$j$};
\node (B) at (0,0.4){};
\draw (A) to (B);
\end{tikzpicture}
are an ONB of states corresponding to the elements of $G$, and \begin{tikzpicture}
\node[state,scale=0.25] (A)  at (0,0){$i$};
\node (B) at (0,0.4){};
\draw (A) to (B);
\end{tikzpicture}
are an ONB of states corresponding to \begin{tikzpicture}
\node[state,black,scale=0.25] (A)  at (0,0){j};
\node (B) at (0,0.4){};
\draw (A) to (B);
\end{tikzpicture},
under the change of basis represented by the matrix of the Fourier transform of $G$ which have then been normalised. And where $(H,\tinymult[blackdot],\tinyunit[blackdot])$ and $(H,\tinymult[whitedot],\tinyunit[whitedot])$ are the main  classical structure and normalised $G$-classical structure associated to those ONBs. \\

\textit{Minimal Shift and Multiply Bases}\\
With \begin{tikzpicture}
\node[state,black,scale=0.25] (A)  at (0,0){$j$};
\node (B) at (0,0.4){};
\draw (A) to (B);
\end{tikzpicture}
,
\begin{tikzpicture}
\node[state,scale=0.25] (A)  at (0,0){$i$};
\node (B) at (0,0.4){};
\draw (A) to (B);
\end{tikzpicture}
,
$(\tinymult[blackdot],\tinyunit[blackdot])$
 and
$(\tinymult[whitedot],\tinyunit[whitedot])$ as above the Fourier transform matrix, $H$ is:\\
\begin{equation}
H
\quad=\quad
\sqrt{d} \sum_{k=0}^{d-1}
\begin{aligned}\begin{tikzpicture}[yscale=0.75,xscale=0.65]
          \node (0w) at (0,-0.75) {};
          \node (0b) at (0,1.75) {};
          \node (w)[state,hflip,black,scale=0.5] at (0,0) {$k$};
          \node (b)[state,scale=0.5] at(0,1) {$k$};

          \draw[string] (0w) to (w);
          \draw[string] (0b) to (b);        
      \end{tikzpicture}\end{aligned}
      \end{equation}
which takes any black basis state and returns the corresponding non-normalised white basis state.
Let $+$ represent the binary operation of $G$. Then the latin square, $L$ is given by:
\begin{equation}
L\quad:=\quad
\begin{pmatrix}
   
    0
    &
    1 
    & \cdots
    & \cdots
    & \cdots
    & \matrix{d-1}
    \\
    1
    & \matrix{1+1} & \hdots & \hdots & \hdots & \matrix{(d-1)+1}
    \\
    \vdots & \vdots & \ddots && & \vdots
    \\
     \vdots & \vdots && \matrix{i+j}  &  & \vdots
    \\
    \vdots & \vdots &&& \ddots & \vdots
    \\
    \matrix{d-1}
    & \matrix {1+(d-1)} & \hdots & \hdots & \hdots&
    \matrix{(d-1)+(d-1)}
  \end{pmatrix}
  \end{equation}
Let $P_j$ represent the projection matrix for the $j^{th}$ row of $L$, so that, 
$
P_j\ket{i}=\ket{L_{ij}}.$ \\

Then we have:
\\
\begin{equation}
P_j
\quad=\quad
\sum_{k=0}^{d-1}
\begin{aligned}\begin{tikzpicture}[yscale=0.75,xscale=0.65]
          \node (0w) at (0,-0.75) {};
          \node (0b) at (0,1.75) {};
          \node (w)[state,hflip,black,scale=0.5] at (0,0) {$k$};
          \node (b)[state,black,scale=0.5] at(0,1) {$k+j$};

          \draw[string] (0w) to (w);
          \draw[string] (0b) to (b);        
      \end{tikzpicture}\end{aligned}
      \quad=\quad
 \begin{aligned}\begin{tikzpicture}
\node (0a) at (0.75,-1) {};
         
          \node (j)[state,black,scale=0.5] at(2,0) {$j$};
          \node (2a)[agg] at (1.5,0.5) {};
          \node (2b) at (1.5,1.5) {};
          
          \draw[string,out=90,in=180] (0a) to (2a);

          \draw[string,out=90,in=0] (j) to (2a);
          \draw[string,out=90,in=270] (2a) to (2b);
         
      \end{tikzpicture}\end{aligned}      
      \end{equation}
Let $H_{\text{diag}(i)}$ represent the $(d\times d)$ matrix with the $i{th}$ row of $H$ written along the diagonal and zeros elsewhere. \\
Then
\begin{equation}
H_{\text{diag}(i)}
\quad=\quad
\sqrt{d} \sum_{k=0}^{d-1}
\begin{aligned}\begin{tikzpicture}[yscale=0.75,xscale=0.65]
          \node (0w) at (0,-0.75) {};
          \node (0b) at (0,1.75) {};
          \node (w)[state,hflip,black,scale=0.5] at (0,0) {$k$};
          \node (b)[state,black,scale=0.5] at(0,1) {$k$};
          \node (b2)[state,black,hflip,scale=0.5] at(-0.75,0.5) {$k$};
          \node (b3)[state,scale=0.5] at(-0.75,0.5) {$i$};

          \draw[string] (0w) to (w);
          \draw[string] (0b) to (b);        
      \end{tikzpicture}\end{aligned}     
 \quad=\quad
\sqrt{d} \begin{aligned}\begin{tikzpicture}
\node (0a) at (0.75,-1) {};
         
          \node (i)[state,scale=0.5] at(-0.5,0) {$i$};
          \node (2a)[blackdot] at (0,0.5) {};
          \node (2b) at (0,1.5) {};
          
          \draw[string,out=90,in=0] (0a) to (2a);

          \draw[string,out=90,in=180] (i) to (2a);
          \draw[string,out=90,in=270] (2a) to (2b);
         
      \end{tikzpicture}\end{aligned}
      \end{equation}      
 Now we consider the effect of applying the matrix        $\begin{aligned}\begin{tikzpicture}[scale=0.5]

          \node (2a) at (1.25,0.5) {};
          \node (2b)[blackdot] at (0.75,1.75) {};
          \node (3a)[whitedot] at (0,0.75){};
          \node (3b) at (-0.5,2){};

          \draw[string,out=180,in=0] (2b) to (3a);
          \draw[string,out=180,in=270] (3a) to (3b);

          \draw[string,out=90,in=00] (2a) to (2b);
      \end{tikzpicture}\end{aligned}$ to an arbitrary black basis state $g\in G$ The first equality holds due to the copyability of black states by the black classical structure. The final equality is by definition of \tinycomultagg \, as the adjoint of our group operation:
\begin{equation}
 \begin{aligned}\begin{tikzpicture}

          \node (2a)[state,black,scale=0.75] at (1.25,0.5) {$g$};
          \node (2b)[blackdot] at (0.75,1.75) {};
          \node (3a)[whitedot] at (0,0.75){};
          \node (3b) at (-0.5,2){};

          \draw[string,out=180,in=0] (2b) to (3a);
          \draw[string,out=180,in=270] (3a) to (3b);

          \draw[string,out=90,in=0] (2a) to (2b);
      \end{tikzpicture}\end{aligned}
 \quad=\quad
 \begin{aligned}\begin{tikzpicture}

          \node (2b)[state,black,hflip,scale=0.75] at (0.5,1.5) {$g$};
         
          \node (3a)[whitedot] at (0,0.75){};
          \node (3b) at (-0.5,2){};

          \draw[string,out=270,in=0] (2b) to (3a);
          \draw[string,out=180,in=270] (3a) to (3b);

      \end{tikzpicture}\end{aligned}      
\quad\overset{\text{(Wh,SM)}}{=}\quad
\begin{aligned}\begin{tikzpicture}

          \node (2b)[state,black,hflip,scale=0.75] at (0.5,1.5) {$g$};
         
          \node (3a)[whitedot] at (0,0.75){};
          \node (3d)[whitedot] at (0,0.4){};
          \node (3b) at (-0.5,2){};
         
          \draw[string] (3d) to (3a);
          \draw[string,out=270,in=0] (2b) to (3a);
          \draw[string,out=180,in=270] (3a) to (3b);

      \end{tikzpicture}\end{aligned} 
      \quad\overset{\text{by (2.3)}}{=}\quad
\frac{1}{\sqrt{d}}\begin{aligned}\begin{tikzpicture}

          \node (2b)[state,black,hflip,scale=0.75] at (0.5,1.5) {$g$};
         
          \node (3a)[agg] at (0,0.75){};
          \node (3d)[whitedot] at (0,0.4){};
          \node (3b) at (-0.5,2){};
         
          \draw[string] (3d) to (3a);
          \draw[string,out=270,in=0] (2b) to (3a);
          \draw[string,out=180,in=270] (3a) to (3b);

      \end{tikzpicture}\end{aligned}
            \quad $$\\$$ \overset{\text{by (2.4)}}{=}\quad
\begin{aligned}\begin{tikzpicture}

          \node (2b)[state,black,hflip,scale=0.75] at (0.5,1.5) {$g$};
         
          \node (3a)[agg] at (0,0.75){};
          \node (3d)[agg] at (0,0.4){};
          \node (3b) at (-0.5,2){};
         
          \draw[string] (3d) to (3a);
          \draw[string,out=270,in=0] (2b) to (3a);
          \draw[string,out=180,in=270] (3a) to (3b);

      \end{tikzpicture}\end{aligned}       
\quad =\quad
\begin{aligned}\begin{tikzpicture}
\node[state,black,scale=0.75] (A){$g^{-1}$};
\draw[string] (0,1) to (A);
\end{tikzpicture}\end{aligned}
\end{equation}
It represents the permutation on the group elements that takes each element to it's unique inverse. Applying this permutation before each $P_j$ gives the $j^{th}$ row of the latin square obtained from permuting the symbols of our original latin square. This latin square is, by definition, isotopic to our original latin square, and thus gives an equivalent shift and multiply basis. Let us call the permutation matrix corresponding to the $j^{th}$ row of our new latin square $P'_j$. Let $E'_{ij}$ be the shift and multiply basis obtained using $P'_j$ . \\ So we have:
\begin{equation*}
P'_j
\quad := \quad 
 \begin{aligned}\begin{tikzpicture}[yscale=0.75,xscale=0.65]
          \node (0a) at (1.5,1.5) {};
          
          \node (j)[state,black,scale=0.5] at(2,0) {$j$};
          \node (2a)[agg] at (1.5,0.5) {};
          \node (2b)[blackdot] at (2.25,-0.75) {};
          \node (3a)[whitedot] at (1.5,-1.25){};
          \node (3b) at (2.85,-2){};
         
          \draw[string] (0a) to (2a);
          \draw[string,out=180,in=0] (2b) to (3a);
          \draw[string,out=0,in=90] (2b) to (3b.center);

          \draw[string,out=90,in=0] (j) to (2a);
          \draw[string,out=180,in=180] (2a) to (3a);
         
      \end{tikzpicture}\end{aligned}
      \end{equation*}
Now
\begin{equation*}
E_{ij}
\quad:=\quad
P_j \circ H_{\text{diag}(i)}  
\quad \equiv \quad
P'_j \circ H_{\text{diag}(i)} 
\quad=\quad
E'_{ij}
\end{equation*}
Thus 
\begin{equation*}
E_{ij}
\quad \equiv \quad
E'_{ij}
\quad= \quad 
\sqrt{d} \begin{aligned}\begin{tikzpicture}[yscale=0.75,xscale=0.65]
          \node (0a) at (2.5,-1) {};
          \node (i)[state,black,scale=0.5] at (2,1.25) {$j$};
          \node (j)[state,scale=0.5] at(1,0) {$i$};
          \node (2a)[blackdot] at (1.5,0.5) {};
          \node (2b)[blackdot] at (0.75,1.5) {};
          \node (3a)[whitedot] at (0,1){};
          \node (3b)[agg] at (0,2){};
          \node (0b) at (0,3){};
          \draw[string,out=90,in=0] (0a) to (2a);
          \draw[string,out=180,in=0] (2b) to (3a);
          \draw[string,out=180,in=180] (3a) to (3b);
          \draw[string] (3b) to (0b);
          \draw[string,out=0,in=90] (3b) to(i);
          \draw[string,out=90,in=180] (j) to (2a);
          \draw[string,out=90,in=00] (2a) to (2b);
         
      \end{tikzpicture}\end{aligned}
\quad\overset{\text{by (2.3)}}{=} \quad 
d \begin{aligned}\begin{tikzpicture}[yscale=0.75,xscale=0.65]
          \node (0a) at (2.5,-1) {};
          \node (i)[state,black,scale=0.5] at (2,1.25) {$j$};
          \node (j)[state,scale=0.5] at(1,0) {$i$};
          \node (2a)[blackdot] at (1.5,0.5) {};
          \node (2b)[blackdot] at (0.75,1.5) {};
          \node (3a)[whitedot] at (0,1){};
          \node (3b)[whitedot] at (0,2){};
          \node (0b) at (0,3){};
          \draw[string,out=90,in=0] (0a) to (2a);
          \draw[string,out=180,in=0] (2b) to (3a);
          \draw[string,out=180,in=180] (3a) to (3b);
          \draw[string] (3b) to (0b);
          \draw[string,out=0,in=90] (3b) to(i);
          \draw[string,out=90,in=180] (j) to (2a);
          \draw[string,out=90,in=00] (2a) to (2b);
         
      \end{tikzpicture}\end{aligned}
\quad $$\\$$ \overset{(\text{Wh,C})}{=} \quad 
d \begin{aligned}\begin{tikzpicture}[yscale=0.75,xscale=0.65]
          \node (0a) at (0.5,-1) {};
          \node (i)[state,black,scale=0.5] at (-1.25,1.5) {$j$};
          \node (j)[state,scale=0.5] at(2,0) {$i$};
          \node (2a)[blackdot] at (1.5,0.5) {};
          \node (2b)[blackdot] at (0.75,1.5) {};
          \node (3a)[whitedot] at (0,1){};
          \node (3b)[whitedot] at (-0.75,2){};
          \node (0b) at (-0.75,3){};
          \draw[string,out=90,in=180] (0a) to (2a);
          \draw[string,out=180,in=0] (2b) to (3a);
          \draw[string,out=180,in=0] (3a) to (3b);
          \draw[string] (3b) to (0b);
          \draw[string,out=180,in=90] (3b) to(i);
          \draw[string,out=90,in=0] (j) to (2a);
          \draw[string,out=90,in=00] (2a) to (2b);
         
      \end{tikzpicture}\end{aligned}
      \quad  =\quad
      M_{ij}
      \end{equation*}
      
 So the two UEB constructions are equivalent.\\
 \end{proof}


\chapter{Graphical Characterisation of a Latin Square in Hilbert Space}
A latin square order $d$, can be characterised as a function $\phi :A\times A \rightarrow A$,  where $A$ is a finite set of cardinality $d,$ as follows. $\forall x,y \in A$ the functions $x:A \rightarrow A$ defined by $x(z)=\phi (x,z)$ are injective, and the functions $y:A \rightarrow A$ defined by $y(z)=\phi (z,y)$ are injective.

Now if we consider a latin square to be the multiplication table for a binary operation $*$, defined on $A$, then this amounts to the following:  
 
$\forall a,c \in A, \exists$ unique $b$ s.t $a*b=c$ and $\forall y,z \in A, \exists$ unique $x$ s.t $x*y=z$. I will refer to this as the latin square property.

Since $b$ uniquely exists we can uniquely define $b$:= $a\backslash c$. Which stands for the element that equals $c$ when multiplied by $a$ on the left. Similarly we can define $x:=z/y$, the element that equals $z$ when multiplied on the right by $y$. This leads to the following equalities:
\begin{equation}
(a/b)*b=a  
\end{equation}
\begin{equation}
b*(b\backslash a)=a 
\end{equation}
\begin{equation}
(a*b)/b=a  
\end{equation}
and
\begin{equation}
b\backslash(b* a)=a 
\end{equation}
These equations are equivalent to the latin square property above and thus fully characterise a latin square as follows.\\
To show that $a/b$ is unique: suppose that $c*b=a$. By equation (3.1) $(a/b)*b=a$ Thus $c*b=(a/b)*b \Rightarrow (c*b)/b=((a/b)*b)/b$, so using equation (3.3) we have: $c=a/b$. So we have a unique element $a/b$ s.t $(a/b)*b=a$. The other side is similar using equations (3.1) and (3.4).

Now let $H$ be a $d$ dimensional Hilbert space with an othonormal basis given by the elements of our set, $A$. And let $\tinymultls$ represent our binary operation linearly extended to a linear map $H \otimes H \rightarrow H$. So \tinymultls\,, takes $(a,b)$ to $a*b$ where $a$ and $b$ are basis states, as well as elements of our set $A$. We will also make use of the adjoint of $\tinymultls$:
\begin{equation} \tinycomultls:=\left( \tinymultls \right )^{\dag}. 
\end{equation}
This is the linear map $H  \rightarrow H \otimes H$ taking basis state $c$ to $\{a*b:a*b=c\}$.

I will at times take the liberty of referring to elements of $A$ and basis states of the ONB of $H$ as though they are one and the same.

\section{Graphical Rules for Latin Square Structures}
\subsection{Unitality and Counitality}
The algebraic structure $(A,*)$ with $A$ and $*$ as described above is a quasigroup. Quasigroups are precisely the algebraic structures that have latin squares as their multiplication tables. For a general latin square the quasigroup associated to it will not necessarily have an identity element. A quasigroup with an identity element is known as a loop. We are only interested in latin squares up to isotopy class as described in Chapter 2. Every quasigroup  is isotopic to a loop ~\cite{quasigroups} and so henceforth we will assume that our latin square is a loop. 

We represent the basis state of $H$ corresponding to the unit element of $A$ as \tinyunitls \, and its adjoint, the counit as \tinycounitls. 
\begin{definition}[Latin Square Structure]
I will refer to $(H,\tinymultls, \tinycomultls, \tinyunitls, \tinycounitls)$, as defined above as a latin square structure.
\end{definition}

We can now derive our first graphical rules for a latin square structure:
\begin{equation}
\begin{aligned}\begin{tikzpicture}
          
          \node (0a) at (-1,0) {};
          \node (0b) at (-0.5,0) {};
          \node[ls] (0c) at (0,0) {};
          \node[ls] (7) at (-0.5,-0.5) {};       
          \node (10a) at (-0.5,-1) {};       
          \draw[string, out=270, in =180] (0a.center) to (7);
          \draw[string, out=0, in=270] (7) to (0c);
          \draw[string] (7) to (10a.center);
         
          \end{tikzpicture}\end{aligned}   
  \quad   = \quad
\begin{aligned}\begin{tikzpicture}
           \node (10a) at (-0.5,0) {}; 
           \node (5a) at (-0.5,1) {}; 
           \draw[string] (5a.center) to (10a.center);
\end{tikzpicture}\end{aligned}
\quad=\quad
\begin{aligned}\begin{tikzpicture}

          \node[ls] (0a) at (-1,0) {};
          \node (0b) at (-0.5,0) {};
          \node (0c) at (0,0) {};
          \node[ls] (7) at (-0.5,-0.5) {};       
          \node (10a) at (-0.5,-1) {};       
          \draw[string, out=270, in =180] (0a) to (7);
          \draw[string, out=0, in=270] (7) to (0c.center);
          \draw[string] (7) to (10a.center);

\end{tikzpicture}\end{aligned}  
  \end{equation}
\begin{equation}
\begin{aligned}\begin{tikzpicture}
          
          \node (0a) at (-1,-0.5) {};
          \node (0b) at (-0.5,-0.5) {};
          \node[ls] (0c) at (0,-0.5) {};
          \node[ls] (7) at (-0.5,0) {};       
          \node (10a) at (-0.5,0.5) {};       
          \draw[string, out=90, in =180] (0a.center) to (7);
          \draw[string, out=0, in=90] (7) to (0c);
          \draw[string] (7) to (10a.center);
         
          \end{tikzpicture}\end{aligned}   
  \quad   = \quad
\begin{aligned}\begin{tikzpicture}
           \node (10a) at (-0.5,0) {}; 
           \node (5a) at (-0.5,1) {}; 
           \draw[string] (5a.center) to (10a.center);
\end{tikzpicture}\end{aligned}
\quad=\quad
\begin{aligned}\begin{tikzpicture}

      \node[ls] (0a) at (-1,-0.5) {};
          \node (0b) at (-0.5,-0.5) {};
          \node (0c) at (0,-0.5) {};
          \node[ls] (7) at (-0.5,0) {};       
          \node (10a) at (-0.5,0.5) {};       
          \draw[string, out=90, in =180] (0a) to (7);
          \draw[string, out=0, in=90] (7) to (0c.center);
          \draw[string] (7) to (10a.center);

\end{tikzpicture}\end{aligned}  
  \end{equation}
 The elements of $A$ form an orthonormal basis of $H$, a classical structure with comultiplication that copies these states can thus be canoncally defined. Let $(H,\tinymult[blackdot],\tinyunit[blackdot])$ represent this classical structure. 

\subsection{The Bialgebra Laws}
Given the comonoid part, of the classical structure above with comultiplication $\tinycomult[blackdot]$ and counit $\tinycounit[blackdot]$, the product comonoid, $(H\otimes H$,\tinycomultdb[blackdot]   $\begin{aligned},\tinycounit[blackdot] \tinycounit[blackdot])\end{aligned}$
is also a comonoid ~\cite{cqm2014}.

 Given any function, $f:A\times A \rightarrow A$ extended linearly to a linear map $H \otimes H \rightarrow H$, $f$ takes each basis state of $H \otimes H$ to exactly one basis state of $H$. By definition the comonoids, $(H\otimes H$,\tinycomultdb[blackdot]   $\begin{aligned},\tinycounit[blackdot] \tinycounit[blackdot])\end{aligned}$ and $(H,\tinymult[blackdot],\tinyunit[blackdot])$, copy the basis states of $H \otimes H$ and $H$ respectively. Thus:

 \begin{equation}
\tinycomult[blackdot]( \left[ f(a,b)\right])=(f(a,b),f(a,b)) \end{equation} and
\begin{equation}
(f,f)\left[ \tinycomultdb[blackdot]
(a,b)\right] = (f,f)[ (a,b),(a,b)]=(f(a,b),f(a,b))
\end{equation}
hence \begin{equation} \tinycomult[blackdot]([f(a,b)])=
(f,f)\left[ \tinycomultdb[blackdot]
(a,b)\right]
\end{equation}
So $f$ is a comonoid homomorphism from $(H\otimes H$,\tinycomultdb[blackdot]   $\begin{aligned},\tinycounit[blackdot] \tinycounit[blackdot])\end{aligned}$ to $(H,\tinymult[blackdot],\tinyunit[blackdot])$.

$*:A \times A \rightarrow A$ is a function since $(a*b=c )\land (a*b=c' )\Rightarrow (c=c')$. So\tinymultls is a comonoid homomorphism from $(H\otimes H)$,\tinycomultdb[blackdot] $\begin{aligned},\tinycounit[blackdot] \tinycounit[blackdot])\end{aligned}$ to $(H,\tinymult[blackdot],\tinyunit[blackdot])$.

So by the definition of a comonoid homomorphism we have:

   \begin{equation} \begin{aligned}\begin{tikzpicture}[yscale=0.75]
          \node (0a) at (-0.5,0) {};
          \node (0b) at (0.5,0) {};
          \node[ls] (1) at (0,1) {};
          \node[blackdot] (2) at (0,2) {};
          \node (3a) at (-0.5,3) {};
          \node (3b) at (0.5,3) {};
          \draw[string,out=90,in=180] (-0.5,0) to (1);
          \draw[string,out=90,in=0] (0.5,0) to (1);
          \draw[string] (1) to (0,2);
          \draw[string,out=180,in=-90] (0,2) to (-0.5,3);
          \draw[string,out=0,in=270] (2.center) to (0.5,3);
      \end{tikzpicture}\end{aligned}
    =
    \begin{aligned}\begin{tikzpicture}[yscale=0.75]
           \node[blackdot] (2) at (0,0.9) {};
           \node[blackdot] (3) at (1,0.9) {};
           \node[ls] (4) at (0,2.1) {};
           \node[ls] (5) at (1,2.1) {};
           \draw[string] (0,0) to (2.center);
           \draw[string] (1,0) to (3.center);
           \draw[string] (4) to (0,3);
           \draw[string] (5) to (1,3);
           \draw[string, in=180, out=180, looseness=1.2] (2.center) to (4);
           \draw[string, in=0, out=0, looseness=1.2] (3.center) to (5);
           \draw[string, in=180, out=right] (2.center) to (5);
           \draw[string, in=0, out=180] (3.center) to (4);
    \end{tikzpicture}\end{aligned}
    \text{and} 
 \begin{aligned}\begin{tikzpicture}[yscale=0.75]
          \node[ls] (1) at (0,1) {};
          \node[blackdot] (2) at (0,2) {};
          \draw[string,out=90,in=180] (-0.5,0) to (1);
          \draw[string,out=90,in=0] (0.5,0) to (1);
          \draw[string] (1) to (2.center);
          \draw[string,out=180,in=down, white] (2.center) to (-0.5,3);
          \draw[string,out=0,in=270, white] (2.center) to (0.5,3);
      \end{tikzpicture}\end{aligned}
  =
  \begin{aligned}\begin{tikzpicture}[yscale=0.75, xscale=0.75]
        \draw [white, string] (0,0) to (0,3);
        \node[blackdot] (2) at (0,1) {};
        \node[blackdot] (3) at (1,1) {};
        \draw[string] (0,0) to (0,1) {};
        \draw[string] (1,0) to (1,1) {};
    \end{tikzpicture}\end{aligned}\end{equation} 
Since we also know that \tinyunitls \,  is a basis state copyable by the black classical structure we have:

\begin{equation}\begin{aligned}\begin{tikzpicture}[yscale=-0.75]
          \node[blackdot] (1) at (0,1) {};
          \node[ls] (2) at (0,2) {};
          \draw[string,out=90,in=180] (-0.5,0) to (0,1);
          \draw[string,out=90,in=0] (0.5,0) to (1.center);
          \draw[string] (1.center) to (2);
          \draw[string,out=180,in=down, white] (2) to (-0.5,3);
          \draw[string,out=0,in=270, white] (2) to (0.5,3);
      \end{tikzpicture}\end{aligned}
  =
  \begin{aligned}\begin{tikzpicture}[yscale=-0.75, xscale=0.75]
        \draw [white, string] (0,0) to (0,3);
        \node[ls] (2) at (0,1) {};
        \node[ls] (3) at (1,1) {};
        \draw[string] (0,0) to (2) {};
        \draw[string] (1,0) to (3) {};
    \end{tikzpicture}\end{aligned}
\, \, \, \text{and} \, \, \,
    \begin{aligned}\begin{tikzpicture}[yscale=0.75]
        \draw[string, white] (0,-1) to (0,2);
        \node[ls] (0) at (0,0) {};
        \node[blackdot] (1) at (0,1) {};
        \draw[string] (0) to (1);
    \end{tikzpicture}\end{aligned}
    \,\,=\hspace{10pt}\end{equation}   
The RHS of the final equation being the empty picture.

These are precisely the bialgebra equations. So the multiplication of our latin square structure, and the comonoid of our classical structure form a bialgebra. By taking the adjoint of each side of these equations, we also have that the comultiplication of our latin square structure and the monoid of our classical structure form a bialgebra.

\subsection{Utilising the Black Duality}
Our classical structure gives us a duality on the object $H$ enacted by the black cup and cap. 

Using this we can find a special relationship between our latin square structure's multiplication and comultiplication.

First note that for any basis states in our ONB, $a$ and $b$, $\inprod{a}{b} =\inprod{b}{a}=\delta_{ab}$ and $a*b$ is another unique state.

So $\forall a,b,c \in A$ we have:
\begin{equation}
\begin{aligned}\begin{tikzpicture}
          
          \node[state,black,scale=0.5] (0a) at (-1,-0.5) {$a$};
          \node (0b) at (-0.5,-0.5) {};
          \node[state,black,scale=0.5] (0c) at (0,-0.5) {$b$};
          \node[ls] (7) at (-0.5,0) {};       
          \node[state,black,hflip,scale=0.5] (10a) at (-0.5,0.5) {$c$};       
          \draw[string, out=90, in =180] (0a.center) to (7);
          \draw[string, out=0, in=90] (7) to (0c);
          \draw[string] (7) to (10a.center);
         
          \end{tikzpicture}\end{aligned}
\quad= \quad
\begin{aligned}\begin{tikzpicture}
          \node[state,black,scale=0.5] (0) at (0,0) {$a*b$};
          \node[state,black,hflip,scale=0.5] (1) at (0,0) {$c$};
\end{tikzpicture}\end{aligned}
\quad=\quad \inprod{a*b}{c} =\inprod{c}{a*b}
\quad=\quad
\begin{aligned}\begin{tikzpicture}
          \node[state,black,hflip,scale=0.5] (0a) at (0,0) {$a*b$};
          \node[state,black,scale=0.5] (10a) at (0,0) {$c$};
\end{tikzpicture}\end{aligned}
\quad=\quad
\begin{aligned}\begin{tikzpicture}
          
          \node[state,black,hflip,scale=0.5] (0a) at (-1,0.5) {$a$};
          \node (0b) at (-0.5,0.5) {};
          \node[state,black,hflip,scale=0.5] (0c) at (0,0.5) {$b$};
          \node[ls] (7) at (-0.5,0) {};       
          \node[state,black,scale=0.5] (10a) at (-0.5,-0.5) {$c$};       
          \draw[string, out=270, in =180] (0a.center) to (7);
          \draw[string, out=0, in=270] (7) to (0c);
          \draw[string] (7) to (10a.center);
         
          \end{tikzpicture}\end{aligned}
          \end{equation}
Now note that $\forall a,b,c \in A$:
\begin{equation}
\begin{aligned}\begin{tikzpicture}
          
          \node[state,black,scale=0.5] (0a') at (-2,-1) {$a$};
          \node[blackdot] (a) at (-1.25,0.5) {};
          
          \node (0a) at (-2,-0.5) {};
          \node (0c) at (-1.5,-0.5) {}; 
                            
          \node[blackdot] (b) at (-1,1) {};
          \node[blackdot] (c) at (0,-0.5) {};
          \node (0b) at (-0.5,-0.5) {};
          \node[state,black,scale=0.5] (0c') at (-1.5,-1) {$b$};
          \node[ls] (7) at (-0.5,0) {};       
          \node[state,black,hflip,scale=0.5] (10a) at (0.5,1.5) {$c$};       
          \draw[string, out=90, in =180] (0c.center) to (a);
          \draw[string, out=180, in =0] (7) to (a);
          \draw[string, out=0, in=0] (7) to (b);
          \draw[string, out=180, in=90] (b) to (0a.center);      
          \draw[string,out=270 ,in=180] (7) to (c);
          \draw[string,out=0,in=270] (c) to (10a); 
          \draw[string,out=270,in=90] (0a.center) to (0c');   
          \draw[string,out=270,in=90] (0c.center) to (0a');  
          \end{tikzpicture}\end{aligned}
          \quad=\quad
 \begin{aligned}\begin{tikzpicture}
          
          \node[state,black,hflip,scale=0.5] (0a) at (-1,0.5) {$a$};
          \node (0b) at (-0.5,0.5) {};
          \node[state,black,hflip,scale=0.5] (0c) at (0,0.5) {$b$};
          \node[ls] (7) at (-0.5,0) {};       
          \node[state,black,scale=0.5] (10a) at (-0.5,-0.5) {$c$};       
          \draw[string, out=270, in =180] (0a.center) to (7);
          \draw[string, out=0, in=270] (7) to (0c);
          \draw[string] (7) to (10a.center);
         
          \end{tikzpicture}\end{aligned}         
\end{equation}
So $\forall a,b,c \in A$ we have:
\begin{equation}
\begin{aligned}\begin{tikzpicture}
          
          \node[state,black,scale=0.5] (0a) at (-1,-0.5) {$a$};
          \node (0b) at (-0.5,-0.5) {};
          \node[state,black,scale=0.5] (0c) at (0,-0.5) {$b$};
          \node[ls] (7) at (-0.5,0) {};       
          \node[state,black,hflip,scale=0.5] (10a) at (-0.5,0.5) {$c$};       
          \draw[string, out=90, in =180] (0a.center) to (7);
          \draw[string, out=0, in=90] (7) to (0c);
          \draw[string] (7) to (10a.center);
         
          \end{tikzpicture}\end{aligned}
          \quad=\quad
\begin{aligned}\begin{tikzpicture}
          
          \node[state,black,scale=0.5] (0a') at (-2,-1) {$a$};
          \node[blackdot] (a) at (-1,0.5) {};
          
          \node (0a) at (-2,-0.5) {};
          \node (0c) at (-1.5,-0.5) {}; 
                            
          \node[blackdot] (b) at (-1,1) {};
          \node[blackdot] (c) at (0,-0.5) {};
          \node (0b) at (-0.5,-0.5) {};
          \node[state,black,scale=0.5] (0c') at (-1.5,-1) {$b$};
          \node[ls] (7) at (-0.5,0) {};       
          \node[state,black,hflip,scale=0.5] (10a) at (0.5,1.5) {$c$};       
          \draw[string, out=90, in =180] (0c.center) to (a);
          \draw[string, out=180, in =0] (7) to (a);
          \draw[string, out=0, in=0] (7) to (b);
          \draw[string, out=180, in=90] (b) to (0a.center);      
          \draw[string,out=270 ,in=180] (7) to (c);
          \draw[string,out=0,in=270] (c) to (10a); 
          \draw[string,out=270,in=90] (0a.center) to (0c');   
          \draw[string,out=270,in=90] (0c.center) to (0a');  
          \end{tikzpicture}\end{aligned}
 \end{equation}
 Thus:
\begin{equation}
\begin{aligned}\begin{tikzpicture}
          
          \node (0a) at (-1,-0.5) {};
          \node (0b) at (-0.5,-0.5) {};
          \node (0c) at (0,-0.5) {};
          \node[ls] (7) at (-0.5,0) {};       
          \node (10a) at (-0.5,0.5) {};       
          \draw[string, out=90, in =180] (0a.center) to (7);
          \draw[string, out=0, in=90] (7) to (0c.center);
          \draw[string] (7) to (10a.center);
         
          \end{tikzpicture}\end{aligned}
          \quad=\quad
\begin{aligned}\begin{tikzpicture}
          
          \node (0a') at (-2,-1) {};
          \node[blackdot] (a) at (-1,0.5) {};
          
          \node (0a) at (-2,-0.5) {};
          \node (0c) at (-1.5,-0.5) {}; 
                            
          \node[blackdot] (b) at (-1,1) {};
          \node[blackdot] (c) at (0,-0.5) {};
          \node (0b) at (-0.5,-0.5) {};
          \node (0c') at (-1.5,-1) {};
          \node[ls] (7) at (-0.5,0) {};       
          \node (10a) at (0.5,1.5) {};       
          \draw[string, out=90, in =180] (0c.center) to (a);
          \draw[string, out=180, in =0] (7) to (a);
          \draw[string, out=0, in=0] (7) to (b);
          \draw[string, out=180, in=90] (b) to (0a.center);      
          \draw[string,out=270 ,in=180] (7) to (c);
          \draw[string,out=0,in=270] (c) to (10a); 
          \draw[string,out=270,in=90] (0a.center) to (0c');   
          \draw[string,out=270,in=90] (0c.center) to (0a');  
          \end{tikzpicture}\end{aligned}
 \end{equation} 
 Taking the adjoint of both sides we obtain:
 \begin{equation}
\begin{aligned}\begin{tikzpicture}
          
          \node (0a) at (-1,0.5) {};
          \node (0b) at (-0.5,0.5) {};
          \node (0c) at (0,0.5) {};
          \node[ls] (7) at (-0.5,0) {};       
          \node (10a) at (-0.5,-0.5) {};       
          \draw[string, out=270, in =180] (0a.center) to (7);
          \draw[string, out=0, in=270] (7) to (0c.center);
          \draw[string] (7) to (10a.center);
         
          \end{tikzpicture}\end{aligned}
          \quad=\quad
\begin{aligned}\begin{tikzpicture}
          
          \node (0a') at (-2,1) {};
          \node[blackdot] (a) at (-1,-0.5) {};
          
          \node (0a) at (-2,0.5) {};
          \node (0c) at (-1.5,0.5) {}; 
                            
          \node[blackdot] (b) at (-1,-1) {};
          \node[blackdot] (c) at (0,0.5) {};
          \node (0b) at (-0.5,0.5) {};
          \node (0c') at (-1.5,1) {};
          \node[ls] (7) at (-0.5,0) {};       
          \node (10a) at (0.5,-1.5) {};       
          \draw[string, out=270, in =180] (0c.center) to (a);
          \draw[string, out=180, in =0] (7) to (a);
          \draw[string, out=0, in=0] (7) to (b);
          \draw[string, out=180, in=270] (b) to (0a.center);      
          \draw[string,out=90 ,in=180] (7) to (c);
          \draw[string,out=0,in=90] (c) to (10a); 
          \draw[string,out=90,in=270] (0a.center) to (0c');   
          \draw[string,out=90,in=270] (0c.center) to (0a');  
          \end{tikzpicture}\end{aligned}
 \end{equation} 
\subsection{Unitary Rules}
 \begin{proposition}
 The four equations, 3.1-3.4, that characterise a quasigroup plus the uniqueness of $a*b$, are equivalent to the following statements in our graphical representation of a latin square:
\begin{equation}
 \begin{pic}[scale=0.5]
\node (0) at (0.5,-0.5) {};
\node[ls] (1) at (-1,2) {};
\node (2) at (1.5,3.5) {};
\node[blackdot] (B) at (0.5,1) {};
\node (3) at (-2,-0.5) {};
\node (4) at (-1,3.5) {};
\draw[string,out=180,in=90] (1) to (3);
\draw [string, out=90,in=270] (0) to (B);
\draw [string, out=180,in=0] (B) to (1);
\draw [string, out=0,in=270] (B) to (2);
\draw[string] (1) to (4);
\end{pic} \text{and}  
\begin{pic}[scale=0.5]
\node (0) at (0.5,-0.5) {};
\node[blackdot] (1) at (-1,2) {};
\node (2) at (1.5,3.5) {};
\node[ls] (B) at (0.5,1) {};
\node (3) at (-2,-0.5) {};
\node (4) at (-1,3.5) {};
\draw[string,out=180,in=90] (1) to (3);
\draw [string, out=90,in=270] (0) to (B);
\draw [string, out=180,in=0] (B) to (1);
\draw [string, out=0,in=270] (B) to (2);
\draw[string] (1) to (4);
\end{pic} \text{ are unitary.}
\end{equation} 
 \end{proposition}

\begin{proof}

\textit{Equation (3.1)}\\
$(a/b)*b=a$ can be represented in the following way diagramatically:
 \begin{equation}
 \hspace{-15mm}\forall a,b \in A \, \,  s.t \, \,a=b 
 \left[
\begin{aligned}\begin{pic}[yscale=0.75]
\node[ls] (1) at (0,0) {};
\node[ls] (2) at (0,2) {};
\node[state,black, hflip,scale=0.5] (b) at (0.5,0.5) {$b$};
\node[state,black,scale=0.5] (a) at (0.5,1.5) {$a$};
\draw[string,out=180,in=180] (1) to (2);
\draw[string,out=0,in=270] (1) to (b);
\draw[string,out=0,in=90] (2) to (a);
\draw[string] (0,-1) to (1);
\draw[string] (0,3) to (2);

\end{pic}\end{aligned}
\quad=\quad
\begin{aligned}\begin{pic}[yscale=0.75]
\draw[string] (0,-1) to (0,3);
\end{pic}\end{aligned}
\, \, \right]
\quad  \Leftrightarrow \quad
\forall a,b \in A
 \left[
\begin{aligned}\begin{pic}[yscale=0.75]
\node[ls] (1) at (0,0) {};
\node[ls] (2) at (0,2) {};
\node[state,black, hflip,scale=0.5] (b) at (0.5,0.5) {$b$};
\node[state,black,scale=0.5] (a) at (0.5,1.5) {$a$};
\draw[string,out=180,in=180] (1) to (2);
\draw[string,out=0,in=270] (1) to (b);
\draw[string,out=0,in=90] (2) to (a);
\draw[string] (0,-1) to (1);
\draw[string] (0,3) to (2);

\end{pic}\end{aligned} \delta_{ab}
\quad=\quad
\begin{aligned}\begin{pic}[yscale=0.75]
\draw[string] (0,-1) to (0,3);
\end{pic}\end{aligned} \delta_{ab}
\, \, \right]
\quad $$\\$$ \hspace{-20mm}\Leftrightarrow \quad
\forall a,b \in A
 \left[
\begin{aligned}\begin{pic}[yscale=0.75]
\node[ls] (1) at (0,0) {};
\node[ls] (2) at (0,2) {};
\node[state,black, hflip,scale=0.5] (b) at (0.5,0.5) {$b$};
\node[state,black,scale=0.5] (a) at (0.5,1.5) {$a$};
\draw[string,out=180,in=180] (1) to (2);
\draw[string,out=0,in=270] (1) to (b);
\draw[string,out=0,in=90] (2) to (a);
\draw[string] (0,-1) to (1);
\draw[string] (0,3) to (2);
\node[state,black,scale=0.5] (a') at (1.25,-1) {$b$};
\node[state,black,hflip,scale=0.5] (b') at (1.25,3) {$a$};
\draw[string] (a') to (b');
\end{pic}\end{aligned} 
\quad=\quad
\begin{aligned}\begin{pic}[yscale=0.75]
\draw[string] (0,-1) to (0,3);
\node[state,black,scale=0.5] (a) at (0.75,-1) {$b$};
\node[state,black,hflip,scale=0.5] (b) at (0.75,3) {$a$};
\draw[string] (a) to (b);
\end{pic}\end{aligned} 
\, \, \right]
\quad  \Leftrightarrow \quad
\forall a,b \in A
 \left[
\begin{aligned}\begin{pic}[yscale=0.75]
\node[ls] (1) at (0,0) {};
\node[ls] (2) at (0,2) {};
\node[blackdot] (b) at (0.625,0.5) {};
\node[blackdot] (a) at (0.625,1.5) {};
\draw[string,out=180,in=180] (1) to (2);
\draw[string,out=0,in=180] (1) to (b);
\draw[string,out=0,in=180] (2) to (a);
\draw[string] (a) to (b);
\draw[string] (0,-1) to (1);
\draw[string] (0,3) to (2);
\node[state,black,scale=0.5] (a') at (1.25,-1) {$b$};
\node[state,black,hflip,scale=0.5] (b') at (1.25,3) {$a$};
\draw[string,out=90,in=0] (a') to (b);
\draw[string,out=270,in=0] (b') to (a);
\end{pic}\end{aligned} 
\quad=\quad
\begin{aligned}\begin{pic}[yscale=0.75]
\draw[string] (0,-1) to (0,3);
\node[state,black,scale=0.5] (a) at (0.75,-1) {$b$};
\node[state,black,hflip,scale=0.5] (b) at (0.75,3) {$a$};
\draw[string] (a) to (b);
\end{pic}\end{aligned} 
\, \, \right]
\quad $$\\$$ \Leftrightarrow \quad
 \left[
\begin{aligned}\begin{pic}[yscale=0.75]
\node[ls] (1) at (0,0) {};
\node[ls] (2) at (0,2) {};
\node[blackdot] (b) at (0.625,0.5) {};
\node[blackdot] (a) at (0.625,1.5) {};
\draw[string,out=180,in=180] (1) to (2);
\draw[string,out=0,in=180] (1) to (b);
\draw[string,out=0,in=180] (2) to (a);
\draw[string] (a) to (b);
\draw[string] (0,-1) to (1);
\draw[string] (0,3) to (2);
\node (a') at (1.25,-1) {};
\node (b') at (1.25,3) {};
\draw[string,out=90,in=0] (a'.center) to (b);
\draw[string,out=270,in=0] (b'.center) to (a);
\end{pic}\end{aligned} 
\quad=\quad
\begin{aligned}\begin{pic}[yscale=0.75]
\draw[string] (0,-1) to (0,3);
\node (a) at (0.75,-1) {};
\node (b) at (0.75,3) {};
\draw[string] (a.center) to (b.center);
\end{pic}\end{aligned} 
\, \, \right]
\end{equation}\pagebreak

\textit{Equation (3.2)}\\
Similarly:

$b*(b\backslash a)=a \, \, \, \Leftrightarrow$
 \begin{equation}
 \hspace{-15mm}\forall a,b \in A \, \,  s.t \, \,a=b
 \left[
\begin{aligned}\begin{pic}[yscale=0.75]
\node[ls] (1) at (0,0) {};
\node[ls] (2) at (0,2) {};
\node[state,black, hflip,scale=0.5] (b) at (-0.5,0.5) {$b$};
\node[state,black,scale=0.5] (a) at (-0.5,1.5) {$a$};
\draw[string,out=0,in=0] (1) to (2);
\draw[string,out=180,in=270] (1) to (b);
\draw[string,out=180,in=90] (2) to (a);
\draw[string] (0,-1) to (1);
\draw[string] (0,3) to (2);

\end{pic}\end{aligned}
\quad=\quad
\begin{aligned}\begin{pic}[yscale=0.75]
\draw[string] (0,-1) to (0,3);
\end{pic}\end{aligned}
\, \, \right]
\quad  \Leftrightarrow \quad
\forall a,b \in A
 \left[
 \delta_{ab}
\begin{aligned}\begin{pic}[yscale=0.75]
\node[ls] (1) at (0,0) {};
\node[ls] (2) at (0,2) {};
\node[state,black, hflip,scale=0.5] (b) at (-0.5,0.5) {$b$};
\node[state,black,scale=0.5] (a) at (-0.5,1.5) {$a$};
\draw[string,out=0,in=0] (1) to (2);
\draw[string,out=180,in=270] (1) to (b);
\draw[string,out=180,in=90] (2) to (a);
\draw[string] (0,-1) to (1);
\draw[string] (0,3) to (2);

\end{pic}\end{aligned} 
\quad=\quad
\delta_{ab}
\begin{aligned}\begin{pic}[yscale=0.75]
\draw[string] (0,-1) to (0,3);
\end{pic}\end{aligned} 
\, \, \right]
\quad $$\\$$ \hspace{-20mm}\Leftrightarrow \quad
\forall a,b \in A
 \left[
\begin{aligned}\begin{pic}[yscale=0.75]
\node[ls] (1) at (0,0) {};
\node[ls] (2) at (0,2) {};
\node[state,black, hflip,scale=0.5] (b) at (-0.5,0.5) {$b$};
\node[state,black,scale=0.5] (a) at (-0.5,1.5) {$a$};
\draw[string,out=0,in=0] (1) to (2);
\draw[string,out=180,in=270] (1) to (b);
\draw[string,out=180,in=90] (2) to (a);
\draw[string] (0,-1) to (1);
\draw[string] (0,3) to (2);
\node[state,black,scale=0.5] (a') at (-1.25,-1) {$b$};
\node[state,black,hflip,scale=0.5] (b') at (-1.25,3) {$a$};
\draw[string] (a') to (b');
\end{pic}\end{aligned} 
\quad=\quad
\begin{aligned}\begin{pic}[yscale=0.75]
\draw[string] (0,-1) to (0,3);
\node[state,black,scale=0.5] (a) at (-0.75,-1) {$b$};
\node[state,black,hflip,scale=0.5] (b) at (-0.75,3) {$a$};
\draw[string] (a) to (b);
\end{pic}\end{aligned} 
\, \, \right]
\quad \\ \Leftrightarrow \quad
\forall a,b \in A
 \left[
\begin{aligned}\begin{pic}[yscale=0.75]
\node[ls] (1) at (0,0) {};
\node[ls] (2) at (0,2) {};
\node[blackdot] (b) at (-0.625,0.5) {};
\node[blackdot] (a) at (-0.625,1.5) {};
\draw[string,out=0,in=0] (1) to (2);
\draw[string,out=180,in=0] (1) to (b);
\draw[string,out=180,in=0] (2) to (a);
\draw[string] (a) to (b);
\draw[string] (0,-1) to (1);
\draw[string] (0,3) to (2);
\node[state,black,scale=0.5] (a') at (-1.25,-1) {$b$};
\node[state,black,hflip,scale=0.5] (b') at (-1.25,3) {$a$};
\draw[string,out=90,in=180] (a') to (b);
\draw[string,out=270,in=180] (b') to (a);
\end{pic}\end{aligned} 
\quad=\quad
\begin{aligned}\begin{pic}[yscale=0.75]
\draw[string] (0,-1) to (0,3);
\node[state,black,scale=0.5] (a) at (-0.75,-1) {$b$};
\node[state,black,hflip,scale=0.5] (b) at (-0.75,3) {$a$};
\draw[string] (a) to (b);
\end{pic}\end{aligned} 
\, \, \right]
\quad $$\\$$ \Leftrightarrow \quad
 \left[
\begin{aligned}\begin{pic}[yscale=0.75]
\node[ls] (1) at (0,0) {};
\node[ls] (2) at (0,2) {};
\node[blackdot] (b) at (-0.625,0.5) {};
\node[blackdot] (a) at (-0.625,1.5) {};
\draw[string,out=0,in=0] (1) to (2);
\draw[string,out=180,in=0] (1) to (b);
\draw[string,out=180,in=0] (2) to (a);
\draw[string] (a) to (b);
\draw[string] (0,-1) to (1);
\draw[string] (0,3) to (2);
\node (a') at (-1.25,-1) {};
\node (b') at (-1.25,3) {};
\draw[string,out=90,in=180] (a'.center) to (b);
\draw[string,out=270,in=180] (b'.center) to (a);
\end{pic}\end{aligned} 
\quad=\quad
\begin{aligned}\begin{pic}[yscale=0.75]
\draw[string] (0,-1) to (0,3);
\node (a) at (0.75,-1) {};
\node (b) at (0.75,3) {};
\draw[string] (a.center) to (b.center);
\end{pic}\end{aligned} 
\, \, \right]
\end{equation}
\pagebreak

\textit{Equation (3.3)}\\
$ (a*b)/b=a \, \, \, \Leftrightarrow$ 
\begin{equation}
\hspace{-20mm}\forall a,b \in A \, \,  s.t \, \,a=b \left[
\begin{aligned}\begin{pic}[yscale=0.75]
\node[ls] (1) at (0,0) {};
\node[ls] (2) at (0,1) {};
\node[state,black,scale=0.5] (b) at (0.5,-1) {$b$};
\node[state,black,hflip,scale=0.5] (a) at (0.5,2) {$a$};
\draw[string] (1) to (2);
\draw[string,out=0,in=90] (1) to (b);
\draw[string,out=0,in=270] (2) to (a);
\draw[string,out=90,in=180] (-0.5,-1) to (1);
\draw[string,out=270,in=180] (-0.5,2) to (2);

\end{pic}\end{aligned}
\quad=\quad
\begin{aligned}\begin{pic}[yscale=0.75]
\draw[string] (0,-1) to (0,2);
\end{pic}\end{aligned}
\, \, \right]
\quad  \Leftrightarrow \quad
\forall a,b \in A
 \left[
\begin{aligned}\begin{pic}[yscale=0.75]
\node[ls] (1) at (0,0) {};
\node[ls] (2) at (0,1) {};
\node[state,black,scale=0.5] (b) at (0.5,-1) {$b$};
\node[state,black,hflip,scale=0.5] (a) at (0.5,2) {$a$};
\draw[string] (1) to (2);
\draw[string,out=0,in=90] (1) to (b);
\draw[string,out=0,in=270] (2) to (a);
\draw[string,out=90,in=180] (-0.5,-1) to (1);
\draw[string,out=270,in=180] (-0.5,2) to (2);
\end{pic}\end{aligned} \delta_{ab}
\quad=\quad
\begin{aligned}\begin{pic}[yscale=0.75]
\draw[string] (0,-1) to (0,2);
\end{pic}\end{aligned} \delta_{ab}
\, \, \right]
\quad $$\\$$ \hspace{-20mm}\Leftrightarrow \quad
\forall a,b \in A
 \left[
\begin{aligned}\begin{pic}[yscale=0.75]
\node[ls] (1) at (0,0) {};
\node[ls] (2) at (0,1) {};
\node[state,black,scale=0.5] (b) at (0.5,-1) {$b$};
\node[state,black,hflip,scale=0.5] (a) at (0.5,2) {$a$};
\draw[string] (1) to (2);
\draw[string,out=0,in=90] (1) to (b);
\draw[string,out=0,in=270] (2) to (a);
\draw[string,out=90,in=180] (-0.5,-1) to (1);
\draw[string,out=270,in=180] (-0.5,2) to (2);
\node[state,black,scale=0.5] (a') at (1.25,-1) {$b$};
\node[state,black,hflip,scale=0.5] (b') at (1.25,2) {$a$};
\draw[string] (a') to (b');
\end{pic}\end{aligned} 
\quad=\quad
\begin{aligned}\begin{pic}[yscale=0.75]
\draw[string] (0,-1) to (0,2);
\node[state,black,scale=0.5] (a) at (0.75,-1) {$b$};
\node[state,black,hflip,scale=0.5] (b) at (0.75,2) {$a$};
\draw[string] (a) to (b);
\end{pic}\end{aligned} 
\, \, \right]
\quad \\ \Leftrightarrow \quad
\forall a,b \in A
 \left[
\begin{aligned}\begin{pic}[yscale=0.375,xscale=0.5]
 \node[state,black,scale=0.5] (0) at (0.5,-3.5) {$b$};
\node[ls] (1) at (-1,-1) {};
\node (2) at (1.5,0) {};
\node[blackdot] (B) at (0.5,-2) {};
\node (3) at (-2,-3.5) {};
\node (4) at (-1,0) {};
\draw[string,out=180,in=90] (1) to (3);
\draw [string, out=90,in=270] (0) to (B);
\draw [string, out=180,in=0] (B) to (1);
\draw [string, out=0,in=270] (B) to (2.center);
\draw[string] (1) to (4.center);   

 \node[state,black,hflip,scale=0.5] (5) at (0.5,3.5) {$a$};
\node[ls] (6) at (-1,1) {};

\node[blackdot] (A) at (0.5,2) {};
\node (8) at (-2,3.5) {};

\draw[string,out=180,in=270] (6) to (8);
\draw [string, out=270,in=90] (5) to (A);
\draw [string, out=180,in=0] (A) to (6);
\draw [string, out=0,in=90] (A) to (2.center);
\draw[string] (6) to (4.center);  

\end{pic}\end{aligned} 
\quad=\quad
\begin{aligned}\begin{pic}[yscale=0.75]
\draw[string] (0,-1) to (0,2);
\node[state,black,scale=0.5] (a) at (0.75,-1) {$b$};
\node[state,black,hflip,scale=0.5] (b) at (0.75,2) {$a$};
\draw[string] (a) to (b);
\end{pic}\end{aligned} 
\, \, \right]
\quad $$\\$$ \Leftrightarrow \quad
 \left[
\begin{aligned}\begin{pic}[yscale=0.375,xscale=0.5]
 \node (0) at (0.5,-3.5) {};
\node[ls] (1) at (-1,-1) {};
\node (2) at (1.5,0) {};
\node[blackdot] (B) at (0.5,-2) {};
\node (3) at (-2,-3.5) {};
\node (4) at (-1,0) {};
\draw[string,out=180,in=90] (1) to (3);
\draw [string, out=90,in=270] (0) to (B);
\draw [string, out=180,in=0] (B) to (1);
\draw [string, out=0,in=270] (B) to (2.center);
\draw[string] (1) to (4.center);   

 \node (5) at (0.5,3.5) {};
\node[ls] (6) at (-1,1) {};

\node[blackdot] (A) at (0.5,2) {};
\node (8) at (-2,3.5) {};

\draw[string,out=180,in=270] (6) to (8);
\draw [string, out=270,in=90] (5) to (A);
\draw [string, out=180,in=0] (A) to (6);
\draw [string, out=0,in=90] (A) to (2.center);
\draw[string] (6) to (4.center);  

\end{pic}\end{aligned} 
\quad=\quad
\begin{aligned}\begin{pic}[yscale=0.75]
\draw[string] (0,-1) to (0,2);
\node (a) at (0.75,-1) {};
\node (b) at (0.75,2) {};
\draw[string] (a.center) to (b.center);
\end{pic}\end{aligned} 
\, \, \right]\end{equation}\newpage
\hspace{-6.5mm}\textit{Equation (3.4)}\\
And finally,\\
$b\backslash(b*a)=a \, \, \, \Leftrightarrow$
\begin{equation}
\hspace{-15mm}\forall a,b \in A \, \,  s.t \, \,a=b
 \left[
\begin{aligned}\begin{pic}[yscale=0.75]
\node[ls] (1) at (0,0) {};
\node[ls] (2) at (0,1) {};
\node[state,black,scale=0.5] (b) at (-0.5,-1) {$b$};
\node[state,black,hflip,scale=0.5] (a) at (-0.5,2) {$a$};
\draw[string] (1) to (2);
\draw[string,out=180,in=90] (1) to (b);
\draw[string,out=180,in=270] (2) to (a);
\draw[string,out=90,in=0] (0.5,-1) to (1);
\draw[string,out=270,in=0] (0.5,2) to (2);

\end{pic}\end{aligned}
\quad=\quad
\begin{aligned}\begin{pic}[yscale=0.75]
\draw[string] (0,-1) to (0,2);
\end{pic}\end{aligned}
\, \, \right]
\quad  \Leftrightarrow \quad
\forall a,b \in A
 \left[
 \delta_{ab}
\begin{aligned}\begin{pic}[yscale=0.75]
\node[ls] (1) at (0,0) {};
\node[ls] (2) at (0,1) {};
\node[state,black,scale=0.5] (b) at (-0.5,-1) {$b$};
\node[state,black,hflip,scale=0.5] (a) at (-0.5,2) {$a$};
\draw[string] (1) to (2);
\draw[string,out=180,in=90] (1) to (b);
\draw[string,out=180,in=270] (2) to (a);
\draw[string,out=90,in=0] (0.5,-1) to (1);
\draw[string,out=270,in=0] (0.5,2) to (2);
\end{pic}\end{aligned} 
\quad=\quad
\delta_{ab}
\begin{aligned}\begin{pic}[yscale=0.75]
\draw[string] (0,-1) to (0,2);
\end{pic}\end{aligned} 
\, \, \right]
\quad $$\\$$ \hspace{-20mm}\Leftrightarrow \quad
\forall a,b \in A
 \left[
\begin{aligned}\begin{pic}[yscale=0.75]
\node[ls] (1) at (0,0) {};
\node[ls] (2) at (0,1) {};
\node[state,black,scale=0.5] (b) at (-0.5,-1) {$b$};
\node[state,black,hflip,scale=0.5] (a) at (-0.5,2) {$a$};
\draw[string] (1) to (2);
\draw[string,out=180,in=90] (1) to (b);
\draw[string,out=180,in=270] (2) to (a);
\draw[string,out=90,in=0] (0.5,-1) to (1);
\draw[string,out=270,in=0] (0.5,2) to (2);
\node[state,black,scale=0.5] (a') at (-1.25,-1) {$b$};
\node[state,black,hflip,scale=0.5] (b') at (-1.25,2) {$a$};
\draw[string] (a') to (b');
\end{pic}\end{aligned} 
\quad=\quad
\begin{aligned}\begin{pic}[yscale=0.75]
\draw[string] (0,-1) to (0,2);
\node[state,black,scale=0.5] (a) at (-0.75,-1) {$b$};
\node[state,black,hflip,scale=0.5] (b) at (-0.75,2) {$a$};
\draw[string] (a) to (b);
\end{pic}\end{aligned} 
\, \, \right]
\quad \\ \Leftrightarrow \quad
\forall a,b \in A
 \left[
\begin{aligned}\begin{pic}[yscale=0.375,xscale=0.5]
 \node[state,black,scale=0.5] (0) at (-1.5,-3.5) {$b$};
\node[ls] (1) at (0,-1) {};
\node (2) at (-2.5,0) {};
\node[blackdot] (B) at (-1.5,-2) {};
\node (3) at (1,-3.5) {};
\node (4) at (0,0) {};
\draw[string,out=0,in=90] (1) to (3);
\draw [string, out=90,in=270] (0) to (B);
\draw [string, out=0,in=180] (B) to (1);
\draw [string, out=180,in=270] (B) to (2.center);
\draw[string] (1) to (4.center);   

 \node[state,black,hflip,scale=0.5] (5) at (-1.5,3.5) {$a$};
\node[ls] (6) at (0,1) {};

\node[blackdot] (A) at (-1.5,2) {};
\node (8) at (1,3.5) {};

\draw[string,out=0,in=270] (6) to (8);
\draw [string, out=270,in=90] (5) to (A);
\draw [string, out=0,in=180] (A) to (6);
\draw [string, out=180,in=90] (A) to (2.center);
\draw[string] (6) to (4.center);  

\end{pic}\end{aligned} 
\quad=\quad
\begin{aligned}\begin{pic}[yscale=0.75]
\draw[string] (0,-1) to (0,2);
\node[state,black,scale=0.5] (a) at (-0.75,-1) {$b$};
\node[state,black,hflip,scale=0.5] (b) at (-0.75,2) {$a$};
\draw[string] (a) to (b);
\end{pic}\end{aligned} 
\, \, \right]
\quad $$\\$$ \Leftrightarrow \quad
 \left[
\begin{aligned}\begin{pic}[yscale=0.375,xscale=0.5]
 \node (0) at (-1.5,-3.5) {};
\node[ls] (1) at (0,-1) {};
\node (2) at (-2.5,0) {};
\node[blackdot] (B) at (-1.5,-2) {};
\node (3) at (1,-3.5) {};
\node (4) at (0,0) {};
\draw[string,out=0,in=90] (1) to (3);
\draw [string, out=90,in=270] (0) to (B);
\draw [string, out=0,in=180] (B) to (1);
\draw [string, out=180,in=270] (B) to (2.center);
\draw[string] (1) to (4.center);   

 \node (5) at (-1.5,3.5) {};
\node[ls] (6) at (0,1) {};

\node[blackdot] (A) at (-1.5,2) {};
\node (8) at (1,3.5) {};

\draw[string,out=0,in=270] (6) to (8);
\draw [string, out=270,in=90] (5) to (A);
\draw [string, out=0,in=180] (A) to (6);
\draw [string, out=180,in=90] (A) to (2.center);
\draw[string] (6) to (4.center);

\end{pic}\end{aligned} 
\quad=\quad
\begin{aligned}\begin{pic}[yscale=0.75]
\draw[string] (0,-1) to (0,2);
\node (a) at (0.75,-1) {};
\node (b) at (0.75,2) {};
\draw[string] (a.center) to (b.center);
\end{pic}\end{aligned} 
\, \, \right]
\end{equation}
\end{proof}
\section{Full Characterisation of a Latin Square} 
I have introduced various graphical rules and shown that they are obeyed by our latin square structure. However, are these rules sufficient to fully characterise a latin square?
\begin{proposition}
Given morphisms on $H \in \text{Ob} (\cat{FHilb}), \, \tinymultls :H \otimes H \rightarrow H$, \tinycomultls : $H \rightarrow H \otimes H$, \tinyunitls: $I \rightarrow H$, \tinycounitls : $H \rightarrow I$, and \tinymult[blackdot] a classical structure on $H$, with $\text{dim}(H)=d$ : $(H,\tinymultls, \tinycomultls, \tinyunitls, \tinycounitls)$ is a latin square structure if the following equations are satisfied:

\textit{(Co)unitality}
\begin{equation}
\begin{aligned}\begin{tikzpicture}[scale=2/3]
          
          \node (0a) at (-1,0) {};
          \node (0b) at (-0.5,0) {};
          \node[ls] (0c) at (0,0) {};
          \node[ls] (7) at (-0.5,-0.5) {};       
          \node (10a) at (-0.5,-1) {};       
          \draw[string, out=270, in =180] (0a.center) to (7);
          \draw[string, out=0, in=270] (7) to (0c);
          \draw[string] (7) to (10a.center);
         
          \end{tikzpicture}\end{aligned}   
 =
\begin{aligned}\begin{tikzpicture}[scale=2/3]

          \node[ls] (0a) at (-1,0) {};
          \node (0b) at (-0.5,0) {};
          \node (0c) at (0,0) {};
          \node[ls] (7) at (-0.5,-0.5) {};       
          \node (10a) at (-0.5,-1) {};       
          \draw[string, out=270, in =180] (0a) to (7);
          \draw[string, out=0, in=270] (7) to (0c.center);
          \draw[string] (7) to (10a.center);

\end{tikzpicture}\end{aligned} 
=
\begin{aligned}\begin{tikzpicture}[scale=2/3]
           \node (10a) at (-0.5,0) {}; 
           \node (5a) at (-0.5,1) {}; 
           \draw[string] (5a.center) to (10a.center);
\end{tikzpicture}\end{aligned}
= 
\begin{aligned}\begin{tikzpicture}[scale=2/3]
          
          \node (0a) at (-1,-0.5) {};
          \node (0b) at (-0.5,-0.5) {};
          \node[ls] (0c) at (0,-0.5) {};
          \node[ls] (7) at (-0.5,0) {};       
          \node (10a) at (-0.5,0.5) {};       
          \draw[string, out=90, in =180] (0a.center) to (7);
          \draw[string, out=0, in=90] (7) to (0c);
          \draw[string] (7) to (10a.center);
         
          \end{tikzpicture}\end{aligned}   
=
\begin{aligned}\begin{tikzpicture}[scale=2/3]

      \node[ls] (0a) at (-1,-0.5) {};
          \node (0b) at (-0.5,-0.5) {};
          \node (0c) at (0,-0.5) {};
          \node[ls] (7) at (-0.5,0) {};       
          \node (10a) at (-0.5,0.5) {};       
          \draw[string, out=90, in =180] (0a) to (7);
          \draw[string, out=0, in=90] (7) to (0c.center);
          \draw[string] (7) to (10a.center);

\end{tikzpicture}\end{aligned}  
  \end{equation}

\textit{Bialgebra Laws}
\begin{equation}
    \begin{aligned}\begin{tikzpicture}[yscale=0.325,xscale=0.5]
          \node (0a) at (-0.5,0) {};
          \node (0b) at (0.5,0) {};
          \node[ls] (1) at (0,1) {};
          \node[blackdot] (2) at (0,2) {};
          \node (3a) at (-0.5,3) {};
          \node (3b) at (0.5,3) {};
          \draw[string,out=90,in=180] (-0.5,0) to (1);
          \draw[string,out=90,in=0] (0.5,0) to (1);
          \draw[string] (1) to (0,2);
          \draw[string,out=180,in=-90] (0,2) to (-0.5,3);
          \draw[string,out=0,in=270] (2.center) to (0.5,3);
      \end{tikzpicture}\end{aligned}
    =
    \begin{aligned}\begin{tikzpicture}[yscale=0.325,xscale=0.5]
           \node[blackdot] (2) at (0,0.9) {};
           \node[blackdot] (3) at (1,0.9) {};
           \node[ls] (4) at (0,2.1) {};
           \node[ls] (5) at (1,2.1) {};
           \draw[string] (0,0) to (2.center);
           \draw[string] (1,0) to (3.center);
           \draw[string] (4) to (0,3);
           \draw[string] (5) to (1,3);
           \draw[string, in=180, out=180, looseness=1.2] (2.center) to (4);
           \draw[string, in=0, out=0, looseness=1.2] (3.center) to (5);
           \draw[string, in=180, out=right] (2.center) to (5);
           \draw[string, in=0, out=180] (3.center) to (4);
    \end{tikzpicture}\end{aligned}
    , 
 \begin{aligned}\begin{tikzpicture}[yscale=0.325,xscale=0.5]
          \node[ls] (1) at (0,1) {};
          \node[blackdot] (2) at (0,2) {};
          \draw[string,out=90,in=180] (-0.5,0) to (1);
          \draw[string,out=90,in=0] (0.5,0) to (1);
          \draw[string] (1) to (2.center);
          \draw[string,out=180,in=down, white] (2.center) to (-0.5,3);
          \draw[string,out=0,in=270, white] (2.center) to (0.5,3);
      \end{tikzpicture}\end{aligned}
  =
  \begin{aligned}\begin{tikzpicture}[yscale=0.325, xscale=0.375]
        \draw [white, string] (0,0) to (0,3);
        \node[blackdot] (2) at (0,1) {};
        \node[blackdot] (3) at (1,1) {};
        \draw[string] (0,0) to (0,1) {};
        \draw[string] (1,0) to (1,1) {};
    \end{tikzpicture}\end{aligned} 
\quad,\quad
\begin{aligned}\begin{tikzpicture}[yscale=-0.325,xscale=0.5]
          \node[blackdot] (1) at (0,1) {};
          \node[ls] (2) at (0,2) {};
          \draw[string,out=90,in=180] (-0.5,0) to (0,1);
          \draw[string,out=90,in=0] (0.5,0) to (1.center);
          \draw[string] (1.center) to (2);
          \draw[string,out=180,in=down, white] (2) to (-0.5,3);
          \draw[string,out=0,in=270, white] (2) to (0.5,3);
      \end{tikzpicture}\end{aligned}
  =
  \begin{aligned}\begin{tikzpicture}[yscale=-0.325, xscale=0.375]
        \draw [white, string] (0,0) to (0,3);
        \node[ls] (2) at (0,1) {};
        \node[ls] (3) at (1,1) {};
        \draw[string] (0,0) to (2) {};
        \draw[string] (1,0) to (3) {};
    \end{tikzpicture}\end{aligned}
,
    \begin{aligned}\begin{tikzpicture}[yscale=0.325,xscale=0.5]
        \draw[string, white] (0,-1) to (0,2);
        \node[ls] (0) at (0,0) {};
        \node[blackdot] (1) at (0,1) {};
        \draw[string] (0) to (1);
    \end{tikzpicture}\end{aligned}
    \,\,=\hspace{10pt}   
\end{equation}    

\textit{Duality Relations}
\begin{equation}
\begin{aligned}\begin{tikzpicture}[scale=0.5]
          
          \node (0a) at (-1,-0.5) {};
          \node (0b) at (-0.5,-0.5) {};
          \node (0c) at (0,-0.5) {};
          \node[ls] (7) at (-0.5,0) {};       
          \node (10a) at (-0.5,0.5) {};       
          \draw[string, out=90, in =180] (0a.center) to (7);
          \draw[string, out=0, in=90] (7) to (0c.center);
          \draw[string] (7) to (10a.center);
         
          \end{tikzpicture}\end{aligned}
          =
\begin{aligned}\begin{tikzpicture}[scale=0.5]
          
          \node (0a') at (-2,-1) {};
          \node[blackdot] (a) at (-1,0.5) {};
          
          \node (0a) at (-2,0.25) {};
          \node (0c) at (-1.5,0.25) {}; 
                            
          \node[blackdot] (b) at (-1,1) {};
          \node[blackdot] (c) at (0,-0.5) {};
          \node (0b) at (-0.5,-0.5) {};
          \node (0c') at (-1.5,-1) {};
          \node[ls] (7) at (-0.5,0) {};       
          \node (10a) at (0.5,1.5) {};       
          \draw[string, out=90, in =180] (0c.center) to (a);
          \draw[string, out=180, in =0] (7) to (a);
          \draw[string, out=0, in=0] (7) to (b);
          \draw[string, out=180, in=90] (b) to (0a.center);      
          \draw[string,out=270 ,in=180] (7) to (c);
          \draw[string,out=0,in=270] (c) to (10a); 
          \draw[string,out=270,in=90] (0a.center) to (0c');   
          \draw[string,out=270,in=90] (0c.center) to (0a');  
          \end{tikzpicture}\end{aligned}
 \end{equation},
 \begin{equation}
\begin{aligned}\begin{tikzpicture}[scale=0.5]
          
          \node (0a) at (-1,0.5) {};
          \node (0b) at (-0.5,0.5) {};
          \node (0c) at (0,0.5) {};
          \node[ls] (7) at (-0.5,0) {};       
          \node (10a) at (-0.5,-0.5) {};       
          \draw[string, out=270, in =180] (0a.center) to (7);
          \draw[string, out=0, in=270] (7) to (0c.center);
          \draw[string] (7) to (10a.center);
         
          \end{tikzpicture}\end{aligned}
=
\begin{aligned}\begin{tikzpicture}[scale=0.5]
          
          \node (0a') at (-2,1) {};
          \node[blackdot] (a) at (-1,-0.5) {};
          
          \node (0a) at (-2,-0.25) {};
          \node (0c) at (-1.5,-0.25) {}; 
                            
          \node[blackdot] (b) at (-1,-1) {};
          \node[blackdot] (c) at (0,0.5) {};
          \node (0b) at (-0.5,0.5) {};
          \node (0c') at (-1.5,1) {};
          \node[ls] (7) at (-0.5,0) {};       
          \node (10a) at (0.5,-1.5) {};       
          \draw[string, out=270, in =180] (0c.center) to (a);
          \draw[string, out=180, in =0] (7) to (a);
          \draw[string, out=0, in=0] (7) to (b);
          \draw[string, out=180, in=270] (b) to (0a.center);      
          \draw[string,out=90 ,in=180] (7) to (c);
          \draw[string,out=0,in=90] (c) to (10a); 
          \draw[string,out=90,in=270] (0a.center) to (0c');   
          \draw[string,out=90,in=270] (0c.center) to (0a');  
          \end{tikzpicture}\end{aligned}
 \end{equation}
 
 \textit{Unitarity Property}
 \begin{equation}
\begin{aligned}\begin{tikzpicture}[scale=3/14]
         \node (0) at (0.5,0) {};
\node[ls] (1) at (-1,2) {};
\node (2) at (1.5,3.5) {};
\node[blackdot] (B) at (0.5,1) {};
\node (3) at (-2,0) {};
\node (4) at (-1,3.5) {};
\draw[string,out=180,in=90] (1) to (3.center);
\draw [string, out=90,in=270] (0.center) to (B);
\draw [string, out=180,in=0] (B) to (1);
\draw [string, out=0,in=270] (B) to (2);
\draw[string] (1) to (4);

\node[ls] (6) at (-1,-2) {};
\node (7) at (1.5,-3.5) {};
\node[blackdot] (A) at (0.5,-1) {};

\node (9) at (-1,-3.5) {};
\draw[string,out=180,in=270] (6) to (3.center);
\draw [string, out=270,in=90] (0.center) to (A);
\draw [string, out=180,in=0] (A) to (6);
\draw [string, out=0,in=90] (A) to (7);
\draw[string] (6) to (9);

\end{tikzpicture}\end{aligned}
=
\begin{aligned}\begin{tikzpicture}[scale=3/14]
 \node (0) at (0.5,-3.5) {};
\node[ls] (1) at (-1,-1) {};
\node (2) at (1.5,0) {};
\node[blackdot] (B) at (0.5,-2) {};
\node (3) at (-2,-3.5) {};
\node (4) at (-1,0) {};
\draw[string,out=180,in=90] (1) to (3);
\draw [string, out=90,in=270] (0) to (B);
\draw [string, out=180,in=0] (B) to (1);
\draw [string, out=0,in=270] (B) to (2.center);
\draw[string] (1) to (4.center);   

 \node (5) at (0.5,3.5) {};
\node[ls] (6) at (-1,1) {};

\node[blackdot] (A) at (0.5,2) {};
\node (8) at (-2,3.5) {};

\draw[string,out=180,in=270] (6) to (8);
\draw [string, out=270,in=90] (5) to (A);
\draw [string, out=180,in=0] (A) to (6);
\draw [string, out=0,in=90] (A) to (2.center);
\draw[string] (6) to (4.center);

\end{tikzpicture}\end{aligned}
=
\begin{aligned}\begin{tikzpicture}[xscale=3/14]
           \node (10a) at (0,0) {}; 
           \node (5a) at (0,1.3) {}; 
           \node (10) at (1.25,0) {}; 
           \node (5) at (1.25,1.3) {};
           \draw[string] (5a.center) to (10a.center);
           \draw[string] (5.center) to (10.center);
\end{tikzpicture}\end{aligned}  
=
\begin{aligned}\begin{tikzpicture}[scale=3/14]
       \node (0) at (0.5,-3.5) {};
\node[blackdot] (1) at (-1,-1) {};
\node (2) at (1.5,0) {};
\node[ls] (B) at (0.5,-2) {};
\node (3) at (-2,-3.5) {};
\node (4) at (-1,0) {};
\draw[string,out=180,in=90] (1) to (3);
\draw [string, out=90,in=270] (0) to (B);
\draw [string, out=180,in=0] (B) to (1);
\draw [string, out=0,in=270] (B) to (2.center);
\draw[string] (1) to (4.center);   

 \node (5) at (0.5,3.5) {};
\node[blackdot] (6) at (-1,1) {};

\node[ls] (A) at (0.5,2) {};
\node (8) at (-2,3.5) {};

\draw[string,out=180,in=270] (6) to (8);
\draw [string, out=270,in=90] (5) to (A);
\draw [string, out=180,in=0] (A) to (6);
\draw [string, out=0,in=90] (A) to (2.center);
\draw[string] (6) to (4.center);  

          \end{tikzpicture}\end{aligned}   
=
\begin{aligned}\begin{tikzpicture}[scale=3/14]

         \node (0) at (0.5,0) {};
\node[blackdot] (1) at (-1,2) {};
\node (2) at (1.5,3.5) {};
\node[ls] (B) at (0.5,1) {};
\node (3) at (-2,0) {};
\node (4) at (-1,3.5) {};
\draw[string,out=180,in=90] (1) to (3.center);
\draw [string, out=90,in=270] (0.center) to (B);
\draw [string, out=180,in=0] (B) to (1);
\draw [string, out=0,in=270] (B) to (2);
\draw[string] (1) to (4);

\node[blackdot] (6) at (-1,-2) {};
\node (7) at (1.5,-3.5) {};
\node[ls] (A) at (0.5,-1) {};

\node (9) at (-1,-3.5) {};
\draw[string,out=180,in=270] (6) to (3.center);
\draw [string, out=270,in=90] (0.center) to (A);
\draw [string, out=180,in=0] (A) to (6);
\draw [string, out=0,in=90] (A) to (7);
\draw[string] (6) to (9);

\end{tikzpicture}\end{aligned}  
  \end{equation}

\end{proposition}

\begin{proof}
Let $A$ be the set of basis elements copyable by \tinycomult[blackdot], that form an orthonormal basis of $H$. Let $*:A \times A \rightarrow A$ represent the binary operation on $A$ s.t \mbox{$a*b:= \begin{aligned}\begin{pic}
\node[ls] (A) at (0,0) {}; 
\node (1) at (0,0.35) {};
\node[state,black,scale=0.5] (5) at (-0.275,-0.3) {$a$};
\node[state,black,scale=0.5] (6) at (0.275,-0.3) {$b$};
\draw[string] (1) to (A);
\draw[string,out=180,in=90] (A) to (5);
\draw[string,out=0,in=90] (A) to (6);

\end{pic}\end{aligned}, \forall a,b \in A$}.

By the bialgebra laws, $a*b$ with $a,b \in A$ and \tinyunitls \, are copyable by \tinycomult[blackdot]. Thus $*$ is closed. From the unitality property we have that the element of $A$ corresponding to \tinyunitls, say $1 \in A$ is both a left and right identity for $*$.
The unitarity properties have already been proven to hold for \tinymultls \, iff the four equations (3.1-3.4), hold for $*$. Thus $(A,*)$ is a loop with a multiplication table that forms a latin square up to isotopy.

\end{proof}

The following section is somewhat of an aside, but is a nice result and demonstrates the graphical rules for  latin square structures in action.
\newpage
\section{Frobenius Latin Squares} 
\begin{proposition}
A latin square structure is a Frobenius algebra, iff it is associative. I.e. the underlying quasigroup is a group.
\end{proposition}

\begin{proof}
Suppose \tinymultls \, is a latin square structure and \tinymult[blackdot] is a classical structure, then:
\begin{equation}
  \label{eq:Frobeniusidentities}
  \left[
    \begin{aligned}\begin{tikzpicture}[scale=0.7]
          \node (0a) at (-0.5,0) {};
          \node (0b) at (0.5,0) {};
          \node[ls] (1) at (0,1) {};
          \node[ls] (2) at (0,2) {};
          \node (3a) at (-0.5,3) {};
          \node (3b) at (0.5,3) {};
          \draw[string,out=90,in=180] (0a) to (1);
          \draw[string,out=90,in=0] (0b) to (1);
          \draw[string] (1) to (2);
          \draw[string,out=180,in=270] (2) to (3a);
          \draw[string,out=0,in=270] (2) to (3b);
          
      \end{tikzpicture}\end{aligned}
    \quad = \quad
    \begin{aligned}\begin{tikzpicture}[scale=0.7]
          \node (0) at (0,0) {};
          \node (0a) at (0,1) {};
          \node[ls] (1) at (0.5,2) {};
          \node[ls] (2) at (1.5,1) {};
          \node (3) at (1.5,0) {};
          \node (4) at (2,3) {};
          \node (4a) at (2,2) {};
          \node (5) at (0.5,3) {};
          \draw[string] (0) to (0a.center);
          \draw[string,out=90,in=180] (0a.center) to (1);
          \draw[string,out=0,in=180] (1) to (2);
          \draw[string,out=0,in=270] (2) to (4a.center);
          \draw[string] (4a.center) to (4);
          \draw[string] (2) to (3);
          \draw[string] (1) to (5);
                    
      \end{tikzpicture}\end{aligned}
      \right]
    \quad  \Rightarrow \quad
    \left[
    \begin{aligned}\begin{tikzpicture}[yscale=0.525,xscale=2/3]
          \node[ls] (0b) at (0.5,0) {};
          \node[ls] (1) at (0,1) {};
          \node[ls] (2) at (0,2) {};
          \node (3a) at (0.5,3) {};
          \node (3b) at (-0.5,3) {};
          \node (-1a) at (-1,-1) {};
          \node (-1b) at (0,-1) {};
          \node (-1c) at (1,-1) {}; 
          \draw[string,out=90,in=180] (-1a) to (1);
          \draw[string,out=90,in=180] (-1b) to (0b);
          \draw[string,out=90,in=0] (-1c) to (0b);
          \draw[string,out=90,in=0] (0b) to (1);
          \draw[string] (1) to (2);
          \draw[string,out=0,in=270] (2) to (3a);
          \draw[string,out=180,in=270] (2) to (3b);
              
      \end{tikzpicture}\end{aligned}
    \quad = \quad
\begin{aligned}\begin{tikzpicture}[yscale=0.525]
          \node (0) at (0,-1) {};
          \node (0a) at (0,1) {};
          \node[ls] (1) at (0.5,2) {};
          \node[ls] (2) at (1.5,1) {};
          \node[ls] (3) at (1.5,0) {};
          \node (4) at (2,3) {};
          \node (4a) at (2,2) {};
          \node (5) at (0.5,3) {};
          \node (6a) at (1,-1) {};
          \node (6b) at (2,-1) {};
          \draw[string] (0) to (0a.center);
          \draw[string,out=90,in=180] (0a.center) to (1);
          \draw[string,out=0,in=180] (1) to (2);
          \draw[string,out=0,in=270] (2) to (4a.center);
          \draw[string] (4a.center) to (4);
          \draw[string] (2) to (3);
          \draw[string,out=180,in=90] (3) to (6a);
          \draw[string,out=0,in=90] (3) to (6b);
          \draw[string] (1) to (5);

      \end{tikzpicture}\end{aligned}
      \right]  
     $$ \\ $$
       \quad \Rightarrow \quad
       \left[
\begin{aligned}\begin{tikzpicture}[yscale=0.5]  
          \node (0a) at (-1,-0.5) {};
          \node (0b) at (-0.5,-0.5) {};
          \node (0c) at (0.5,-0.5) {};
          \node (7a) at (-1,0.5) {};
          \node (7b) at (-0.5,0.5) {};
          \node[blackdot] (6) at (0.5,0.5) {};
          \node[ls] (1) at (0,1) {};
          \node[ls] (2) at (-0.5,1.5) {};
          \node[ls] (3) at (-0.5, 2) {};
          \node (8) at (-1,3) {};
          \node[blackdot] (4) at (0.5,3.5) {};  
          \node (5a) at (-1,4.5) {};
          \node (5b) at (0.5,4.5) {};
          \node (9) at (1.25,2) {};
          \draw[string] (0a) to (7a.center);
          \draw[string, out=90, in =180] (7a.center) to (2);
          \draw[string, out=0, in=90] (2) to (1);
          \draw[string, out=0, in=180] (1) to (6);
          \draw[string, out=180, in=90] (1) to (7b.center);
          \draw[string] (7b.center) to (0b);
          \draw[string] (0c) to (6); 
          \draw[string] (2) to (3);
          \draw[string, out=180, in=270] (3) to (8.center);
          \draw[string] (8.center) to (5a);
          \draw[string, out=0, in=180] (3) to (4);
          \draw[string] (4) to (5b);
          \draw[string, out=0,in=270] (6) to (9.center);
          \draw[string, out=90, in=0] (9.center) to (4);         
\end{tikzpicture}\end{aligned} 
\quad = \quad
\begin{aligned}\begin{tikzpicture}[yscale=0.5,xscale=0.7]
          \node (0) at (0,-2) {};
          \node (0a) at (0,1) {};
          \node[ls] (1) at (0.5,2) {};
          \node[ls] (2) at (1.5,1) {};
          \node[ls] (3) at (1.5,0) {};
          \node (4) at (2.5,3) {};
          \node[blackdot] (4a) at (2.5,2) {};
          \node (5) at (0.5,3) {};
          \node (6a) at (1,-1) {};
          \node[blackdot] (6b) at (2.5,-1) {};
          \node (7a) at (1,-2) {};
          \node (7b) at (2.5,-2) {};
          \node (8) at (3.25,0.5) {};
          \draw[string] (0) to (0a.center);
          \draw[string,out=90,in=180] (0a.center) to (1);
          \draw[string,out=0,in=180] (1) to (2);
          \draw[string,out=0,in=180] (2) to (4a.center);
          \draw[string] (4a.center) to (4);
          \draw[string] (2) to (3);
          \draw[string,out=180,in=90] (3) to (6a.center);
          \draw[string,out=0,in=180] (3) to (6b);
          \draw[string] (1) to (5);
          \draw[string] (7a) to (6a.center);
          \draw[string] (7b) to (6b.center);
          \draw[string, out=0, in=270] (6b) to (8.center);
          \draw[string, out=90, in=0] (8.center) to (4a);

\end{tikzpicture}\end{aligned}
\right]
       \quad \Rightarrow \quad
       \left[
\begin{aligned}\begin{tikzpicture}[yscale=0.5,xscale=0.85]  
          \node (0a) at (-1,-0.5) {};
          \node (0b) at (-0.5,-0.5) {};
          \node (0c) at (0.5,-0.5) {};
          \node (7a) at (-1,0.5) {};
          \node (7b) at (-0.5,0.5) {};
          \node[blackdot] (6) at (0.5,0.5) {};
          \node[ls] (1) at (0,1) {};
          \node[ls] (2) at (-0.5,1.5) {};
          \node[ls] (3) at (-0.5, 2) {};
          \node (8) at (-1,2.5) {};
          \node[blackdot] (4) at (0.5,3) {};  
          \node (5a) at (-0.25,4.25) {};
          \node[ls] (5) at (-0.25,3.5) {};
          \node (9) at (1.25,1.75) {};
          \draw[string] (0a) to (7a.center);
          \draw[string, out=90, in =180] (7a.center) to (2);
          \draw[string, out=0, in=90] (2) to (1);
          \draw[string, out=0, in=180] (1) to (6);
          \draw[string, out=180, in=90] (1) to (7b.center);
          \draw[string] (7b.center) to (0b);
          \draw[string] (0c) to (6); 
          \draw[string] (2) to (3);
          \draw[string, out=180, in=270] (3) to (8.center);
          \draw[string, out=90, in=180] (8.center) to (5);
          \draw[string, out=0,in=90] (5) to (4);
          \draw[string, out=0, in=180] (3) to (4);         
          \draw[string, out=0,in=270] (6) to (9.center);
          \draw[string, out=90, in=0] (9.center) to (4);         
          \draw[string] (5) to (5a);
\end{tikzpicture}\end{aligned} 
\quad = \quad
\begin{aligned}\begin{tikzpicture}[yscale=0.40,xscale=0.5]
          \node (0) at (0,-2) {};
          \node (0a) at (0,1) {};
          \node[ls] (1) at (0.5,2) {};
          \node[ls] (2) at (1.5,1) {};
          \node[ls] (3) at (1.5,0) {};
          \node (5a) at (1.5,4) {};
          \node[blackdot] (4a) at (2.5,2) {};
          \node[ls] (5) at (1.5,3) {};
          \node (6a) at (1,-1) {};
          \node[blackdot] (6b) at (2.5,-1) {};
          \node (7a) at (1,-2) {};
          \node (7b) at (2.5,-2) {};
          \node (8) at (3.25,0.5) {};
          \draw[string] (0) to (0a.center);
          \draw[string,out=90,in=180] (0a.center) to (1);
          \draw[string,out=0,in=180] (1) to (2);
          \draw[string,out=0,in=180] (2) to (4a.center);
          \draw[string] (5) to (5a);
          \draw[string] (2) to (3);
          \draw[string,out=180,in=90] (3) to (6a.center);
          \draw[string,out=0,in=180] (3) to (6b);
          \draw[string,out=90,in=180] (1) to (5);
          \draw[string, out=90,in=0] (4a) to (5);
          \draw[string] (7a) to (6a.center);
          \draw[string] (7b) to (6b.center);
          \draw[string, out=0, in=270] (6b) to (8.center);
          \draw[string, out=90, in=0] (8.center) to (4a);

\end{tikzpicture}\end{aligned}
\right]
$$\\$$
 \quad \text{LHS}  = \quad
\begin{aligned}\begin{tikzpicture}[yscale=0.5,xscale=0.85]  
          \node (0a) at (-1,-0.5) {};
          \node (0b) at (-0.5,-0.5) {};
          \node (0c) at (0.5,-0.5) {};
          \node (7a) at (-1,0.5) {};
          \node (7b) at (-0.5,0.5) {};
          \node[blackdot] (6) at (0.5,0.5) {};
          \node[ls] (1) at (0,1) {};
          \node[ls] (2) at (-0.5,1.5) {};
          \node[ls] (3) at (-0.5, 2) {};
          \node (8) at (-1,2.5) {};
          \node[blackdot] (4) at (0.5,3) {};  
          \node (5a) at (-0.25,4.25) {};
          \node[ls] (5) at (-0.25,3.5) {};
          \node (9) at (1.25,1.75) {};
          \draw[string] (0a) to (7a.center);
          \draw[string, out=90, in =180] (7a.center) to (2);
          \draw[string, out=0, in=90] (2) to (1);
          \draw[string, out=0, in=180] (1) to (6);
          \draw[string, out=180, in=90] (1) to (7b.center);
          \draw[string] (7b.center) to (0b);
          \draw[string] (0c) to (6); 
          \draw[string] (2) to (3);
          \draw[string, out=180, in=270] (3) to (8.center);
          \draw[string, out=90, in=180] (8.center) to (5);
          \draw[string, out=0,in=90] (5) to (4);
          \draw[string, out=0, in=180] (3) to (4);         
          \draw[string, out=0,in=270] (6) to (9.center);
          \draw[string, out=90, in=0] (9.center) to (4);         
          \draw[string] (5) to (5a);
\end{tikzpicture}\end{aligned} 
\quad \overset{(\text{Bl,S})}{=}\quad
\begin{aligned}\begin{tikzpicture}[yscale=0.5,xscale=0.85]
          \node (0a) at (-1,-0.5) {};
          \node (0b) at (-0.5,-0.5) {};
          \node (0c) at (0.5,-0.5) {};
          \node (7a) at (-1,0.5) {};
          \node (7b) at (-0.5,0.5) {};
          \node[blackdot] (6) at (0.5,0.5) {};
          \node[ls] (1) at (0,1) {};
          \node[ls] (2) at (-0.5,1.5) {};
          \node[ls] (3) at (-0.5, 2) {};
          \node (8) at (-1,2.5) {};
          \node[blackdot] (4) at (0.5,3) {};  
          \node (5a) at (-0.25,4.25) {};
          \node[ls] (5) at (-0.25,3.5) {};
          \node[blackdot] (9a) at (1,1.25) {};
          \node[blackdot] (9b) at (1,2.25) {};
          \draw[string] (0a) to (7a.center);
          \draw[string, out=90, in =180] (7a.center) to (2);
          \draw[string, out=0, in=90] (2) to (1);
          \draw[string, out=0, in=180] (1) to (6);
          \draw[string, out=180, in=90] (1) to (7b.center);
          \draw[string] (7b.center) to (0b);
          \draw[string] (0c) to (6); 
          \draw[string] (2) to (3);
          \draw[string, out=180, in=270] (3) to (8.center);
          \draw[string, out=90, in=180] (8.center) to (5);
          \draw[string, out=0,in=90] (5) to (4);
          \draw[string, out=0, in=180] (3) to (4);         
          \draw[string, out=0,in=270] (6) to (9a);
          \draw[string, out=90, in=0] (9b) to (4);         
          \draw[string] (5) to (5a);
          \draw[string,out=180,in=180] (9a) to (9b);
          \draw[string,out=0,in=0] (9a) to (9b);

\end{tikzpicture}\end{aligned}
\quad \overset{(\text{Bl,A})}{=} \quad
\begin{aligned}\begin{tikzpicture}[yscale=0.5,xscale=0.85]
          \node (0a) at (-1,-0.5) {};
          \node (0b) at (-0.5,-0.5) {};
          \node (0c) at (0.5,-0.5) {};
          \node (7a) at (-1,0.5) {};
          \node (7b) at (-0.5,0.5) {};
          \node[blackdot] (6) at (0.5,0.5) {};
          \node[ls] (1) at (0,1) {};
          \node[ls] (2) at (-0.5,1.5) {};
          \node[ls] (3) at (-0.5, 2) {};
          \node (8) at (-1,2.5) {};
          \node[blackdot] (4) at (0.5,3) {};  
          \node (5a) at (-0.25,4.25) {};
          \node[ls] (5) at (-0.25,3.5) {};
          \node[blackdot] (9a) at (1,1.25) {};
          \node[blackdot] (9b) at (0,2.25) {};
          \draw[string] (0a) to (7a.center);
          \draw[string, out=90, in =180] (7a.center) to (2);
          \draw[string, out=0, in=90] (2) to (1);
          \draw[string, out=0, in=180] (1) to (6);
          \draw[string, out=180, in=90] (1) to (7b.center);
          \draw[string] (7b.center) to (0b);
          \draw[string] (0c) to (6); 
          \draw[string] (2) to (3);
          \draw[string, out=180, in=270] (3) to (8.center);
          \draw[string, out=90, in=180] (8.center) to (5);
          \draw[string, out=0,in=90] (5) to (4);
          \draw[string, out=0, in=180] (3) to (9b);         
          \draw[string, out=0,in=270] (6) to (9a);
          \draw[string, out=90, in=180] (9b) to (4);         
          \draw[string] (5) to (5a);
          \draw[string,out=180,in=0] (9a) to (9b);
          \draw[string,out=0,in=0] (9a) to (4);
          \end{tikzpicture}\end{aligned}
\quad $$\\$$ \overset{(\text{Bl,SN})}{=} \quad
\begin{aligned}\begin{tikzpicture}[yscale=0.5,xscale=0.85]
          \node (0a) at (-1,-0.5) {};
          \node (0b) at (-0.5,-0.5) {};
          \node (0c) at (0.5,-0.5) {};
          \node (7a) at (-1,0.5) {};
          \node (7b) at (-0.5,0.5) {};
          \node[blackdot] (6) at (0.5,0.5) {};
          \node[ls] (1) at (0,1) {};
          \node[ls] (2) at (-0.5,1.5) {};
          \node[ls] (3) at (-0.5, 2) {};
          \node (8) at (-1,2.5) {};
          \node[blackdot] (4) at (0.5,3) {};  
          \node (5a) at (-0.25,4.25) {};
          \node[ls] (5) at (-0.25,3.5) {};
          \node[blackdot] (9a) at (1,1.25) {};
          \node[blackdot] (9b) at (0,2.25) {};
          \node[blackdot] (9c) at (1,2.25) {};
          \node[blackdot] (9d) at (1.5,3) {};                    
          \draw[string] (0a) to (7a.center);
          \draw[string, out=90, in =180] (7a.center) to (2);
          \draw[string, out=0, in=90] (2) to (1);
          \draw[string, out=0, in=180] (1) to (6);
          \draw[string, out=180, in=90] (1) to (7b.center);
          \draw[string] (7b.center) to (0b);
          \draw[string] (0c) to (6); 
          \draw[string] (2) to (3);
          \draw[string, out=180, in=270] (3) to (8.center);
          \draw[string, out=90, in=180] (8.center) to (5);
          \draw[string, out=0,in=90] (5) to (4);
          \draw[string, out=0, in=180] (3) to (9b);         
          \draw[string, out=0,in=270] (6) to (9a);
          \draw[string, out=90, in=180] (9b) to (4);         
          \draw[string] (5) to (5a);
          \draw[string,out=180,in=0] (9a) to (9b);
          \draw[string,out=0,in=0] (9a) to (9d);
          \draw[string,out=0,in=180] (4) to (9c);
          \draw[string,out=0,in=180] (9c) to (9d);
          \end{tikzpicture}\end{aligned}
          \quad \overset{(\text{Bl,SM})}{=} \quad
\begin{aligned}\begin{tikzpicture}[yscale=0.5,xscale=0.85]
          \node (0a) at (-1,-0.5) {};
          \node (0b) at (-0.5,-0.5) {};
          \node (0c) at (0.5,-0.5) {};
          \node (7a) at (-1,0.5) {};
          \node (7b) at (-0.5,0.5) {};
          \node[blackdot] (6) at (0.5,0.5) {};
          \node[ls] (1) at (0,1) {};
          \node[ls] (2) at (-0.5,1.5) {};
          \node[ls] (3) at (-0.5, 2) {};
          \node (8) at (-1,2.5) {};
          \node[blackdot] (4) at (0.25,3) {};  
          \node (5a) at (-0.5,4.25) {};
          \node[ls] (5) at (-0.5,3.5) {};
          \node[blackdot] (9a) at (1,1.25) {};
          \node[blackdot] (9b) at (0.25,2.5) {};
        
          \node[blackdot] (9d) at (1,3.5) {};                    
          \draw[string] (0a) to (7a.center);
          \draw[string, out=90, in =180] (7a.center) to (2);
          \draw[string, out=0, in=90] (2) to (1);
          \draw[string, out=0, in=180] (1) to (6);
          \draw[string, out=180, in=90] (1) to (7b.center);
          \draw[string] (7b.center) to (0b);
          \draw[string] (0c) to (6); 
          \draw[string] (2) to (3);
          \draw[string, out=180, in=270] (3) to (8.center);
          \draw[string, out=90, in=180] (8.center) to (5);
          \draw[string, out=0,in=180] (5) to (4);
          \draw[string, out=0, in=180] (3) to (9b);         
          \draw[string, out=0,in=270] (6) to (9a);
          \draw[string, out=90, in=270] (9b) to (4);         
          \draw[string] (5) to (5a);
          \draw[string,out=180,in=0] (9a) to (9b);
          \draw[string,out=0,in=0] (9a) to (9d);
          \draw[string,out=0,in=180] (4) to (9d);
          \end{tikzpicture}\end{aligned}
\quad  \overset{(\text{LS2})}{=} \quad
\begin{aligned}\begin{tikzpicture}[yscale=0.5,xscale=0.85]
          \node (0a) at (-1,-0.5) {};
          \node (0b) at (-0.5,-0.5) {};
          \node (0c) at (0.5,-0.5) {};
          \node (7a) at (-1,0.5) {};
          \node (7b) at (-0.5,0.5) {};
          \node[blackdot] (6) at (0.5,0.5) {};
          \node[ls] (1) at (0,1) {};
          \node[ls] (2) at (-0.5,1.5) {};

          \node (5a) at (-0.5,4.25) {};
          
          \node[blackdot] (9a) at (1,1.25) {};

          \node[blackdot] (9d) at (1,3.5) {};                    
          \node[blackdot] (9e) at (1,4) {};                    
          \draw[string] (0a) to (7a.center);
          \draw[string, out=90, in =180] (7a.center) to (2);
          \draw[string, out=0, in=90] (2) to (1);
          \draw[string, out=0, in=180] (1) to (6);
          \draw[string, out=180, in=90] (1) to (7b.center);
          \draw[string] (7b.center) to (0b);
          \draw[string] (0c) to (6); 
          \draw[string] (9d) to (9e);
                           
          \draw[string, out=0,in=270] (6) to (9a);

          \draw[string,out=180,in=180] (9a) to (9d);
          \draw[string,out=0,in=0] (9a) to (9d);
          \draw[string] (2) to (5a.center);
         
          \end{tikzpicture}\end{aligned}
\quad $$\\$$ \overset{(\text{Bl,S})}{=}\quad
\begin{aligned}\begin{tikzpicture}[scale=2/3]
          \node (0a) at (-1,-0.5) {};
          \node (0b) at (-0.5,-0.5) {};
          \node (0c) at (0.5,-0.5) {};
          \node (7a) at (-1,0.5) {};
          \node (7b) at (-0.5,0.5) {};
          \node[blackdot] (6) at (0.5,0.5) {};
          \node[ls] (1) at (0,1) {};
          \node[ls] (2) at (-0.5,1.5) {};

          \node (5a) at (-0.5,3.1) {};
          
          \node[blackdot] (9a) at (1,1.75) {};
                             
          \draw[string] (0a) to (7a.center);
          \draw[string, out=90, in =180] (7a.center) to (2);
          \draw[string, out=0, in=90] (2) to (1);
          \draw[string, out=0, in=180] (1) to (6);
          \draw[string, out=180, in=90] (1) to (7b.center);
          \draw[string] (7b.center) to (0b);
          \draw[string] (0c) to (6);

          \draw[string, out=0,in=270] (6) to (9a);

          \draw[string] (2) to (5a.center);
         
          \end{tikzpicture}\end{aligned}
\quad \overset{(\text{Bl,CU})}{=} \quad
\begin{aligned}\begin{tikzpicture}[scale=7/6]
          \node (0a) at (-1,0.25) {};
          \node (0b) at (-0.5,0.25) {};
          \node (0c) at (0.5,0.25) {};

          \node[ls] (1) at (0,1) {};
          \node[ls] (2) at (-0.5,1.5) {};

          \node (5a) at (-0.5,2) {};

          \draw[string, out=90, in =180] (0a) to (2);
          \draw[string, out=0, in=90] (2) to (1);
          \draw[string, out=0, in=90] (1) to (0c);
          \draw[string, out=180, in=90] (1) to (0b);

          \draw[string] (2) to (5a.center);
         
          \end{tikzpicture}\end{aligned}
\quad $$\\$$ RHS  = \quad
\begin{aligned}\begin{tikzpicture}[yscale=0.40,xscale=0.435]
          \node (0) at (0,-2) {};
          \node (0a) at (0,1) {};
          \node[ls] (1) at (0.5,2) {};
          \node[ls] (2) at (1.5,1) {};
          \node[ls] (3) at (1.5,0) {};
          \node (5a) at (1.5,4) {};
          \node[blackdot] (4a) at (2.5,2) {};
          \node[ls] (5) at (1.5,3) {};
          \node (6a) at (1,-1) {};
          \node[blackdot] (6b) at (2.5,-1) {};
          \node (7a) at (1,-2) {};
          \node (7b) at (2.5,-2) {};
          \node (8) at (3.25,0.5) {};
          \draw[string] (0) to (0a.center);
          \draw[string,out=90,in=180] (0a.center) to (1);
          \draw[string,out=0,in=180] (1) to (2);
          \draw[string,out=0,in=180] (2) to (4a.center);
          \draw[string] (5) to (5a);
          \draw[string] (2) to (3);
          \draw[string,out=180,in=90] (3) to (6a.center);
          \draw[string,out=0,in=180] (3) to (6b);
          \draw[string,out=90,in=180] (1) to (5);
          \draw[string, out=90,in=0] (4a) to (5);
          \draw[string] (7a) to (6a.center);
          \draw[string] (7b) to (6b.center);
          \draw[string, out=0, in=270] (6b) to (8.center);
          \draw[string, out=90, in=0] (8.center) to (4a);

\end{tikzpicture}\end{aligned}
\quad   \overset{(\text{LS2})}{=} \quad
\begin{aligned}\begin{tikzpicture}[scale=7/6]
          \node (0a) at (0.75,0.25) {};
          \node (0b) at (0.25,0.25) {};
          \node (0c) at (-0.75,0.25) {};

          \node[ls] (1) at (-0.25,1) {};
          \node[ls] (2) at (0.25,1.5) {};

          \node (5a) at (0.25,2) {};

          \draw[string, out=90, in =00] (0a) to (2);
          \draw[string, out=180, in=90] (2) to (1);
          \draw[string, out=180, in=90] (1) to (0c);
          \draw[string, out=0, in=90] (1) to (0b);

          \draw[string] (2) to (5a.center);
\end{tikzpicture}\end{aligned}
\quad $$\\ \text Thus:-$$  \quad
 \left[
    \begin{aligned}\begin{tikzpicture}[scale=0.7]
          \node (0a) at (-0.5,0) {};
          \node (0b) at (0.5,0) {};
          \node[ls] (1) at (0,1) {};
          \node[ls] (2) at (0,2) {};
          \node (3a) at (-0.5,3) {};
          \node (3b) at (0.5,3) {};
          \draw[string,out=90,in=180] (0a) to (1);
          \draw[string,out=90,in=0] (0b) to (1);
          \draw[string] (1) to (2);
          \draw[string,out=180,in=270] (2) to (3a);
          \draw[string,out=0,in=270] (2) to (3b);
          
      \end{tikzpicture}\end{aligned}
    \quad = \quad
    \begin{aligned}\begin{tikzpicture}[scale=0.7]
          \node (0) at (0,0) {};
          \node (0a) at (0,1) {};
          \node[ls] (1) at (0.5,2) {};
          \node[ls] (2) at (1.5,1) {};
          \node (3) at (1.5,0) {};
          \node (4) at (2,3) {};
          \node (4a) at (2,2) {};
          \node (5) at (0.5,3) {};
          \draw[string] (0) to (0a.center);
          \draw[string,out=90,in=180] (0a.center) to (1);
          \draw[string,out=0,in=180] (1) to (2);
          \draw[string,out=0,in=270] (2) to (4a.center);
          \draw[string] (4a.center) to (4);
          \draw[string] (2) to (3);
          \draw[string] (1) to (5);
                    
      \end{tikzpicture}\end{aligned}
      \right]
    \quad  \Rightarrow \quad
    \left[
 \begin{aligned}\begin{tikzpicture}[yscale=210/175,xscale=0.6]
          \node (0a) at (-1,0.25) {};
          \node (0b) at (-0.5,0.25) {};
          \node (0c) at (0.5,0.25) {};

          \node[ls] (1) at (0,1) {};
          \node[ls] (2) at (-0.5,1.5) {};

          \node (5a) at (-0.5,2) {};

          \draw[string, out=90, in =180] (0a) to (2);
          \draw[string, out=0, in=90] (2) to (1);
          \draw[string, out=0, in=90] (1) to (0c);
          \draw[string, out=180, in=90] (1) to (0b);

          \draw[string] (2) to (5a.center);
         
          \end{tikzpicture}\end{aligned}   
  \quad   = \quad
\begin{aligned}\begin{tikzpicture}[yscale=210/175,xscale=0.6]
          \node (0a) at (0.75,0.25) {};
          \node (0b) at (0.25,0.25) {};
          \node (0c) at (-0.75,0.25) {};

          \node[ls] (1) at (-0.25,1) {};
          \node[ls] (2) at (0.25,1.5) {};

          \node (5a) at (0.25,2) {};

          \draw[string, out=90, in =00] (0a) to (2);
          \draw[string, out=180, in=90] (2) to (1);
          \draw[string, out=180, in=90] (1) to (0c);
          \draw[string, out=0, in=90] (1) to (0b);

          \draw[string] (2) to (5a.center);
\end{tikzpicture}\end{aligned}  
\right]
  \end{equation}


Now assume that our latin square structure is associative:
\begin{equation} \left[
    \begin{aligned}\begin{tikzpicture}[yscale=147/120,xscale=0.7]
          \node (0a) at (-0.5,0) {};
          \node (0b) at (0.5,0) {};
          \node[ls] (1) at (0,1) {};
          \node (2) at (0,2) {};

          \draw[string,out=90,in=180] (0a.center) to (1);
          \draw[string,out=90,in=0] (0b.center) to (1);
          \draw[string] (1) to (2.center);
          \draw[string] (1,0) to (1,2);
          
      \end{tikzpicture}\end{aligned}
    \quad = \quad
    \begin{aligned}\begin{tikzpicture}[yscale=147/120,xscale=0.7]
          \node (0a) at (-0.5,0) {};
          \node (0b) at (0.5,0) {};
          \node[ls] (1) at (0,1) {};
          \node (2) at (0,2) {};

          \draw[string,out=90,in=180] (0a.center) to (1);
          \draw[string,out=90,in=0] (0b.center) to (1);
          \draw[string] (1) to (2.center);
          \draw[string] (1,0) to (1,2);
      \end{tikzpicture}\end{aligned}
       \right]
    \quad  \overset{\text{(LS2)}}{\Leftrightarrow} \quad
 \left[
    \begin{aligned}\begin{tikzpicture}[scale=49/60]
          \node (0a) at (0,0) {};
          \node (0b) at (1,0) {};
          \node (0c) at (2,0) {};
          \node[ls] (w1) at (1,0.5) {};
          \node[ls] (w2) at (1,2) {};
          \node[ls] (w3) at (0.5,2.5) {};
          \node[blackdot] (b1) at (1.5,1) {};
          \node[blackdot] (b2) at (1.5,1.5) {}; 
          \node (1a) at (0.5,3) {};
          \node (1b) at (2,3) {}; 
          \node (8) at (0,2) {};
          \draw[string,out=90,in=270] (0a.center) to (8.center);
          \draw[string,out=90,in=180] (8.center) to (w3);
          \draw[string] (w3) to (1a.center);
          \draw[string,out=0,in=90] (w3) to (w2);
          \draw[string,out=180,in=180] (w2) to (w1);
          \draw[string] (w1) to (0b.center);
          \draw[string,out=0,in=180] (w1) to (b1);
          \draw[string,out=0,in=90] (b1) to (0c.center);
          \draw[string,out=90,in=270] (b1) to (b2);
          \draw[string,out=180,in=0] (b2) to (w2);
          \draw[string,out=0,in=270] (b2) to (1b.center);

      \end{tikzpicture}\end{aligned}
    \quad \overset{\text{by assumption}}{=} \quad
        \begin{aligned}\begin{tikzpicture}[scale=49/60]
          \node (0a) at (-0.5,0) {};
          \node (0b) at (1,0) {};
          \node (0c) at (2,0) {};
          \node[ls] (w1) at (1,0.5) {};
          \node[ls] (w2) at (0.5,2.5) {};
          \node[ls] (w3) at (0,1.5) {};
          \node[blackdot] (b1) at (1.5,1) {};
          \node[blackdot] (b2) at (1.5,1.5) {}; 
          \node (1a) at (0.5,3) {};
          \node (1b) at (2,3) {}; 
          \node (8) at (-0.5,1) {}; 
          \draw[string,out=90,in=270] (0a.center) to (8.center);
          \draw[string,out=90,in=180] (8.center) to (w3);
          \draw[string] (w2) to (1a.center);
          \draw[string,out=0,in=180] (w3) to (w1);
          \draw[string,out=180,in=90] (w2) to (w3);
          \draw[string] (w1) to (0b.center);
          \draw[string,out=0,in=180] (w1) to (b1);
          \draw[string,out=0,in=90] (b1) to (0c.center);
          \draw[string,out=90,in=270] (b1) to (b2);
          \draw[string,out=180,in=0] (b2) to (w2);
          \draw[string,out=0,in=270] (b2) to (1b.center);

      \end{tikzpicture}\end{aligned}
    \quad = \quad
    \begin{aligned}\begin{tikzpicture}[scale=0.7]
          \node (0a) at (0.5,-0.5) {};
          \node (0b) at (1.5,-0.5) {};
          \node (0c) at (2,-0.5) {};
          \node[ls] (w1) at (1,0.5) {};
          \node[ls] (w2) at (1,2) {};
          \node[ls] (w3) at (1,0) {};
          \node[blackdot] (b1) at (1.5,1) {};
          \node[blackdot] (b2) at (1.5,1.5) {}; 
          \node (1a) at (1,3) {};
          \node (1b) at (2,3) {}; 
          
          \draw[string,out=90,in=180] (0a.center) to (w3);
     
          \draw[string] (w2) to (1a.center);
          \draw[string] (w3) to (w1);
          \draw[string,out=180,in=180] (w2) to (w1);
          \draw[string,out=0,in=90] (w3) to (0b.center);
          \draw[string,out=0,in=180] (w1) to (b1);
          \draw[string,out=0,in=90] (b1) to (0c.center);
          \draw[string,out=90,in=270] (b1) to (b2);
          \draw[string,out=180,in=0] (b2) to (w2);
          \draw[string,out=0,in=270] (b2) to (1b.center);
      \end{tikzpicture}\end{aligned}
       \right]    
\quad $$\\$$ \Rightarrow \quad
\left[
      \begin{aligned}\begin{tikzpicture}[scale=0.7]
          \node (0a) at (-0.5,0) {};
          \node (0b) at (1,0) {};
          \node (0c) at (2,0) {};
          \node[ls] (w1) at (1,0.5) {};
          \node[ls] (w2) at (0.5,2.5) {};
          \node[ls] (w3) at (0,1.5) {};
          \node[blackdot] (b1) at (1.5,1) {};
          \node[blackdot] (b2) at (1.5,1.5) {}; 
          \node[blackdot] (b3) at (1.5,3.5) {};
          \node[ls] (1a) at (0.5,3) {};
          \node (2a) at (0,4) {};
          \node (2b) at (1.5,4) {}; 
          \node (8) at (-0.5,1) {}; 
          \draw[string,out=90,in=270] (0a.center) to (8.center);
          \draw[string,out=90,in=180] (8.center) to (w3);
          \draw[string] (w2) to (1a);
          \draw[string,out=0,in=180] (w3) to (w1);
          \draw[string,out=180,in=90] (w2) to (w3);
          \draw[string] (w1) to (0b.center);
          \draw[string,out=0,in=180] (w1) to (b1);
          \draw[string,out=0,in=90] (b1) to (0c.center);
          \draw[string,out=90,in=270] (b1) to (b2);
          \draw[string,out=180,in=0] (b2) to (w2);
          \draw[string,out=0,in=0] (b2) to (b3);
          \draw[string] (b3) to (2b.center);
          \draw[string,out=180,in=270] (1a) to (2a.center);
          \draw[string,out=0,in=180] (1a) to (b3);

      \end{tikzpicture}\end{aligned}
    \quad = \quad
    \begin{aligned}\begin{tikzpicture}[scale=0.7]
          \node (0a) at (0.5,-0.5) {};
          \node (0b) at (1.5,-0.5) {};
          \node (0c) at (2,-0.5) {};
          \node[ls] (w1) at (1,0.5) {};
          \node[ls] (w2) at (1,2) {};
          \node[ls] (w3) at (1,0) {};
          \node[blackdot] (b1) at (1.5,1) {};
          \node[blackdot] (b2) at (1.5,1.5) {};
          \node[blackdot] (b3) at (1.5,3.5) {}; 
          \node[ls] (1a) at (1,3) {};
          \node (2a) at (0,4) {};
          \node (2b) at (1.5,4) {}; 
          
          \draw[string,out=90,in=180] (0a.center) to (w3);
          \draw[string] (b3) to (2b.center);     
          \draw[string] (w2) to (1a.center);
          \draw[string] (w3) to (w1);
          \draw[string,out=180,in=180] (w2) to (w1);
          \draw[string,out=0,in=90] (w3) to (0b.center);
          \draw[string,out=0,in=180] (w1) to (b1);
          \draw[string,out=0,in=90] (b1) to (0c.center);
          \draw[string,out=90,in=270] (b1) to (b2);
          \draw[string,out=180,in=0] (b2) to (w2);
          \draw[string,out=180,in=270] (1a) to (2a.center);
          \draw[string,out=0,in=0] (b2) to (b3);
          \draw[string,out=0,in=180] (1a) to (b3);
      \end{tikzpicture}\end{aligned}
       \right]          
\quad  \overset{(\text{LS2 both sides})}{\Leftrightarrow} \quad
\left[
      \begin{aligned}\begin{tikzpicture}[scale=0.7]
          \node (0a) at (-0.5,0) {};
          \node (0b) at (1,0) {};
          \node (0c) at (2,0) {};
          \node[ls] (w1) at (1,0.5) {};
          \node[ls] (w3) at (0,1.5) {};
          \node[blackdot] (b1) at (1.5,1) {};
          \node (2a) at (0,2) {};
          \node (2b) at (1.5,2) {}; 
          \node (8) at (-0.5,1) {}; 
          \draw[string,out=90,in=270] (0a.center) to (8.center);
          \draw[string,out=90,in=180] (8.center) to (w3);
          \draw[string,out=0,in=180] (w3) to (w1);
          \draw[string] (w1) to (0b.center);
          \draw[string,out=0,in=180] (w1) to (b1);
          \draw[string,out=0,in=90] (b1) to (0c.center);
          \draw[string] (b1) to (2b.center);
          \draw[string] (w3) to (2a.center);                    
          
      \end{tikzpicture}\end{aligned}
    \quad = \quad
    \begin{aligned}\begin{tikzpicture}[scale=0.7]
          \node (0a) at (0.5,-0.5) {};
          \node (0b) at (1.5,-0.5) {};
          \node (0c) at (2,-0.5) {};
          \node[ls] (w1) at (1,0.5) {};
          \node[ls] (w3) at (1,0) {};
          \node[blackdot] (b1) at (1.5,1) {};
          \node (2a) at (0,2) {};
          \node (2b) at (1.5,2) {}; 
          
          \draw[string,out=90,in=180] (0a.center) to (w3);
          \draw[string] (b1) to (2b.center);     
          \draw[string,out=0,in=90] (w3) to (0b.center);
          \draw[string] (w3) to (w1);
          \draw[string,out=0,in=180] (w1) to (b1);
          \draw[string,out=0,in=90] (b1) to (0c.center);
          \draw[string,out=180,in=270] (w1) to (2a.center);
      \end{tikzpicture}\end{aligned}
       \right] 
\quad $$\\$$ \Leftrightarrow \quad
\left[
      \begin{aligned}\begin{tikzpicture}[scale=0.7]
          \node (0a) at (-0.5,0) {};
          \node (0b) at (1,0) {};

          \node[ls] (w1) at (1,0.5) {};
          \node[ls] (w3) at (0,1.5) {};
          \node[blackdot] (b1) at (1.5,1) {};

          \node (2a) at (0,2) {};
          \node (2b) at (1.5,2) {}; 
          \node (8) at (-0.5,1) {}; 
          \draw[string,out=90,in=270] (0a.center) to (8.center);
          \draw[string,out=90,in=180] (8.center) to (w3);
          \draw[string,out=0,in=180] (w3) to (w1);
          \draw[string] (w1) to (0b.center);
          \draw[string,out=0,in=180] (w1) to (b1);

          \draw[string] (b1) to (2b.center);
          \draw[string] (w3) to (2a.center);   
          \draw[string,out=180,in=0] (2,0.5) to (b1);
          \node[blackdot](b2) at (2,0.5) {};  
          \node (2c) at (2.5,2) {};
          \draw[string,out=0,in=270] (b2) to (2c.center);                          
          
      \end{tikzpicture}\end{aligned}
    \quad = \quad
    \begin{aligned}\begin{tikzpicture}[scale=0.7]
          \node (0a) at (0.5,-0.5) {};
          \node (0b) at (1.5,-0.5) {};

          \node[ls] (w1) at (1,0.5) {};
          \node[ls] (w3) at (1,0) {};
          \node[blackdot] (b1) at (1.5,1) {};
          \node (2a) at (0,2) {};
          \node (2b) at (1.5,2) {}; 
          
          \draw[string,out=90,in=180] (0a.center) to (w3);
          \draw[string] (b1) to (2b.center);     
          \draw[string,out=0,in=90] (w3) to (0b.center);
          \draw[string] (w3) to (w1);
          \draw[string,out=0,in=180] (w1) to (b1);
          \draw[string,out=0,in=180] (b1) to (2,0.5);
          \draw[string,out=180,in=270] (w1) to (2a.center);

          \node[blackdot] (2) at (2,0.5) {};  
          \node (2c) at (2.5,2) {};
          \draw[string,out=0,in=270] (2) to (2c.center);  
      \end{tikzpicture}\end{aligned}
       \right] 
\quad \Leftrightarrow \quad
\left[
      \begin{aligned}\begin{tikzpicture}[scale=0.7]
          \node (0a) at (-0.5,0) {};
          \node (0b) at (1,0) {};

          \node[ls] (w1) at (1,0.5) {};
          \node[ls] (w3) at (0,1.5) {};
          \node[blackdot] (b1) at (1.5,1) {};

          \node (2a) at (0,2) {};
          \node (2b) at (2,2) {}; 
          \node (8) at (-0.5,1) {}; 
          \draw[string,out=90,in=270] (0a.center) to (8.center);
          \draw[string,out=90,in=180] (8.center) to (w3);
          \draw[string,out=0,in=180] (w3) to (w1);
          \draw[string] (w1) to (0b.center);
          \draw[string,out=0,in=180] (w1) to (b1);

          \draw[string] (2,1.5) to (2b.center);
          \draw[string] (w3) to (2a.center);   
          \draw[string,out=180,in=0] (2,0.5) to (b1);
          \node[blackdot](b2) at (2,0.5) {};  
          \node[blackdot] (2c) at (2,1.5) {};
          \draw[string,out=0,in=0] (b2) to (2c.center);                          \draw[string,out=90,in=180] (b1) to (2c);
          
      \end{tikzpicture}\end{aligned}
    \quad = \quad
    \begin{aligned}\begin{tikzpicture}[scale=0.7]
          \node (0a) at (0.5,-0.5) {};
          \node (0b) at (1.5,-0.5) {};

          \node[ls] (w1) at (1,0.5) {};
          \node[ls] (w3) at (1,0) {};
          \node[blackdot] (b1) at (1.5,1) {};
          \node (2a) at (0,2) {};
          \node (2b) at (1.5,2) {}; 
          \node[blackdot] (A) at (2,1.5) {};
          \draw[string,out=90,in=180] (0a.center) to (w3);
          \draw[string,in=180,out=90] (b1) to (A);     
          \draw[string,out=0,in=90] (w3) to (0b.center);
          \draw[string] (w3) to (w1);
          \draw[string,out=0,in=180] (w1) to (b1);
          \draw[string,out=0,in=180] (b1) to (2,0.5);
          \draw[string,out=180,in=270] (w1) to (2a.center);

          \node[blackdot] (2) at (2,0.5) {};  
          \node (2c) at (2.5,2) {};
          \draw[string,out=0,in=0] (2) to (A);  
          \node (B) at (2,2) {};
          \draw[string] (A) to (B.center);
      \end{tikzpicture}\end{aligned}
       \right]
\quad   $$\\$$ \overset{(\text{Bl,SM and S})}{\Leftrightarrow} \quad
\left[
    \begin{aligned}\begin{tikzpicture}[scale=0.6]
          \node (0a) at (-0.5,0) {};
          \node (0b) at (0.5,0) {};
          \node[ls] (1) at (0,1) {};
          \node[ls] (2) at (0,2) {};
          \node (3a) at (-0.5,3) {};
          \node (3b) at (0.5,3) {};
          \draw[string,out=90,in=180] (0a) to (1);
          \draw[string,out=90,in=0] (0b) to (1);
          \draw[string] (1) to (2);
          \draw[string,out=180,in=270] (2) to (3a);
          \draw[string,out=0,in=270] (2) to (3b);
          
      \end{tikzpicture}\end{aligned}
    \quad = \quad
    \begin{aligned}\begin{tikzpicture}[scale=0.6]
          \node (0) at (0,0) {};
          \node (0a) at (0,1) {};
          \node[ls] (1) at (0.5,2) {};
          \node[ls] (2) at (1.5,1) {};
          \node (3) at (1.5,0) {};
          \node (4) at (2,3) {};
          \node (4a) at (2,2) {};
          \node (5) at (0.5,3) {};
          \draw[string] (0) to (0a.center);
          \draw[string,out=90,in=180] (0a.center) to (1);
          \draw[string,out=0,in=180] (1) to (2);
          \draw[string,out=0,in=270] (2) to (4a.center);
          \draw[string] (4a.center) to (4);
          \draw[string] (2) to (3);
          \draw[string] (1) to (5);
                    
      \end{tikzpicture}\end{aligned}
      \right]\end{equation}\end{proof}


\chapter{Graphical Shift and Multiply Basis}
In Chapter 2 we found that a minimal shift and multiply basis $E_{ij}=P_j \circ H_{\text{diag}(i)}$
where 
\begin{equation}
P_j
\quad=\quad
\sum_{k=0}^{d-1}
\begin{aligned}\begin{tikzpicture}[yscale=0.75,xscale=0.65]
          \node (0w) at (0,-0.75) {};
          \node (0b) at (0,1.75) {};
          \node (w)[state,hflip,black,scale=0.5] at (0,0) {$k$};
          \node (b)[state,black,scale=0.5] at(0,1) {$k+j$};

          \draw[string] (0w) to (w);
          \draw[string] (0b) to (b);        
      \end{tikzpicture}\end{aligned}
      \quad=\quad
 \begin{aligned}\begin{tikzpicture}
\node (0a) at (0.75,-1) {};
         
          \node (j)[state,black,scale=0.5] at(2,0) {$j$};
          \node (2a)[agg] at (1.5,0.5) {};
          \node (2b) at (1.5,1.5) {};
          
          \draw[string,out=90,in=180] (0a) to (2a);

          \draw[string,out=90,in=0] (j) to (2a);
          \draw[string,out=90,in=270] (2a) to (2b);
         
      \end{tikzpicture}\end{aligned}      
\end{equation} 
and 
\begin{equation}
H_{\text{diag}(i)}
\quad=\quad
\sqrt{d} \sum_{k=0}^{d-1}
\begin{aligned}\begin{tikzpicture}[yscale=0.75,xscale=0.65]
          \node (0w) at (0,-0.75) {};
          \node (0b) at (0,1.75) {};
          \node (w)[state,hflip,black,scale=0.5] at (0,0) {$k$};
          \node (b)[state,black,scale=0.5] at(0,1) {$k$};
          \node (b2)[state,black,hflip,scale=0.5] at(-0.75,0.5) {$k$};
          \node (b3)[state,scale=0.5] at(-0.75,0.5) {$i$};

          \draw[string] (0w) to (w);
          \draw[string] (0b) to (b);        
      \end{tikzpicture}\end{aligned}     
 \quad=\quad
\sqrt{d} \begin{aligned}\begin{tikzpicture}
\node (0a) at (0.75,-1) {};
         
          \node (i)[state,scale=0.5] at(-0.5,0) {$i$};
          \node (2a)[blackdot] at (0,0.5) {};
          \node (2b) at (0,1.5) {};
          
          \draw[string,out=90,in=0] (0a) to (2a);

          \draw[string,out=90,in=180] (i) to (2a);
          \draw[string,out=90,in=270] (2a) to (2b);
         
      \end{tikzpicture}\end{aligned}
      \end{equation} 
The difference now is that our main classical structure will have an ONB of copyable states taken to be the elements of a loop rather than an abelian group. We also have the more general latin square structure \tinymultls, replacing the $G$-Frobenius algebra $\tinymultagg$, and a family of $d$ Hadamard matrices giving us $d$ different white orthonormal bases with corresponding white classical structures (all complementary to our main classical structure  from which the others are obtained via the Hadamard matrices and then normalised). We now have no relationship between the latin square and the Hadamard matrices.  I will denote the $i^{th}$ basis state of the $j^{th}$ white ONB (corresponding to the $j^{th}$ Hadamard matrix) by: 
$\begin{pic} 
\node[state,scale=0.5,label={[yshift=-0.3cm]0:{\tiny $j$}}] (A) at (0,0) {$i$};
\draw[string] (A) to (0,0.5);
\end{pic}$      

So we have a shift multiply basis $S_{ij}=P_j \circ H^j_{\text{diag}(i)}$ where:
\begin{equation}
P_j
\quad=\quad
 \begin{aligned}\begin{tikzpicture}
\node (0a) at (0.75,-1) {};
         
          \node (j)[state,black,scale=0.5] at(2,0) {$j$};
          \node (2a)[ls] at (1.5,0.5) {};
          \node (2b) at (1.5,1.5) {};
          
          \draw[string,out=90,in=180] (0a) to (2a);

          \draw[string,out=90,in=0] (j) to (2a);
          \draw[string,out=90,in=270] (2a) to (2b);
         
      \end{tikzpicture}\end{aligned}
\end{equation}
and
\begin{equation}
H^j_{\text{diag}(i)}
\quad=\quad
\sqrt{d} \begin{aligned}\begin{tikzpicture}
\node (0a) at (0.75,-1) {};
         
          \node (i)[state,scale=0.5,label={[yshift=-0.3cm]0:{\tiny $j$}}] at(-0.5,0) {$i$};
          \node (2a)[blackdot] at (0,0.5) {};
          \node (2b) at (0,1.5) {};
          
          \draw[string,out=90,in=0] (0a) to (2a);

          \draw[string,out=90,in=180] (i) to (2a);
          \draw[string,out=90,in=270] (2a) to (2b);
         
      \end{tikzpicture}\end{aligned}
      \end{equation} 
Thus
\begin{equation}
S_{ij}=
\sqrt{d} \begin{aligned}\begin{pic}

      \node (0a) at (0.75,-1) {};         
      \node (i)[state,scale=0.5,label={[yshift=-0.3cm]0:{\tiny $j$}}] at(-0.5,0) {$i$};
          \node (2a)[blackdot] at (0,0.5) {};
          \node (2b) at (0,0.75) {};          
          \draw[string,out=90,in=0] (0a) to (2a);                   
          \draw[string,out=90,in=180] (i) to (2a);
          \draw[string,out=90,in=270] (2a) to (2b.center);
          
          \node (0b) at (0,0.5) {};
          \node (j)[state,black,scale=0.5] at(1.25,1.5) {$j$};
          \node (2a')[ls] at (0.75,2) {};
          \node (2b') at (0.75,3) {};          
          \draw[string,out=90,in=180] (2b.center) to (2a');                  
          \draw[string,out=90,in=0] (j) to (2a');
          \draw[string,out=90,in=270] (2a') to (2b');
\end{pic}\end{aligned}
\end{equation}
I am now ready to present a fully graphical version of Werner's proof of the correctness of the combinatorial construction based on my axiomatisation of a latin square structure.
\begin{theorem}
Shift and multiply bases are unitary error bases.
\end{theorem}
\begin{proof}
\textit{Unitarity}

As proven in Chapter 2, $H_{\text{diag}(i)}$ is unitary. $H^j_{\text{diag}(i)} $ is unitary in exactly the same way. $P_j$ is unitary iff the rule LS2 holds, as shown below:
\begin{equation}
\forall j \left[
\begin{aligned}\begin{tikzpicture}
          
          \node (0a) at (-1,-1.5) {};
          \node (0b) at (-0.5,-0.5) {};
          \node[state,black,scale=0.5] (0c) at (0,-0.5) {$j$};
          \node[ls] (7) at (-0.5,0) {};       
          \node (10a) at (-0.5,1.5) {};       
          \draw[string, out=90, in =180] (0a.center) to (7);
          \draw[string, out=0, in=90] (7) to (0c);
          \draw[string] (7) to (10a.center);
         
          \end{tikzpicture}\end{aligned}
unitary \right]
\quad \Rightarrow \quad
\forall j
 \left[
\begin{aligned}\begin{pic}[yscale=0.75]
\node[ls] (1) at (0,0) {};
\node[ls] (2) at (0,1) {};
\node[state,black,scale=0.5] (b) at (0.5,-1) {$j$};
\node[state,hflip,black,scale=0.5] (a) at (0.5,2) {$j$};
\draw[string] (1) to (2);
\draw[string,out=0,in=90] (1) to (b);
\draw[string,out=0,in=270] (2) to (a);
\draw[string,out=90,in=180] (-0.5,-1) to (1);
\draw[string,out=270,in=180] (-0.5,2) to (2);

\end{pic}\end{aligned}
\quad=\quad
\begin{aligned}\begin{pic}[yscale=0.75]
\draw[string] (0,-1) to (0,2);
\end{pic}\end{aligned}
\, \, \right]
\quad $$\\$$ \overset{\text{by proof of proposition 3.2}}{\Leftrightarrow} \quad
 \left[
\begin{aligned}\begin{pic}[yscale=0.375,xscale=0.5]
 \node (0) at (0.5,-3.5) {};
\node[ls] (1) at (-1,-1) {};
\node (2) at (1.5,0) {};
\node[blackdot] (B) at (0.5,-2) {};
\node (3) at (-2,-3.5) {};
\node (4) at (-1,0) {};
\draw[string,out=180,in=90] (1) to (3);
\draw [string, out=90,in=270] (0) to (B);
\draw [string, out=180,in=0] (B) to (1);
\draw [string, out=0,in=270] (B) to (2.center);
\draw[string] (1) to (4.center);   

 \node (5) at (0.5,3.5) {};
\node[ls] (6) at (-1,1) {};

\node[blackdot] (A) at (0.5,2) {};
\node (8) at (-2,3.5) {};

\draw[string,out=180,in=270] (6) to (8);
\draw [string, out=270,in=90] (5) to (A);
\draw [string, out=180,in=0] (A) to (6);
\draw [string, out=0,in=90] (A) to (2.center);
\draw[string] (6) to (4.center);  

\end{pic}\end{aligned} 
\quad=\quad
\begin{aligned}\begin{pic}[yscale=0.75]
\draw[string] (0,-1) to (0,2);
\node (a) at (0.75,-1) {};
\node (b) at (0.75,2) {};
\draw[string] (a.center) to (b.center);
\end{pic}\end{aligned} 
\, \, \right]          
\end{equation} 
The other direction straightforwardly follows in a similar way.
Thus $S_{ij}$ is composed of two unitaries and hence is unitary.\\

\vspace{10 mm}
\textit{orthogonality}\\

\begin{equation}
\tr (S^{\dag}_{ij} \circ S_{i'j'})=
d \begin{aligned}\begin{pic}[scale=0.6]

      \node (0a) at (1.5,-3.5) {};         
      \node (i)[state,scale=0.5,label={[yshift=-0.3cm]0:{\tiny $j'$}}] at(0.5,-3) {$i'$};
          \node (2a)[blackdot] at (1,-2.5) {};
          \node (2b) at (1,-2.25) {};          
          \draw[string,out=90,in=0] (0a.center) to (2a);                   
          \draw[string,out=90,in=180] (i) to (2a);
          \draw[string,out=90,in=270] (2a) to (2b.center);
          
          \node (0b) at (0,-1.5) {};
          \node (j)[state,black,scale=0.5] at(2.25,-1.5) {$j'$};
          \node (2a')[ls] at (1.75,-1) {};
          \node (2b') at (1.75,0) {};          
          \draw[string,out=90,in=180] (2b.center) to (2a');                  
          \draw[string,out=90,in=0] (j) to (2a');

           \node (A) at (1.5,3.5) {};         
      \node (B)[state,hflip,scale=0.5,label={[yshift=0.3cm]0:{\tiny $j$}}] at(0.5,3) {$i$};
          \node (C)[blackdot] at (1,2.5) {};
          \node (D) at (1,2.25) {};          
          \draw[string,out=270,in=0] (A.center) to (C);                   
          \draw[string,out=270,in=180] (B) to (C);
          \draw[string,out=270,in=90] (C) to (D.center);
          
          \node (E) at (0,1.5) {};
          \node (F)[state,black,hflip,scale=0.5] at(2.25,1.5) {$j$};
          \node (G)[ls] at (1.75,1) {};
         
          \draw[string,out=270,in=180] (D.center) to (G);                  
          \draw[string,out=270,in=0] (F) to (G);
          \draw[string] (G) to (2a');
          
          \node[blackdot] (X) at (-0.125,4.25) {};
          \node[blackdot] (Y) at (-0.125,-4.25) {};
          \draw[string,in=180,out=180] (X) to (Y);
          \draw[string,in=270,out=0] (Y) to (0a.center);
          \draw[string,in=90,out=0] (X) to (A.center);
\end{pic}\end{aligned}
\quad =\quad
d \begin{aligned}\begin{pic}[scale=0.6]

      \node (0a) at (1.5,-3.5) {};         
      \node (i)[state,scale=0.5,label={[yshift=-0.3cm]0:{\tiny $j'$}}] at(0.5,-4.75) {$i'$};
          \node (2a)[blackdot] at (1,-2.5) {};
          \node (2b) at (1,-2.25) {};          
          \draw[string,out=90,in=0] (0a.center) to (2a);                   
          \draw[string,out=90,in=180] (i) to (2a);
          \draw[string,out=90,in=270] (2a) to (2b.center);
          
          \node (0b) at (0,-1.5) {};
          \node (j)[state,black,scale=0.5] at(2.25,-1.5) {$j'$};
          \node (2a')[ls] at (1.75,-1) {};
          \node (2b') at (1.75,0) {};          
          \draw[string,out=90,in=180] (2b.center) to (2a');                  
          \draw[string,out=90,in=0] (j) to (2a');

           \node (A) at (1.5,3.5) {};         
      \node (B)[state,hflip,scale=0.5,label={[yshift=0.3cm]0:{\tiny $j$}}] at(0.5,4.75) {$i$};
          \node (C)[blackdot] at (1,2.5) {};
          \node (D) at (1,2.25) {};          
          \draw[string,out=270,in=0] (A.center) to (C);                   
          \draw[string,out=270,in=180] (B) to (C);
          \draw[string,out=270,in=90] (C) to (D.center);
          
          \node (E) at (0,1.5) {};
          \node (F)[state,black,hflip,scale=0.5] at(2.25,1.5) {$j$};
          \node (G)[ls] at (1.75,1) {};
         
          \draw[string,out=270,in=180] (D.center) to (G);                  
          \draw[string,out=270,in=0] (F) to (G);
          \draw[string] (G) to (2a');
          
          \node[blackdot] (X) at (-0.125,4.25) {};
          \node[blackdot] (Y) at (-0.125,-4.25) {};
          \draw[string,in=180,out=180] (X) to (Y);
          \draw[string,in=270,out=0] (Y) to (0a.center);
          \draw[string,in=90,out=0] (X) to (A.center);
\end{pic}\end{aligned}
\quad $$\\$$ \overset{2 \times \text{(Bl,SM)}}{=}\quad
d \begin{aligned}\begin{pic}[yscale=0.375,xscale=0.5]
 \node[state,scale=0.5,label={[yshift=-0.3cm]0:{\tiny $j'$}}] (0) at (-1.5,-3.5) {$i'$};
\node[ls] (1) at (0,-1) {};
\node (2) at (-2.5,0) {};
\node[blackdot] (B) at (-1.5,-2) {};
\node[state,black,scale=0.5] (3) at (1,-3.5) {$j'$};
\node (4) at (0,0) {};
\draw[string,out=0,in=90] (1) to (3);
\draw [string, out=90,in=270] (0) to (B);
\draw [string, out=0,in=180] (B) to (1);
\draw [string, out=180,in=270] (B) to (2.center);
\draw[string] (1) to (4.center);   

 \node[state,hflip,scale=0.5,label={[yshift=0.3cm]0:{\tiny $j$}}] (5) at (-1.5,3.5) {$i$};
\node[ls] (6) at (0,1) {};

\node[blackdot] (A) at (-1.5,2) {};
\node[state,hflip,black,scale=0.5] (8) at (1,3.5) {$j$};

\draw[string,out=0,in=270] (6) to (8);
\draw [string, out=270,in=90] (5) to (A);
\draw [string, out=0,in=180] (A) to (6);
\draw [string, out=180,in=90] (A) to (2.center);
\draw[string] (6) to (4.center);  

\end{pic}\end{aligned}
\quad \overset{2 \times \text{(LS1)}}{=} \quad
\begin{aligned}\begin{pic}[yscale=0.375,xscale=0.5]
 \node[state,scale=0.5,label={[yshift=-0.3cm]0:{\tiny $j'$}}] (0) at (-1.5,3.5) {$i'$};
\node[state,black,scale=0.5] (3) at (0,3.5) {$j'$}; 
 \node[state,hflip,scale=0.5,label={[yshift=0.3cm]0:{\tiny $j$}}] (5) at (-1.5,3.5) {$i$};
\node[state,hflip,black,scale=0.5] (8) at (0,3.5) {$j$};

\end{pic}\end{aligned} d
\quad=\quad
\begin{aligned}\begin{pic}[yscale=0.375,xscale=0.5]
 \node[state,scale=0.5,label={[yshift=-0.3cm]0:{\tiny $j'$}}] (0) at (-1.5,3.5) {$i'$};
 
 \node[state,hflip,scale=0.5,label={[yshift=0.3cm]0:{\tiny $j$}}] (5) at (-1.5,3.5) {$i$};

\end{pic}\end{aligned}
\delta_{jj'}d
=\delta_{ii'} \delta_{jj'}d \end{equation}
\end{proof}


\chapter{Generalised Shift and Multiply Basis}
In Chapter 3, I axiomatised latin squares categorically. I then used this axiomatisation to derive a graphical representation of a shift and multiply basis in Chapter 4. Using the axioms of a latin square structure, I then proved the correctness of the combinatorial construction. I am now ready to derive which axioms are surplus to the requirements of a UEB. 

Recall that a shift and multiply basis is represented as follows:
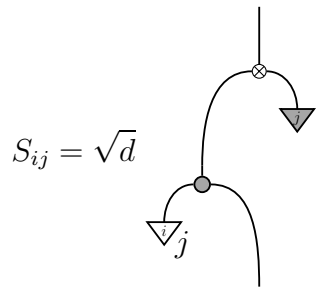
\begin{equation*}S_{ij}=
\sqrt{d} \begin{aligned}\begin{pic}

      \node (0a) at (0.75,-1) {};         
      \node (i)[state,scale=0.5,label={[yshift=-0.3cm]0:$j$}] at(-0.5,0) {$i$};
          \node (2a)[blackdot] at (0,0.5) {};
          \node (2b) at (0,0.75) {};          

          \draw[string,out=90,in=0] (0a) to (2a);                   
          \draw[string,out=90,in=180] (i) to (2a);
          \draw[string,out=90,in=270] (2a) to (2b.center);
          
          \node (0b) at (0,0.5) {};
          \node (j)[state,black,scale=0.5] at(1.25,1.5) {$j$};
          \node (2a')[ls] at (0.75,2) {};
          \node (2b') at (0.75,3) {};          
          \draw[string,out=90,in=180] (2b.center) to (2a');                  
          \draw[string,out=90,in=0] (j) to (2a');
          \draw[string,out=90,in=270] (2a') to (2b');
\end{pic}\end{aligned}\end{equation*}

To ensure that my proof of theorem 4.1 goes through, we can see that any candidate to replace the latin square structure will have to obey the rules LS1 and LS2. 
i.e.
\begin{equation*} U_1=
 \begin{pic}[scale=0.2]
\node (0) at (0.5,-0.5) {};
\node[blackdot,scale=0.7] (1) at (-1,2) {};
\node (2) at (1.5,3.5) {};
\node[ls,scale=0.7] (B) at (0.5,1) {};
\node (3) at (-2,-0.5) {};
\node (4) at (-1,3.5) {};
\draw[string,out=180,in=90] (1) to (3);
\draw [string, out=90,in=270] (0) to (B);
\draw [string, out=180,in=0] (B) to (1);
\draw [string, out=0,in=270] (B) to (2);
\draw[string] (1) to (4);
\end{pic}\text{ and } U_2=
 \begin{pic}[scale=0.2]
\node (0) at (0.5,-0.5) {};
\node[ls,scale=0.7] (1) at (-1,2) {};
\node (2) at (1.5,3.5) {};
\node[blackdot,scale=0.7] (B) at (0.5,1) {};
\node (3) at (-2,-0.5) {};
\node (4) at (-1,3.5) {};
\draw[string,out=180,in=90] (1) to (3);
\draw [string, out=90,in=270] (0) to (B);
\draw [string, out=180,in=0] (B) to (1);
\draw [string, out=0,in=270] (B) to (2);
\draw[string] (1) to (4);
\end{pic}\text{are unitary.} 
\end{equation*}
But that is all that we require. 

\begin{definition}[Generalised Latin Square Structure]
Given a finite dimensional Hilbert space $H$, a generalised latin square structure is a linear map $\tinymultlss: H \otimes H \rightarrow H$ and a linear map $\tinycomultlss: H \rightarrow H \otimes H$ such that: 
\begin{equation}
U_1=
 \begin{pic}[scale=0.3]
\node (0) at (0.5,-0.5) {};
\node[blackdot,scale=0.7] (1) at (-1,2) {};
\node (2) at (1.5,3.5) {};
\node[ls',scale=0.7] (B) at (0.5,1) {};
\node (3) at (-2,-0.5) {};
\node (4) at (-1,3.5) {};
\draw[string,out=180,in=90] (1) to (3);
\draw [string, out=90,in=270] (0) to (B);
\draw [string, out=180,in=0] (B) to (1);
\draw [string, out=0,in=270] (B) to (2);
\draw[string] (1) to (4);
\end{pic}  \text{ and }  U_2=
 \begin{pic}[scale=0.3]
\node (0) at (0.5,-0.5) {};
\node[ls',scale=0.7] (1) at (-1,2) {};
\node (2) at (1.5,3.5) {};
\node[blackdot,scale=0.7] (B) at (0.5,1) {};
\node (3) at (-2,-0.5) {};
\node (4) at (-1,3.5) {};
\draw[string,out=180,in=90] (1) to (3);
\draw [string, out=90,in=270] (0) to (B);
\draw [string, out=180,in=0] (B) to (1);
\draw [string, out=0,in=270] (B) to (2);
\draw[string] (1) to (4);
\end{pic}\text{are unitary}\end{equation}
for some classical structure $\tinymult[blackdot]$. I will call the relationship between these two structures `quasi-complementarity'. 
  \end{definition}
So please note in particular that \tinymultlss \, need not obey the (co)unitality (3.23), bialgebra (3.24), and duality relation (3.25)-(3.26) axioms of a latin square structure. 
\begin{definition}[Generalised Shift and Multiply Basis]
Let \tinymultlss \, be a generalised latin square structure and \tinymult[blackdot] a classical structure  quasi-complementary to it and suppose we have $d$ indexed white classical structures all complementary to $\tinymult[blackdot]$. Then $B_{ij}$ as defined below is a generalised shift and multiply basis for $0 \leq i,j < d$:
\begin{equation}
B_{ij}:=
\sqrt{d} \begin{aligned}\begin{pic}

      \node (0a) at (0.75,-1) {};         
      \node (i)[state,scale=0.5,label={[yshift=-0.2cm]0:{\tiny$j$}}] at(-0.5,0) {$i$};
          \node (2a)[blackdot] at (0,0.5) {};
          \node (2b) at (0,0.75) {};          

          \draw[string,out=90,in=0] (0a) to (2a);                   
          \draw[string,out=90,in=180] (i) to (2a);
          \draw[string,out=90,in=270] (2a) to (2b.center);
          
          \node (0b) at (0,0.5) {};
          \node (j)[state,black,scale=0.5] at(1.25,1.5) {$j$};
          \node (2a')[ls'] at (0.75,2) {};
          \node (2b') at (0.75,3) {};          
          \draw[string,out=90,in=180] (2b.center) to (2a');                  
          \draw[string,out=90,in=0] (j) to (2a');
          \draw[string,out=90,in=270] (2a') to (2b');
\end{pic}\end{aligned}
\end{equation}
Where the white state is the $i^{th}$ basis state of the $j$-indexed classical structure.
I will refer to this as the generalised combinatorial construction.
\end{definition}

\section{Seeking Non-trivial Models of Generalised Latin Square Structures}
Given a latin square structure \tinymultls, \, I want to find a modification of it which gives a generalised latin square structure that is not a latin square structure. Further, I want to find a model generalised latin square structure which produces a generalised shift and multiply basis that is not trivially equivalent to a shift and multiply basis.

 If we multiply our latin square structure by a phase like so:
\begin{equation}
\begin{aligned}\begin{tikzpicture}
          
          \node (0a) at (-1,-0.5) {};
          \node (0b) at (-0.5,-0.5) {};
          \node (0c) at (0,-0.5) {};
          \node[ls,scale=2] (7) at (-0.5,0) {};       
          \node (10a) at (-0.5,0.5) {};       
          \draw[string, out=90, in =180] (0a.center) to (7);
          \draw[string, out=0, in=90] (7) to (0c.center);
          \draw[string] (7) to (10a.center);
         
          \end{tikzpicture}\end{aligned}
          \quad:=\quad e^{i \theta}
          \begin{aligned}\begin{tikzpicture}
          
          \node (0a) at (-1,-0.5) {};
          \node (0b) at (-0.5,-0.5) {};
          \node (0c) at (0,-0.5) {};
          \node[ls] (7) at (-0.5,0) {};       
          \node (10a) at (-0.5,0.5) {};       
          \draw[string, out=90, in =180] (0a.center) to (7);
          \draw[string, out=0, in=90] (7) to (0c.center);
          \draw[string] (7) to (10a.center);
         
          \end{tikzpicture}\end{aligned}
\end{equation}
then $U_1$ and $U_2$ are clearly still unitary. And the bialgebra laws are violated. However, the resulting basis can be obtained from $S_{ij}$ by uniformly multiplying each row of each Hadamard matrix in our family of Hadamards by a phase.  This leads to an equivalent UEB. 

Another approach is to add some unitary matrix $F$ to the upper wire like so:
\begin{equation}
\begin{aligned}\begin{tikzpicture}
          
          \node (0a) at (-1,-0.5) {};
         
          \node (0c) at (0,-0.5) {};
          \node[circle,scale=1.3, minimum width=8pt, draw, inner sep=0pt, path picture={\draw (path picture bounding box.west) -- (path picture bounding box.north)  (path picture bounding box.south west) -- (path picture bounding box.north east)  (path picture bounding box.south) -- (path picture bounding box.east);}] (7) at (-0.5,0) {};       
          \node (10a) at (-0.5,0.5) {};       
          \draw[string, out=90, in =180] (0a.center) to (7);
          \draw[string, out=0, in=90] (7) to (0c.center);
          \draw[string] (7) to (10a.center);

          \end{tikzpicture}\end{aligned}
          \quad:=\quad 
          \begin{aligned}\begin{tikzpicture}
          
          \node (0a) at (-1,-0.5) {};
        
          \node (0c) at (0,-0.5) {};
          \node[ls] (7) at (-0.5,0) {};       
          \node (10a) at (-0.5,1) {};  
          \node[morphism,wedge,scale=0.5] (F) at (-0.5,0.5) {$F$};     
          \draw[string, out=90, in =180] (0a.center) to (7);
          \draw[string, out=0, in=90] (7) to (0c.center);
          \draw[string] (-0.5,0.65) to (10a.center);
          \draw[string] (-0.5,0.35) to (7);
         
          \end{tikzpicture}\end{aligned}
          \end{equation}
          Again this will leave $U_1$ and $U_2$ unitary but violate the bialgebra law; unfortunately we still end up with a trivially equivalent unitary error basis modified uniformly by composition with a unitary on the left. 

Yet another candidate is to compose by a unitary $D$ on the bottom left wire like so:
\begin{equation}
\begin{aligned}\begin{tikzpicture}
          
          \node (0a) at (-1,-0.5) {};
         
          \node (0c) at (0,-0.5) {};
          \node[circle,scale=1.3, minimum width=8pt, draw, inner sep=0pt, path picture={\draw  (path picture bounding box.south west) -- (path picture bounding box.north)  (path picture bounding box.north east) -- (path picture bounding box.south)  (path picture bounding box.north west) -- (path picture bounding box.east)  (path picture bounding box.south east) -- (path picture bounding box.west) ;}] (7) at (-0.5,0) {};       
          \node (10a) at (-0.5,0.5) {};       
          \draw[string, out=90, in =180] (0a.center) to (7);
          \draw[string, out=0, in=90] (7) to (0c.center);
          \draw[string] (7) to (10a.center);
         
          \end{tikzpicture}\end{aligned}
          \quad:=\quad 
          \begin{aligned}\begin{tikzpicture}
          
          \node (0a) at (-1,-1) {};
        
          \node (0c) at (0,-1) {};
          \node[ls] (7) at (-0.5,0) {};       
          \node (10a) at (-0.5,0.5) {};  
          \node[morphism,wedge,scale=0.5] (F) at (-1,-0.5) {$D$};     
          \draw[string] (0a.center) to (-1,-0.65);
          \draw[string,out=180,in=90] (7) to (-1,-0.35);
          \draw[string, out=0, in=90] (7) to (0c.center);
          \draw[string] (7) to (10a.center);
          \draw[string] (-0.5,0.15) to (7);
\end{tikzpicture}\end{aligned}
\end{equation} 
More care is necessary now. $U_2$ is immediately unitary, but to ensure that $U_1$ is unitary, I will choose $D$ such that it can be moved through a spider.

Let $D$ be a diagonal unitary matrix of size $d \times d$. Suppose that the diagonal entries of $D$ are given by $D_{kk}=e^{i \phi_k}$. Define a function $\theta : \{0,1,...,d-1\} \rightarrow \mathbb{R}$ such that $ \theta (k)= \phi_k$. Now graphically  $D$ can be represented as:
\begin{equation} 
D=
\sum^{d-1}_{k=0} e^{i \theta(k)}
\begin{aligned}\begin{pic}
\node[state,black,hflip,scale=0.5] (A) at (0,-0.35) {$k$};
\node[state,black,scale=0.5] (B) at (0,0.35) {$k$};
\draw[string] (0,-0.75) to (A);
\draw[string] (0,0.75) to (B);
\end{pic}\end{aligned}
\end{equation}
Since scalar factors can move freely around the diagram:
\begin{equation}
\forall, a,b,c \in A \, \, \, \,
 \begin{aligned}\begin{tikzpicture}
          
          \node[state,black,scale=0.5] (0a) at (-1,-1) {$a$};
        
          \node[state,black,scale=0.5] (0c) at (0,-1) {$b$};
          \node[blackdot] (7) at (-0.5,0) {};       
          \node[state,hflip,black,scale=0.5] (10a) at (-0.5,0.5) {$c$};  
          \node[morphism,wedge,scale=0.5] (F) at (-1,-0.5) {$D$};     
          \draw[string] (0a.center) to (-1,-0.65);
          \draw[string,out=180,in=90] (7) to (-1,-0.35);
          \draw[string, out=0, in=90] (7) to (0c.center);
          \draw[string] (7) to (10a.center);
          \draw[string] (-0.5,0.15) to (7);
         
          \end{tikzpicture}\end{aligned}
         \quad =\quad
 \begin{aligned}\begin{tikzpicture}
          
          \node[state,black,scale=0.5] (0a) at (-1,-0.5) {$a$};
        
          \node[state,black,scale=0.5] (0c) at (0,-0.5) {$b$};
          \node[blackdot] (7) at (-0.5,0) {};       
          \node[state,black,hflip,scale=0.5] (10a) at (-0.5,1) {$c$};  
          \node[morphism,wedge,scale=0.5] (F) at (-0.5,0.5) {$D$};     
          \draw[string, out=90, in =180] (0a.center) to (7);
          \draw[string, out=0, in=90] (7) to (0c.center);
          \draw[string] (-0.5,0.65) to (10a.center);
          \draw[string] (-0.5,0.35) to (7);
         
          \end{tikzpicture}\end{aligned}  
          \quad=\quad
\begin{aligned}\begin{tikzpicture}[xscale=-1]
          
          \node[state,black,scale=0.5] (0a) at (-1,-1) {$b$};
        
          \node[state,black,scale=0.5] (0c) at (0,-1) {$a$};
          \node[blackdot] (7) at (-0.5,0) {};       
          \node[state,black,hflip,scale=0.5] (10a) at (-0.5,0.5) {$c$};  
          \node[morphism,wedge,scale=0.5] (F) at (-1,-0.5) {$D$};     
          \draw[string] (0a.center) to (-1,-0.65);
          \draw[string,out=180,in=90] (7) to (-1,-0.35);
          \draw[string] (7) to (10a.center);
          \draw[string] (-0.5,0.15) to (7);
          \draw[string] (0c.center) to (0,-0.35);
          \draw[string,out=0,in=90] (7) to (0,-0.35);         
          \end{tikzpicture}\end{aligned}
          \quad=\quad 
\delta_{ab} \delta_{bc} e^{i \theta(a)}                
\end{equation}
Thus
\begin{equation}
 \begin{aligned}\begin{tikzpicture}
          
          \node (0a) at (-1,-1) {};
        
          \node (0c) at (0,-1) {};
          \node[blackdot] (7) at (-0.5,0) {};       
          \node (10a) at (-0.5,0.5) {};  
          \node[morphism,wedge,scale=0.5] (F) at (-1,-0.5) {$D$};     
          \draw[string] (0a.center) to (-1,-0.65);
          \draw[string,out=180,in=90] (7) to (-1,-0.35);
          \draw[string, out=0, in=90] (7) to (0c.center);
          \draw[string] (7) to (10a.center);
          \draw[string] (-0.5,0.15) to (7);
         
          \end{tikzpicture}\end{aligned}
         \quad =\quad
 \begin{aligned}\begin{tikzpicture}
          
          \node (0a) at (-1,-0.5) {};
        
          \node (0c) at (0,-0.5) {};
          \node[blackdot] (7) at (-0.5,0) {};       
          \node (10a) at (-0.5,1) {};  
          \node[morphism,wedge,scale=0.5] (F) at (-0.5,0.5) {$D$};     
          \draw[string, out=90, in =180] (0a.center) to (7);
          \draw[string, out=0, in=90] (7) to (0c.center);
          \draw[string] (-0.5,0.65) to (10a.center);
          \draw[string] (-0.5,0.35) to (7);
         
          \end{tikzpicture}\end{aligned}  
          \quad=\quad
\begin{aligned}\begin{tikzpicture}[xscale=-1]
          
          \node (0a) at (-1,-1) {};
        
          \node (0c) at (0,-1) {};
          \node[blackdot] (7) at (-0.5,0) {};       
          \node (10a) at (-0.5,0.5) {};  
          \node[morphism,wedge,scale=0.5] (F) at (-1,-0.5) {$D$};     
          \draw[string] (0a.center) to (-1,-0.65);
          \draw[string,out=180,in=90] (7) to (-1,-0.35);
          \draw[string] (7) to (10a.center);
          \draw[string] (-0.5,0.15) to (7);
          \draw[string] (0c.center) to (0,-0.35);
          \draw[string,out=0,in=90] (7) to (0,-0.35);         
          \end{tikzpicture}\end{aligned}
\end{equation}
For any diagonal unitary matrix $D$. Similarly $D$ can move from one leg of any spider to another.

Now we can see that equation (5.5), with $D$ a diagonal unitary matrix satisfies conditions LS1 and LS2. However, the adaptation again turns out to be trivial as $H^j_{\text{diag}(i)}$ is also diagonal and thus commutes with $D$, thus producing a UEB differing from $S_{ij}$ only by uniform composition by a unitary matrix on the right. Again we have a trivially equivalent unitary error basis.

So in order to create a potentially non-equivalent generalised shift and multiply basis we need to choose a \tinymultlss \, that varies with $i$ or $j$ or both $i$ and $j$, and has the properties above to ensure that $U_1$ and $U_2$ are unitary. 
I propose the following model:
\begin{equation}
\begin{aligned}\begin{tikzpicture}
          
          \node (0a) at (-1,-0.5) {};
          \node (0b) at (-0.5,-0.5) {};
          \node (0c) at (0,-0.5) {};
          \node[ls'] (7) at (-0.5,0) {};       
          \node (10a) at (-0.5,0.5) {};       
          \draw[string, out=90, in =180] (0a.center) to (7);
          \draw[string, out=0, in=90] (7) to (0c.center);
          \draw[string] (7) to (10a.center);
         
          \end{tikzpicture}\end{aligned}
          \quad:=\quad
          \sum_{m}
\begin{aligned}\begin{pic}
         \node[ls] (l) at (-0.25,3.25) {};
         \node[whitedot,label={[xshift=-0.1cm,yshift=-0.1cm]0:{\tiny$k$}}] (k) at (-1,1.5) {};
         \node[blackdot] (b) at (0.5,0.75) {};
         \node[state,hflip,scale=0.5,label={[xshift=0.15cm]0:{\tiny$k$}}] (k1) at (-1,1.75) {$m$};
         \node[state,black,scale=0.5] (m1) at (-1,2.5) {$m$};
         \node[state,hflip,black,scale=0.5] (m2) at (-1.5,0.25) {$m$};
         \node[state,scale=0.5,label={[yshift=-0.2cm]0:{\tiny$k$}}] (k2) at (-1.5,1) {$m$};
         \node (0a) at (-1.5,-0.25){};
         \node (0b) at (0.5,-0.25){};
         \node (0c) at (0.5,2.5) {};
         \node (0d) at (-0.25,3.75){};
         \draw[string] (0a) to (m2);
         \draw[string,out=90,in=180] (k2) to (k);
         \draw[string] (k) to (k1);
         \draw[string,out=90,in=180] (m1) to (l);
         \draw[string,out=0,in=0] (l) to (b);
         \draw[string] (b) to (0b);
         \draw[string,out=180,in=0] (b) to (k);
         \draw[string] (l) to (0d.center);
                           
\end{pic}\end{aligned} \end{equation}
Where \tinymultk is the classical structure corresponding to the ONB obtained by applying the $k^{th}$ Hadamard matrix to the black ONB and then normalising.
Now  we obtain generalised shift and multiply basis:
\begin{equation} B'_{ij}:=
\sqrt{d} \begin{aligned}\begin{pic}

      \node (0a) at (0.75,-1) {};         
      \node (i)[state,scale=0.5,label={[yshift=-0.2cm]0:{\tiny$j$}}] at(-0.5,0) {$i$};
          \node (2a)[blackdot] at (0,0.5) {};
          \node (2b) at (0,0.75) {};          

          \draw[string,out=90,in=0] (0a) to (2a);                   
          \draw[string,out=90,in=180] (i) to (2a);
          \draw[string,out=90,in=270] (2a) to (2b.center);
          
          \node (0b) at (0,0.5) {};
          \node (j)[state,black,scale=0.5] at(1.25,1.5) {$j$};
          \node (2a')[ls'] at (0.75,2) {};
          \node (2b') at (0.75,3) {};          
          \draw[string,out=90,in=180] (2b.center) to (2a');                  
          \draw[string,out=90,in=0] (j) to (2a');
          \draw[string,out=90,in=270] (2a') to (2b');
\end{pic}\end{aligned}
=
\begin{aligned}\begin{pic}
         \node[ls] (l) at (-0.25,4.25) {};
         \node[whitedot,label={[xshift=-0.1cm,yshift=-0.1cm]0:{\tiny$k$}}] (k) at (-1,1.5) {};
         \node[state,black,scale=0.5] (b) at (-0.5,1) {$j$};
         \node[state,black,scale=0.5] (c) at (0.25,3.75) {$j$};
         \node[state,hflip,scale=0.5,label={[xshift=0.15cm]0:{\tiny$k$}}] (k1) at (-1,1.75) {$m$};
         \node[state,black,scale=0.5] (m1) at (-1,2.5) {$m$};
         \node[state,hflip,black,scale=0.5] (m2) at (-1.5,0.25) {$m$};
         \node[state,scale=0.5,label={[yshift=-0.2cm]0:{\tiny$k$}}] (k2) at (-1.5,1) {$m$};
         \node[blackdot] (0a) at (-1.5,-1.25){};
         \node (A) at (-2.25,-0.75) {};
         \node (B) at (-2.25,-2.75) {};
         \node[state,scale=0.5,label={[yshift=-0.2cm]0:{\tiny$j$}}] (i) at (-1,-1.75) {$i$};
         \node (0b) at (0.5,-0.25){};
         \node (0c) at (0.5,2.5) {};
         \node (0d) at (-0.25,5.25){};
         \draw[string] (0a) to (m2);
         \draw[string,out=90,in=180] (k2) to (k);
         \draw[string] (k) to (k1);
         \draw[string,out=90,in=180] (m1) to (l);
         \draw[string,out=0,in=90] (l) to (c);

         \draw[string,out=90,in=0] (b) to (k);
         \draw[string] (l) to (0d);
         \draw[string,out=180,in=90] (0a) to (B);
 
         \draw[string,out=0,in=90] (0a) to (i);
                           
\end{pic}\end{aligned}
\end{equation}
\begin{proposition}
$B'_{ij}$ is a unitary error basis. 
\end{proposition}
\begin{proof}
$B_{ij}=P_j \circ D_j \circ H^j_{\text{diag}(i)}$ where $D_j:= 
\sum_m\begin{aligned}\begin{pic}

         \node[whitedot,label={[xshift=-0.1cm,yshift=-0.1cm]0:{\tiny$k$}}] (k) at (-1,1.5) {};
         \node[state,black,scale=0.5] (b) at (-0.5,1) {$j$};

         \node[state,hflip,scale=0.5,label={[xshift=0.15cm]0:{\tiny$k$}}] (k1) at (-1,1.75) {$m$};
         \node[state,black,scale=0.5] (m1) at (-1,2.5) {$m$};
         \node[state,hflip,black,scale=0.5] (m2) at (-1.5,0.25) {$m$};
         \node[state,scale=0.5,label={[yshift=-0.2cm]0:{\tiny$k$}}] (k2) at (-1.5,1) {$m$};
         \node (0a) at (-1.5,-1.25){};

         \node (0b) at (0.5,-0.25){};
         \node (0c) at (0.5,2.5) {};
         \node (0d) at (-0.25,5.25){};
         \draw[string] (-1.5,-0.5) to (m2);
         \draw[string,out=90,in=180] (k2) to (k);
         \draw[string] (k) to (k1);
         \draw[string] (m1) to (-1,3.25);

         \draw[string,out=90,in=0] (b) to (k);

\end{pic}\end{aligned}$
\\$P_j$ and $H^j_{\text{diag}(i)}$ were proven to be unitary in Chapters 2 and 4 respectively. $D_j$ is unitary as follows:\\
\begin{equation}\sum_n \sum_m 
\begin{aligned}\begin{pic}

         \node[whitedot,label={[xshift=-0.1cm,yshift=-0.1cm]0:{\tiny$k$}}] (k) at (-1,-1.5) {};
         \node[state,black,scale=0.5] (b) at (-0.5,-2) {$j$};

         \node[state,hflip,scale=0.5,label={[xshift=0.15cm]0:{\tiny$k$}}] (k1) at (-1,-1.25) {$m$};
         \node[state,black,scale=0.5] (m1) at (-1,-0.5) {$m$};
         \node[state,hflip,black,scale=0.5] (m2) at (-1.5,-2.75) {$m$};
         \node[state,scale=0.5,label={[yshift=-0.2cm]0:{\tiny$k$}}] (k2) at (-1.5,-2) {$m$};
         \node (0a) at (-1.5,-1.25){};

         \node (0b) at (0.5,-0.25){};
         \node (0c) at (0.5,2.5) {};
         \draw[string] (-1.5,-3.5) to (m2);
         \draw[string,out=90,in=180] (k2) to (k);
         \draw[string] (k) to (k1);
         \draw[string] (m1) to (-1,0);

         \draw[string,out=90,in=0] (b) to (k);

                \node[whitedot,label={[xshift=-0.1cm,yshift=-0.1cm]0:{\tiny$k$}}] (k') at (-1,1.5) {};
         \node[state,hflip,black,scale=0.5] (b') at (-0.5,2) {$j$};

         \node[state,scale=0.5,label={[yshift=-0.2cm]0:{\tiny$k$}}] (k1') at (-1,1.25) {$n$};
         \node[state,hflip,black,scale=0.5] (m1') at (-1,0.5) {$n$};
         \node[state,black,scale=0.5] (m2') at (-1.5,2.75) {$n$};
         \node[state,hflip,scale=0.5,label={[xshift=0.15cm]0:{\tiny$k$}}] (k2') at (-1.5,2) {$n$};

         \draw[string] (-1.5,3.5) to (m2');
         \draw[string,out=270,in=180] (k2') to (k');
         \draw[string] (k') to (k1');
         \draw[string] (m1') to (-1,0);

         \draw[string,out=270,in=0] (b') to (k');
                           
\end{pic}\end{aligned}
\quad=\quad
 \sum_m 
\begin{aligned}\begin{pic}

         \node[whitedot,label={[xshift=-0.1cm,yshift=-0.1cm]0:{\tiny$k$}}] (k) at (-1,-1.5) {};
         \node[state,black,scale=0.5] (b) at (-0.5,-2) {$j$};

         \node[state,hflip,scale=0.5,label={[xshift=0.15cm]0:{\tiny$k$}}] (k1) at (-1,-1.25) {$m$};

         \node[state,hflip,black,scale=0.5] (m2) at (-1.5,-2.75) {$m$};
         \node[state,scale=0.5,label={[yshift=-0.2cm]0:{\tiny$k$}}] (k2) at (-1.5,-2) {$m$};
         \node (0a) at (-1.5,-1.25){};

         \node (0b) at (0.5,-0.25){};
         \node (0c) at (0.5,2.5) {};
         \draw[string] (-1.5,-3.5) to (m2);
         \draw[string,out=90,in=180] (k2) to (k);
         \draw[string] (k) to (k1);

         \draw[string,out=90,in=0] (b) to (k);

                \node[whitedot,label={[xshift=-0.1cm,yshift=-0.1cm]0:{\tiny$k$}}] (k') at (-1,-0.25) {};
         \node[state,hflip,black,scale=0.5] (b') at (-0.5,0.25) {$j$};

         \node[state,scale=0.5,label={[yshift=-0.2cm]0:{\tiny$k$}}] (k1') at (-1,-0.5) {$m$};

         \node[state,black,scale=0.5] (m2') at (-1.5,1) {$m$};
         \node[state,hflip,scale=0.5,label={[xshift=0.15cm]0:{\tiny$k$}}] (k2') at (-1.5,0.25) {$m$};

         \draw[string] (-1.5,1.75) to (m2');
         \draw[string,out=270,in=180] (k2') to (k');
         \draw[string] (k') to (k1');

         \draw[string,out=270,in=0] (b') to (k');
                           
\end{pic}\end{aligned}
\quad=\quad
 \sum_m 
\begin{aligned}\begin{pic}

         \node[whitedot,label={[xshift=-0.1cm,yshift=-0.1cm]0:{\tiny$k$}}] (k) at (-1,-1.5) {};
         \node[state,black,scale=0.5] (b) at (-0.5,-2) {$j$};

         \node[state,hflip,black,scale=0.5] (m2) at (-1.5,-2.75) {$m$};
         \node[state,scale=0.5,label={[yshift=-0.2cm]0:{\tiny$k$}}] (k2) at (-1.5,-2) {$m$};
         \node (0a) at (-1.5,-1.25){};

         \node (0b) at (0.5,-0.25){};
         \node (0c) at (0.5,2.5) {};
         \draw[string] (-1.5,-3.5) to (m2);
         \draw[string,out=90,in=180] (k2) to (k);
         \draw[string] (k) to (k1);

         \draw[string,out=90,in=0] (b) to (k);

                \node[whitedot,label={[xshift=-0.1cm,yshift=-0.1cm]0:{\tiny$k$}}] (k') at (-1,-0.25) {};
         \node[state,hflip,black,scale=0.5] (b') at (-0.5,0.25) {$j$};

         \node[state,black,scale=0.5] (m2') at (-1.5,1) {$m$};
         \node[state,hflip,scale=0.5,label={[xshift=0.15cm]0:{\tiny$k$}}] (k2') at (-1.5,0.25) {$m$};

         \draw[string] (-1.5,1.75) to (m2');
         \draw[string,out=270,in=180] (k2') to (k');
         \draw[string] (k') to (k1');

         \draw[string] (k) to (k');
         \draw[string,out=270,in=0] (b') to (k');
                           
\end{pic}\end{aligned}\\
\quad=\quad
 \sum_m 
\begin{aligned}\begin{pic}

         \node[state,hflip,black,scale=0.5] (m2) at (-1.5,-2.75) {$m$};
         \node[state,scale=0.5,label={[yshift=-0.2cm]0:{\tiny$k$}}] (k2) at (-1.5,-2) {$m$};
 
         \draw[string] (-1.5,-3.5) to (m2);


         \node[state,black,scale=0.5] (m2') at (-1.5,-1.25) {$m$};
         \node[state,hflip,scale=0.5,label={[xshift=0.15cm]0:{\tiny$k$}}] (k2') at (-1.5,-2) {$m$};

         \draw[string] (-1.5,-0.5) to (m2');

         \draw[string] (k2) to (k2');                 
\end{pic}\end{aligned}
\quad=\quad
 \sum_m 
\begin{aligned}\begin{pic}

         \node[state,hflip,black,scale=0.5] (m2) at (-1.5,-2.75) {$m$};

         \draw[string] (-1.5,-3.5) to (m2);


         \node[state,black,scale=0.5] (m2') at (-1.5,-2) {$m$};

         \draw[string] (-1.5,-1.25) to (m2');

         \draw[string] (k2) to (k2');                 
\end{pic}\end{aligned}
\quad=\quad
 \sum_m 
\begin{aligned}\begin{pic}
         \draw[string] (-1.5,-3.5) to (-1.5,-1.25);
\end{pic}\end{aligned}\end{equation}
The other direction straightforwardly follows in a similar manner. So $D_j$ is unitary. Thus $B'_{ij}$ is unitary.

Clearly $\sum_{m}
\begin{aligned}\begin{pic}

         \node[whitedot,label={[xshift=-0.1cm,yshift=-0.1cm]0:{\tiny$k$}}] (k) at (-1,1.5) {};
         \node[state,black,scale=0.5] (b) at (-0.5,1) {$j$};

         \node[state,hflip,scale=0.5,label={[xshift=0.15cm]0:{\tiny$k$}}] (k1) at (-1,1.75) {$m$};
         \node[state,black,scale=0.5] (m1) at (-1,2.5) {$m$};
         \node[state,hflip,black,scale=0.5] (m2) at (-1.5,0.25) {$m$};
         \node[state,scale=0.5,label={[yshift=-0.2cm]0:{\tiny$k$}}] (k2) at (-1.5,1) {$m$};
         \node (0a) at (-1.5,-1.25){};

         \node (0b) at (0.5,-0.25){};
         \node (0c) at (0.5,2.5) {};
         \node (0d) at (-0.25,5.25){};
         
         \node[state,black,hflip,scale=0.5] at(-1,3.25){$a$};
         \node[state,black,scale=0.5] at (-1.5,-0.5){$b$}; 
         \draw[string] (-1.5,-0.5) to (m2);
         \draw[string,out=90,in=180] (k2) to (k);
         \draw[string] (k) to (k1);
         \draw[string] (m1) to (-1,3.25);

         \draw[string,out=90,in=0] (b) to (k);

\end{pic}\end{aligned}
\hspace{-5mm}=0, \forall a \neq b.$ Thus  $D_j$ is diagonal $\forall j$.
So $D_j$ can move through spiders by
(5.8). The same is true of $D^{\dag}_j$ which is obviously diagonal too.

\begin{equation}
\tr (B^{\dag}_{ij} \circ B_{i'j'})=
d \begin{aligned}\begin{pic}[scale=0.6]

      \node (0a) at (1.5,-3.5) {};         
      \node (i)[state,scale=0.5,label={[yshift=-0.3cm]0:{\tiny $j'$}}] at(0.5,-3) {$i'$};
          \node (2a)[blackdot] at (1,-2.5) {};
          \node[morphism,wedge,scale=0.5] (2b) at (1,-1.75) {$D_{j'}$};          
          \draw[string,out=90,in=0] (0a.center) to (2a);                   
          \draw[string,out=90,in=180] (i) to (2a);
          \draw[string,out=90,in=270] (2a) to (1,-2);

          \node (j)[state,black,scale=0.5] at(2.25,-1.5) {$j'$};
          \node (2a')[ls] at (1.75,-1) {};
         
          \draw[string,out=90,in=180] (1,-1.5) to (2a');                  
          \draw[string,out=90,in=0] (j) to (2a');

           \node (A) at (1.5,3.5) {};         
      \node (B)[state,hflip,scale=0.5,label={[yshift=0.3cm]0:{\tiny $j$}}] at(0.5,3) {$i$};
          \node (C)[blackdot] at (1,2.5) {};
          \node[morphism,hflip,wedge,scale=0.5] (D) at (1,1.75) {$D_j$};          
          \draw[string,out=270,in=0] (A.center) to (C);                   
          \draw[string,out=270,in=180] (B) to (C);
          \draw[string,out=270,in=90] (C) to (1,2);

          \node (F)[state,black,hflip,scale=0.5] at(2.25,1.5) {$j$};
          \node (G)[ls] at (1.75,1) {};
         
          \draw[string,out=270,in=180] (1,1.5) to (G);                  
          \draw[string,out=270,in=0] (F) to (G);
          \draw[string] (G) to (2a');
          
          \node[blackdot] (X) at (-0.125,4.25) {};
          \node[blackdot] (Y) at (-0.125,-4.25) {};
          \draw[string,in=180,out=180] (X) to (Y);
          \draw[string,in=270,out=0] (Y) to (0a.center);
          \draw[string,in=90,out=0] (X) to (A.center);
\end{pic}\end{aligned}
\quad \overset{\text{(5.8)}}{=}\quad
d \begin{aligned}\begin{pic}[scale=0.6]

      \node (0a) at (1.5,-3.5) {};         
      \node (i)[state,scale=0.5,label={[yshift=-0.3cm]0:{\tiny $j'$}}] at(0.5,-5.75) {$i'$};
          \node (2a)[blackdot] at (1,-2.5) {};
          \node (2b) at (1,-2.25) {};          
          \draw[string,out=90,in=0] (0a.center) to (2a);                   
          \draw[string] (i) to (0.5,-5.5);
          \draw[string,out=90,in=180] (0.5,-5) to (2a);
          \draw[string,out=90,in=270] (2a) to (2b.center);
          
          \node (0b) at (0,-1.5) {};
          \node (j)[state,black,scale=0.5] at(2.25,-1.5) {$j'$};
          \node (2a')[ls] at (1.75,-1) {};
          \node (2b') at (1.75,0) {};          
          \draw[string,out=90,in=180] (2b.center) to (2a');                  
          \draw[string,out=90,in=0] (j) to (2a');

           \node (A) at (1.5,3.5) {};         
      \node (B)[state,hflip,scale=0.5,label={[yshift=0.3cm]0:{\tiny $j$}}] at(0.5,5.75) {$i$};
          \node (C)[blackdot] at (1,2.5) {};
          \node (D) at (1,2.25) {};          
          \draw[string,out=270,in=0] (A.center) to (C);                   
          \draw[string] (B) to (0.5,5.5);
          \draw[string,out=270,in=180] (0.5,5) to (C);
          \draw[string,out=270,in=90] (C) to (D.center);
          
          \node (E) at (0,1.5) {};
          \node (F)[state,black,hflip,scale=0.5] at(2.25,1.5) {$j$};
          \node (G)[ls] at (1.75,1) {};
         
          \draw[string,out=270,in=180] (D.center) to (G);                  
          \draw[string,out=270,in=0] (F) to (G);
          \draw[string] (G) to (2a');
          
          \node[blackdot] (X) at (-0.125,4.25) {};
          \node[blackdot] (Y) at (-0.125,-4.25) {};
          \draw[string,in=180,out=180] (X) to (Y);
          \draw[string,in=270,out=0] (Y) to (0a.center);
          \draw[string,in=90,out=0] (X) to (A.center);

          \node[morphism,hflip,wedge,scale=0.5] (6) at (0.5,5.25) {$D_j$};
          \node[morphism,wedge,scale=0.5] (7) at (0.5,-5.25) {$D_{j'}$};
\end{pic}\end{aligned}
\quad $$\\$$ \overset{2 \times \text{(Bl,SM)}}{=}\quad
d \begin{aligned}\begin{pic}[yscale=0.375,xscale=0.5]
 \node[state,scale=0.5,label={[yshift=-0.3cm]0:{\tiny $j'$}}] (0) at (-1.5,-4.5) {$i'$};
\node[ls] (1) at (0,-1) {};
\node (2) at (-2.5,0) {};
\node[blackdot] (B) at (-1.5,-2) {};
\node[state,black,scale=0.5] (3) at (1,-4.5) {$j'$};
\node (4) at (0,0) {};
\draw[string,out=0,in=90] (1) to (3);
\draw [string, out=90,in=270] (0) to (-1.5,-3.65);
\draw [string, out=90,in=270] (-1.5,-2.85) to (B);
\draw [string, out=0,in=180] (B) to (1);
\draw [string, out=180,in=270] (B) to (2.center);
\draw[string] (1) to (4.center);   

 \node[state,hflip,scale=0.5,label={[yshift=0.3cm]0:{\tiny $j$}}] (5) at (-1.5,4.5) {$i$};
\node[ls] (6) at (0,1) {};

\node[blackdot] (A) at (-1.5,2) {};
\node[state,hflip,black,scale=0.5] (8) at (1,4.5) {$j$};

\draw[string,out=0,in=270] (6) to (8);
\draw [string, out=270,in=90] (5) to (-1.5,3.65);
\draw [string, out=270,in=90] (-1.5,2.85) to (A);
\draw [string, out=0,in=180] (A) to (6);
\draw [string, out=180,in=90] (A) to (2.center);
\draw[string] (6) to (4.center);

          \node[morphism,hflip,wedge,scale=0.5] (6) at (-1.5,3.25) {$D_j$};
          \node[morphism,wedge,scale=0.5] (7) at (-1.5,-3.25) {$D_{j'}$};
\end{pic}\end{aligned}
\quad=\quad
\begin{aligned}\begin{pic}[yscale=0.375,xscale=0.5]
 \node[state,scale=0.5,label={[yshift=-0.3cm]0:{\tiny $j'$}}] (0) at (-1.5,1.5) {$i'$};
\node[state,black,scale=0.5] (3) at (0,3.5) {$j'$}; 
 \node[state,hflip,scale=0.5,label={[yshift=0.3cm]0:{\tiny $j$}}] (5) at (-1.5,5.5) {$i$};
\node[state,hflip,black,scale=0.5] (8) at (0,3.5) {$j$};
          \node[morphism,hflip,wedge,scale=0.5] (6) at (-1.5,4.5) {$D_j$};
          \node[morphism,wedge,scale=0.5] (7) at (-1.5,2.5) {$D_{j'}$};
          \draw[string] (0) to (-1.5,2.1);
          \draw[string] (-1.5,2.9) to (-1.5,4.1);
          \draw[string] (-1.5,4.9) to (5);
          
\end{pic}\end{aligned} d
\quad=\quad
\begin{aligned}\begin{pic}[yscale=0.375,xscale=0.5]
 \node[state,scale=0.5,label={[yshift=-0.3cm]0:{\tiny $j$}}] (0) at (-1.5,1.5) {$i'$};

 \node[state,hflip,scale=0.5,label={[yshift=0.3cm]0:{\tiny $j$}}] (5) at (-1.5,5.5) {$i$};

          \node[morphism,hflip,wedge,scale=0.5] (6) at (-1.5,4.5) {$D_j$};
          \node[morphism,wedge,scale=0.5] (7) at (-1.5,2.5) {$D_j$};
          \draw[string] (0) to (-1.5,2.1);
          \draw[string] (-1.5,2.9) to (-1.5,4.1);
          \draw[string] (-1.5,4.9) to (5);
          
\end{pic}\end{aligned} \delta_{jj'}d
\quad=\quad
\begin{aligned}\begin{pic}[yscale=0.375,xscale=0.5]
 \node[state,scale=0.5,label={[yshift=-0.3cm]0:{\tiny $j$}}] (0) at (-1.5,3.5) {$i'$};
 
 \node[state,hflip,scale=0.5,label={[yshift=0.3cm]0:{\tiny $j$}}] (5) at (-1.5,3.5) {$i$};

\end{pic}\end{aligned}
\delta_{jj'}d
=\delta_{ii'} \delta_{jj'}d
\end{equation}
\end{proof} 
 Since $D_j$ varies with $j$ the unitary error basis $B'_{ij}$ is not obviously equivalent to any shift and multiply basis. Clearly there are many other possible models for generalised latin structures, perhaps varying the choice of white classical structure with $i$ or $j$ or adding a phase which varies with $i$ or $j$. 

\section{Generalised Latin Square}

I will now present a generalised latin square order $6$ that fits my model (5.9), is not a latin square, but does produce a UEB via the generalised combinatorial construction.
I have chosen $d=6$ because it is the lowest dimension for which there exist both non-isotopic latin squares and non-equivalent Hadamard matrices.
As input I have taken the following non-associative latin square found in this paper~\cite{latinsquare}:

\begin{equation*}
\begin{tabular}[]{|c|c|c|c|c|c|}

\hline

$a$  & $b$  & $c$  & $d$ & $e$ & $f$\\

\hline

$b$  & $a$  & $e$ & $f$ & $c$  & $d$\tabularnewline

\hline

$c$  & $f$ & $b$  & $a$  & $d$ & $e$\tabularnewline

\hline

$d$ & $e$ & $a$  & $b$  & $f$ & $c$\tabularnewline

\hline

$e$ & $d$ & $f$ & $c$  & $b$  & $a$\tabularnewline

\hline

$f$ & $c$  & $d$ & $e$ & $a$  & $b$\tabularnewline

\hline

\end{tabular}
\end{equation*}
As the $k^{th}$ member of my family of Hadamard matrices I have chosen $C^{(0)}_6$, which can be found in this paper~\cite{Hadamard} :

\begin{equation*}
\begin{pmatrix}
   
    1
    &
    1 
    & 1
    & 1
    & 1
    & 1
    \\
    1
    & -1 & -p & -p^{2} & p^2 & p
    \\
    1 & -{\overline{p}} & 1 & p^2 & -p^3 & p^2
    \\
     1 & {\overline{p}}^2 & {\overline{p}}^2 & -1  & p^2 & -p^2
    \\
    1 & {\overline{p}}^2 & -{\overline{p}}^3 & {\overline{p}}^2 & 1 & -p
    \\
    1
    & {\overline{p}} & {\overline{p}}^2 & -{\overline{p}}^2 & -{\overline{p}} &
    -1
  \end{pmatrix}
  \end{equation*}
Where $p=\frac{1- \sqrt{3}}{2}+i\sqrt{(\frac{\sqrt{3}}{2}}).$

So equation (5.9) gives the following generalised latin square:
\begin{equation*}
\begin{aligned}\begin{tabular}{|c|c|c|c|c|c|}

\hline
\, \, $a$ \, \,  & $\frac{b}{\sqrt{6}}$ & $\frac{c}{\sqrt{6}}$ & $\frac{d}{\sqrt{6}}$ & $\frac{e}{\sqrt{6}}$ & $\frac{f}{\sqrt{6}}$\\[5pt]

\hline

$\frac{b}{\sqrt{6}}$ & $-a$ & $\frac{-e{\overline{p}}}{\sqrt{6}}$ & $\frac{-f{\overline{p}}^{2}}{\sqrt{6}}$ & $\frac{c\overline{p}^{2}}{\sqrt{6}}$ & $\frac{d{\overline{p}}}{\sqrt{6}}$\\[5pt]

\hline

$\frac{c}{\sqrt{6}}$ & $\frac{-fp}{\sqrt{6}}$ & $\frac{b}{\sqrt{6}}$ & ${\overline{p}}^{2}a$ & $\frac{-d{\overline{p}}^{3}}{\sqrt{6}}$ & $\frac{e{\overline{p}}^{2}}{\sqrt{6}}$\\[5pt]

\hline

$\frac{d}{\sqrt{6}}$ & $\frac{-ep^2}{\sqrt{6}}$ & $p^2a$ & $\frac{-b}{\sqrt{6}}$ & $\frac{f{\overline{p}}^{2}}{\sqrt{6}}$ & $\frac{-c{\overline{p}}^{2}}{\sqrt{6}}$\\[5pt]

\hline

$\frac{e}{\sqrt{6}}$ & $\frac{dp^2}{\sqrt{6}}$ & $\frac{-fp^{3}}{\sqrt{6}}$ & $\frac{cp^2}{\sqrt{6}}$ & $\frac{b}{\sqrt{6}}$ & $-pa$\\[5pt]

\hline

$\frac{f}{\sqrt{6}}$ & $\frac{cp}{\sqrt{6}}$ & $\frac{dp^2}{\sqrt{6}}$ & $\frac{-ep^2}{\sqrt{6}}$ & $\overline{p}a$ & $\frac{-b}{\sqrt{6}}$\\[5pt]

\hline

\end{tabular}
\end{aligned}
\end{equation*}
Where $p$ is as above.


\chapter{Conclusion}
In this dissertation I have achieved the aims that I set out to accomplish. Through representing both shift and multiply bases and MUB error bases in the graphical calculus of categorical quantum mechanics I have been able to find out where the MUB construction fits into the unitary error basis picture. I have proven that the MUB construction is strictly less general than the combinatorial construction.

In explicitly formulating latin squares as latin square structures I have opened up this combinatorial object for further investigation using categorical quantum mechanics.  

Using the high level tools of categorical quantum mechanics I have also been able to formulate shift and multiply bases in a way which is enlightening as to how the construction works, and additionally gives insight into exactly which parts of the construction are strictly necessary to producing unitary error bases. I have thereby generalised the combinatorial construction and given an explicit model of a generalised shift and multiply basis. I have also presented a specific example of a generalised latin square that is not a latin square but does give rise to a UEB.

The connection with quantum teleportation and dense coding protocols makes this line of enquiry a useful one. As we get closer to fully characterising unitary error bases we gain valuable information about these protocols which have potentially profound implications for the fields of quantum information and quantum computing.

I have worked within the symmetric monoidal category \cat{FHilb}, but due to the use of category theory all my results can be carried over into other symmetric monoidal categories.


\chapter{Further Work}
The results of this dissertation naturally open up various avenues of enquiry:

\begin{itemize}
\item to explicitly prove whether a generalised shift and multiply basis is equivalent to a shift and multiply basis;
\item if not then to explore other models of the generalised construction;
\item to explore  latin square structures generalised to other symmetric monoidal categories;
\item to explore the analogues of the various UEB constructions in other symmetric monoidal categories, such as \cat{Rel} as well as the abstract quantum teleportation protocols they produce;
\item to formulate the algebraic construction of UEBs in the graphical calculus in order to investigate similarities with and differences from the combinatorial construction and generalisations.
\end{itemize}


\appendix

\chapter{Glossary of Rules}

In the course of this paper I will make use of diagramatic proof extensively. It is therefore pragmatic to have a glossary of the rules I will use with abbreviations which I will utilise to show which rule I have used to get from one diagram to the next. 

\section{Rules Applying to Classical Structures}
\

\underline{(SM) Spider Merge / Unmerge}\\
\begin{equation}
    \begin{pic}
     \node[blackdot,scale=2] (A){};
     \draw[string,out=190,in=90] (A) to (-2,-1);
     \draw[string,out=350,in=90] (A) to (2,-1);
     \draw[string,out=190,in=90] (A) to (-1.7,-1);
     \draw[string,out=350,in=90] (A) to (1.7,-1);  
     \draw[string,out=10,in=270] (A) to (2,1);
     \draw[string,out=20,in=260] (A) to (2,1);  
     \draw[loosely dotted] (-1.65,-0.9) to (1.65,-0.9); 
     \node[blackdot,scale=2] at (2,1) (B){};
     \draw[string,out=190,in=90] (B)  to (A);
     \draw[string,out=200,in=80] (B) to (A);  
     \draw[string,out=170,in=270] (B) to (0,2);
     \draw[string,out=10,in=270] (B) to (4,2);
     \draw[string,out=170,in=270] (B) to (0.3,2);
     \draw[string,out=10,in=270] (B) to (3.7,2); 
     \draw[loosely dotted] (0.35,1.9) to (3.65,1.9); 
     \draw[dotted] (0.6,0.7) to (1.3,0.35);
     \draw[decoration={brace,amplitude=7pt},decorate] (-0.1,2.05) to (4.1,2.05);
     \draw[decoration={brace,mirror,amplitude=7pt},decorate] (-2.1,-1.05) to (2.1,-1.05);       
     \node at (0,-1.4) {{\tiny $m$}};
     \node at (2,2.4) {{\tiny $n$}};   
    \end{pic}
    \quad=\quad
    \begin{pic}[yscale=3/2]
     \node[blackdot,scale=2] (A){};
     \draw[string,out=190,in=90] (A) to (-2,-1);
     \draw[string,out=350,in=90] (A) to (2,-1);
     \draw[string,out=190,in=90] (A) to (-1.7,-1);
     \draw[string,out=350,in=90] (A) to (1.7,-1);  
     \draw[string,out=170,in=270] (A) to (-2,1);
     \draw[string,out=10,in=270] (A) to (2,1);
     \draw[string,out=170,in=270] (A) to (-1.7,1);
     \draw[string,out=10,in=270] (A) to (1.7,1); 
     \draw[loosely dotted] (-1.65,0.9) to (1.65,0.9); 
     \draw[loosely dotted] (-1.65,-0.9) to (1.65,-0.9);
          \draw[decoration={brace,amplitude=7pt},decorate] (-2.1,1.033) to (2.1,1.033);
     \draw[decoration={brace,mirror,amplitude=7pt},decorate] (-2.1,-1.033) to (2.1,-1.033);       
     \node at (0,-1.26) {{\tiny $m$}};
     \node at (0,1.26) {{\tiny $n$}};       
    \end{pic}
  \end{equation}
\underline{(A)Associativity}
\begin{equation}
\begin{aligned}\begin{tikzpicture}
          \node (0a) at (-1,0.25) {};
          \node (0b) at (-0.5,0.25) {};
          \node (0c) at (0.5,0.25) {};

          \node[dot] (1) at (0,1) {};
          \node[dot] (2) at (-0.5,1.5) {};

          \node (5a) at (-0.5,2) {};

          \draw[string, out=90, in =180] (0a) to (2);
          \draw[string, out=0, in=90] (2) to (1);
          \draw[string, out=0, in=90] (1) to (0c);
          \draw[string, out=180, in=90] (1) to (0b);

          \draw[string] (2) to (5a.center);
         
          \end{tikzpicture}\end{aligned}   
  \quad   = \quad
\begin{aligned}\begin{tikzpicture}
          \node (0a) at (0.75,0.25) {};
          \node (0b) at (0.25,0.25) {};
          \node (0c) at (-0.75,0.25) {};

          \node[dot] (1) at (-0.25,1) {};
          \node[dot] (2) at (0.25,1.5) {};

          \node (5a) at (0.25,2) {};

          \draw[string, out=90, in =00] (0a) to (2);
          \draw[string, out=180, in=90] (2) to (1);
          \draw[string, out=180, in=90] (1) to (0c);
          \draw[string, out=0, in=90] (1) to (0b);

          \draw[string] (2) to (5a.center);
\end{tikzpicture}\end{aligned}  
\end{equation}
  
\underline{(CA)Coassociativity}
\begin{equation}
\begin{aligned}\begin{tikzpicture}[yscale=-1]
          \node (0a) at (-1,0.25) {};
          \node (0b) at (-0.5,0.25) {};
          \node (0c) at (0.5,0.25) {};

          \node[dot] (1) at (0,1) {};
          \node[dot] (2) at (-0.5,1.5) {};

          \node (5a) at (-0.5,2) {};

          \draw[string, out=90, in =180] (0a.center) to (2);
          \draw[string, out=0, in=90] (2) to (1);
          \draw[string, out=0, in=90] (1) to (0c.center);
          \draw[string, out=180, in=90] (1) to (0b.center);

          \draw[string] (2) to (5a.center);
         
          \end{tikzpicture}\end{aligned}   
  \quad   = \quad
\begin{aligned}\begin{tikzpicture}[yscale=-1]
          \node (0a) at (0.75,0.25) {};
          \node (0b) at (0.25,0.25) {};
          \node (0c) at (-0.75,0.25) {};

          \node[dot] (1) at (-0.25,1) {};
          \node[dot] (2) at (0.25,1.5) {};

          \node (5a) at (0.25,2) {};

          \draw[string, out=90, in =00] (0a.center) to (2);
          \draw[string, out=180, in=90] (2) to (1);
          \draw[string, out=180, in=90] (1) to (0c.center);
          \draw[string, out=0, in=90] (1) to (0b.center);

          \draw[string] (2) to (5a.center);
\end{tikzpicture}\end{aligned}  
  \end{equation}

  \underline{(C)Commutativity}
\begin{equation}
\begin{aligned}\begin{tikzpicture}
          \node (0a) at (-1,0.25) {};
          \node (0b) at (-0.5,0.25) {};
          \node (0c) at (0,0.25) {};

          \node[dot] (2) at (-0.5,1.5) {};

          \node (5a) at (-0.5,2) {};

          \draw[string, out=90, in =180] (0a) to (2);
          \draw[string, out=0, in=90] (2) to (0c);

          \draw[string] (2) to (5a.center);
         
          \end{tikzpicture}\end{aligned}   
  \quad   = \quad
\begin{aligned}\begin{tikzpicture}
           \node (0a) at (-1,0.25) {};
          \node (0b) at (-0.5,0.25) {};
          \node (0c) at (0,0.25) {};

          \node[dot] (2) at (-0.5,1.5) {};

          \node (5a) at (-0.5,2) {};

          \draw[string, out=90, in =0] (0a) to (2);
          \draw[string, out=180, in=90] (2) to (0c);
          \draw[string] (2) to (5a.center);
\end{tikzpicture}\end{aligned}  
  \end{equation}
  
  \underline{(CC)Cocommutativity}
\begin{equation}
\begin{aligned}\begin{tikzpicture}[yscale=-1]
          \node (0a) at (-1,0.25) {};
          \node (0b) at (-0.5,0.25) {};
          \node (0c) at (0,0.25) {};

          \node[dot] (2) at (-0.5,1.5) {};

          \node (5a) at (-0.5,2) {};

          \draw[string, out=90, in =180] (0a.center) to (2);
          \draw[string, out=0, in=90] (2) to (0c.center);

          \draw[string] (2) to (5a.center);
         
          \end{tikzpicture}\end{aligned}   
  \quad   = \quad
\begin{aligned}\begin{tikzpicture}[yscale=-1]
           \node (0a) at (-1,0.25) {};
          \node (0b) at (-0.5,0.25) {};
          \node (0c) at (0,0.25) {};

          \node[dot] (2) at (-0.5,1.5) {};

          \node (5a) at (-0.5,2) {};

          \draw[string, out=90, in =0] (0a.center) to (2);
          \draw[string, out=180, in=90] (2) to (0c.center);
          \draw[string] (2) to (5a.center);
\end{tikzpicture}\end{aligned}  
  \end{equation}

 \underline{(S)Specialness}
\begin{equation}
\begin{aligned}\begin{tikzpicture}
          \node (0a) at (-1,0) {};
          \node (0b) at (-0.5,0) {};
          \node (0c) at (0,0) {};

          \node[dot] (2) at (-0.5,0.5) {};
          \node (5a) at (-0.5,1) {};
 
          \node[dot] (7) at (-0.5,-0.5) {};       
          \node (10a) at (-0.5,-1) {};

          \draw[string, out=90, in =180] (0a.center) to (2);
          \draw[string, out=0, in=90] (2) to (0c.center);

          \draw[string] (2) to (5a.center);
          
          \draw[string, out=270, in =180] (0a.center) to (7);
          \draw[string, out=0, in=270] (7) to (0c.center);

          \draw[string] (7) to (10a.center);
         
          \end{tikzpicture}\end{aligned}   
  \quad   = \quad
\begin{aligned}\begin{tikzpicture}
           \node (10a) at (-0.5,-1) {}; 
           \node (5a) at (-0.5,1) {}; 
           \draw[string] (5a.center) to (10a.center);
\end{tikzpicture}\end{aligned}  
  \end{equation}

\underline{(U) Unitality}
\begin{equation}
\begin{aligned}\begin{tikzpicture}
          
          \node (0a) at (-1,0) {};
          \node (0b) at (-0.5,0) {};
          \node[dot] (0c) at (0,0) {};
          \node[dot] (7) at (-0.5,-0.5) {};       
          \node (10a) at (-0.5,-1) {};       
          \draw[string, out=270, in =180] (0a.center) to (7);
          \draw[string, out=0, in=270] (7) to (0c.center);
          \draw[string] (7) to (10a.center);
         
          \end{tikzpicture}\end{aligned}   
  \quad   = \quad
\begin{aligned}\begin{tikzpicture}
           \node (10a) at (-0.5,0) {}; 
           \node (5a) at (-0.5,1) {}; 
           \draw[string] (5a.center) to (10a.center);
\end{tikzpicture}\end{aligned}
\quad=\quad
\begin{aligned}\begin{tikzpicture}

          \node[dot] (0a) at (-1,0) {};
          \node (0b) at (-0.5,0) {};
          \node (0c) at (0,0) {};
          \node[dot] (7) at (-0.5,-0.5) {};       
          \node (10a) at (-0.5,-1) {};       
          \draw[string, out=270, in =180] (0a.center) to (7);
          \draw[string, out=0, in=270] (7) to (0c.center);
          \draw[string] (7) to (10a.center);

\end{tikzpicture}\end{aligned}  
  \end{equation}

\underline{(CU) Counitality}
\begin{equation}
\begin{aligned}\begin{tikzpicture}
          
          \node (0a) at (-1,-0.5) {};
          \node (0b) at (-0.5,-0.5) {};
          \node[dot] (0c) at (0,-0.5) {};
          \node[dot] (7) at (-0.5,0) {};       
          \node (10a) at (-0.5,0.5) {};       
          \draw[string, out=90, in =180] (0a.center) to (7);
          \draw[string, out=0, in=90] (7) to (0c.center);
          \draw[string] (7) to (10a.center);
         
          \end{tikzpicture}\end{aligned}   
  \quad   = \quad
\begin{aligned}\begin{tikzpicture}
           \node (10a) at (-0.5,0) {}; 
           \node (5a) at (-0.5,1) {}; 
           \draw[string] (5a.center) to (10a.center);
\end{tikzpicture}\end{aligned}
\quad=\quad
\begin{aligned}\begin{tikzpicture}

      \node[dot] (0a) at (-1,-0.5) {};
          \node (0b) at (-0.5,-0.5) {};
          \node (0c) at (0,-0.5) {};
          \node[dot] (7) at (-0.5,0) {};       
          \node (10a) at (-0.5,0.5) {};       
          \draw[string, out=90, in =180] (0a.center) to (7);
          \draw[string, out=0, in=90] (7) to (0c.center);
          \draw[string] (7) to (10a.center);

\end{tikzpicture}\end{aligned}  
  \end{equation}
\section{Rules Applying to Latin Square Structures}
\

In what follows \tinymultls \, represents a latin square structure and $\tinymult[blackdot]$ represents the classical structure corresponding to the orthonormal basis associated to the elements of the underlying latin square.

\underline{(LS1) $\begin{pic}[scale=0.2]
\node (0) at (0.5,-0.5) {};
\node[blackdot,scale=0.7] (1) at (-1,2) {};
\node (2) at (1.5,3.5) {};
\node[ls,scale=0.7] (B) at (0.5,1) {};
\node (3) at (-2,-0.5) {};
\node (4) at (-1,3.5) {};
\draw[string,out=180,in=90] (1) to (3);
\draw [string, out=90,in=270] (0) to (B);
\draw [string, out=180,in=0] (B) to (1);
\draw [string, out=0,in=270] (B) to (2);
\draw[string] (1) to (4);
\end{pic}$ is unitary}

\begin{equation}
\begin{aligned}\begin{tikzpicture}[scale=0.5]
       \node (0) at (0.5,-3.5) {};
\node[blackdot] (1) at (-1,-1) {};
\node (2) at (1.5,0) {};
\node[ls] (B) at (0.5,-2) {};
\node (3) at (-2,-3.5) {};
\node (4) at (-1,0) {};
\draw[string,out=180,in=90] (1) to (3);
\draw [string, out=90,in=270] (0) to (B);
\draw [string, out=180,in=0] (B) to (1);
\draw [string, out=0,in=270] (B) to (2.center);
\draw[string] (1) to (4.center);   

 \node (5) at (0.5,3.5) {};
\node[blackdot] (6) at (-1,1) {};

\node[ls] (A) at (0.5,2) {};
\node (8) at (-2,3.5) {};

\draw[string,out=180,in=270] (6) to (8);
\draw [string, out=270,in=90] (5) to (A);
\draw [string, out=180,in=0] (A) to (6);
\draw [string, out=0,in=90] (A) to (2.center);
\draw[string] (6) to (4.center);  

          \end{tikzpicture}\end{aligned}   
  \quad   = \quad
\begin{aligned}\begin{tikzpicture}[scale=0.5]
           \node (10a) at (0,0) {}; 
           \node (5a) at (0,7) {}; 
           \node (10) at (1.25,0) {}; 
           \node (5) at (1.25,7) {};
           \draw[string] (5a.center) to (10a.center);
           \draw[string] (5.center) to (10.center);
\end{tikzpicture}\end{aligned}
\quad=\quad
\begin{aligned}\begin{tikzpicture}[scale=0.5]

         \node (0) at (0.5,0) {};
\node[blackdot] (1) at (-1,2) {};
\node (2) at (1.5,3.5) {};
\node[ls] (B) at (0.5,1) {};
\node (3) at (-2,0) {};
\node (4) at (-1,3.5) {};
\draw[string,out=180,in=90] (1) to (3.center);
\draw [string, out=90,in=270] (0.center) to (B);
\draw [string, out=180,in=0] (B) to (1);
\draw [string, out=0,in=270] (B) to (2);
\draw[string] (1) to (4);

\node[blackdot] (6) at (-1,-2) {};
\node (7) at (1.5,-3.5) {};
\node[ls] (A) at (0.5,-1) {};

\node (9) at (-1,-3.5) {};
\draw[string,out=180,in=270] (6) to (3.center);
\draw [string, out=270,in=90] (0.center) to (A);
\draw [string, out=180,in=0] (A) to (6);
\draw [string, out=0,in=90] (A) to (7);
\draw[string] (6) to (9);

\end{tikzpicture}\end{aligned}  
  \end{equation}
  
\underline{(LS2) $\begin{pic}[scale=0.2]
\node (0) at (0.5,-0.5) {};
\node[ls,scale=0.7] (1) at (-1,2) {};
\node (2) at (1.5,3.5) {};
\node[blackdot,scale=0.7] (B) at (0.5,1) {};
\node (3) at (-2,-0.5) {};
\node (4) at (-1,3.5) {};
\draw[string,out=180,in=90] (1) to (3);
\draw [string, out=90,in=270] (0) to (B);
\draw [string, out=180,in=0] (B) to (1);
\draw [string, out=0,in=270] (B) to (2);
\draw[string] (1) to (4);
\end{pic}$ is unitary}

\begin{equation}
\begin{aligned}\begin{tikzpicture}[scale=0.5]
         \node (0) at (0.5,0) {};
\node[ls] (1) at (-1,2) {};
\node (2) at (1.5,3.5) {};
\node[blackdot] (B) at (0.5,1) {};
\node (3) at (-2,0) {};
\node (4) at (-1,3.5) {};
\draw[string,out=180,in=90] (1) to (3.center);
\draw [string, out=90,in=270] (0.center) to (B);
\draw [string, out=180,in=0] (B) to (1);
\draw [string, out=0,in=270] (B) to (2);
\draw[string] (1) to (4);

\node[ls] (6) at (-1,-2) {};
\node (7) at (1.5,-3.5) {};
\node[blackdot] (A) at (0.5,-1) {};

\node (9) at (-1,-3.5) {};
\draw[string,out=180,in=270] (6) to (3.center);
\draw [string, out=270,in=90] (0.center) to (A);
\draw [string, out=180,in=0] (A) to (6);
\draw [string, out=0,in=90] (A) to (7);
\draw[string] (6) to (9);

\end{tikzpicture}\end{aligned}
\quad=\quad
\begin{aligned}\begin{tikzpicture}[scale=0.5]
           \node (10a) at (0,0) {}; 
           \node (5a) at (0,7) {}; 
           \node (10) at (1.25,0) {}; 
           \node (5) at (1.25,7) {};
           \draw[string] (5a.center) to (10a.center);
           \draw[string] (5.center) to (10.center);
\end{tikzpicture}\end{aligned}
\quad=\quad
\begin{aligned}\begin{tikzpicture}[scale=0.5]
 \node (0) at (0.5,-3.5) {};
\node[ls] (1) at (-1,-1) {};
\node (2) at (1.5,0) {};
\node[blackdot] (B) at (0.5,-2) {};
\node (3) at (-2,-3.5) {};
\node (4) at (-1,0) {};
\draw[string,out=180,in=90] (1) to (3);
\draw [string, out=90,in=270] (0) to (B);
\draw [string, out=180,in=0] (B) to (1);
\draw [string, out=0,in=270] (B) to (2.center);
\draw[string] (1) to (4.center);   

 \node (5) at (0.5,3.5) {};
\node[ls] (6) at (-1,1) {};

\node[blackdot] (A) at (0.5,2) {};
\node (8) at (-2,3.5) {};

\draw[string,out=180,in=270] (6) to (8);
\draw [string, out=270,in=90] (5) to (A);
\draw [string, out=180,in=0] (A) to (6);
\draw [string, out=0,in=90] (A) to (2.center);
\draw[string] (6) to (4.center);

\end{tikzpicture}\end{aligned}  
  \end{equation}

\underline{(SN) Snake equation}
\begin{equation}
 \begin{aligned}\begin{tikzpicture}

          \node (2a) at (1.25,0.5) {};
          \node (2b)[whitedot] at (0.75,1.75) {};
          \node (3a)[whitedot] at (0,0.75){};
          \node (3b) at (-0.5,2){};

          \draw[string,out=180,in=0] (2b) to (3a);
          \draw[string,out=180,in=270] (3a) to (3b);

          \draw[string,out=90,in=00] (2a) to (2b);
      \end{tikzpicture}\end{aligned} 
      \quad= \quad
      \begin{aligned}\begin{tikzpicture}
      \node (A) at (0,0) {};
      \node(B) at (0,3) {};
      \draw[string] (A) to (B);
      \end{tikzpicture}\end{aligned}  
  \quad=\quad
    \begin{aligned}\begin{tikzpicture}

          \node (2a) at (1.75,1.5) {};
          \node (2d)[whitedot] at (1.25,0.25) {};
          \node (3d)[whitedot] at (0.5,1.25){};
          \node (3e) at (0,0){};
          
          \draw[string,out=180,in=0] (2d) to (3d);
          \draw[string,out=180,in=90] (3d) to (3e);
         \draw[string,out=270,in=00] (2a) to (2d);
          
      \end{tikzpicture}\end{aligned}    
      \end{equation}  
\textit{Note: both latin square structures and classical structures obey this rule.}      
       
\underline{(D)Relationship Between Multiplication and Comultiplication}    \begin{equation}
\begin{aligned}\begin{tikzpicture}
          
          \node (0a) at (-1,-0.5) {};
          \node (0b) at (-0.5,-0.5) {};
          \node (0c) at (0,-0.5) {};
          \node[ls] (7) at (-0.5,0) {};       
          \node (10a) at (-0.5,0.5) {};       
          \draw[string, out=90, in =180] (0a.center) to (7);
          \draw[string, out=0, in=90] (7) to (0c.center);
          \draw[string] (7) to (10a.center);
         
          \end{tikzpicture}\end{aligned}
          \quad=\quad
\begin{aligned}\begin{tikzpicture}
          
          \node (0a') at (-2,-1) {};
          \node[blackdot] (a) at (-1,0.5) {};
          
          \node (0a) at (-2,-0.5) {};
          \node (0c) at (-1.5,-0.5) {}; 
                            
          \node[blackdot] (b) at (-1,1) {};
          \node[blackdot] (c) at (0,-0.5) {};
          \node (0b) at (-0.5,-0.5) {};
          \node (0c') at (-1.5,-1) {};
          \node[ls] (7) at (-0.5,0) {};       
          \node (10a) at (0.5,1.5) {};       
          \draw[string, out=90, in =180] (0c.center) to (a);
          \draw[string, out=180, in =0] (7) to (a);
          \draw[string, out=0, in=0] (7) to (b);
          \draw[string, out=180, in=90] (b) to (0a.center);      
          \draw[string,out=270 ,in=180] (7) to (c);
          \draw[string,out=0,in=270] (c) to (10a); 
          \draw[string,out=270,in=90] (0a.center) to (0c');   
          \draw[string,out=270,in=90] (0c.center) to (0a');  
          \end{tikzpicture}\end{aligned}
 \end{equation} 
\\
 \begin{equation}
\begin{aligned}\begin{tikzpicture}
          
          \node (0a) at (-1,0.5) {};
          \node (0b) at (-0.5,0.5) {};
          \node (0c) at (0,0.5) {};
          \node[ls] (7) at (-0.5,0) {};       
          \node (10a) at (-0.5,-0.5) {};       
          \draw[string, out=270, in =180] (0a.center) to (7);
          \draw[string, out=0, in=270] (7) to (0c.center);
          \draw[string] (7) to (10a.center);
         
          \end{tikzpicture}\end{aligned}
          \quad=\quad
\begin{aligned}\begin{tikzpicture}
          
          \node (0a') at (-2,1) {};
          \node[blackdot] (a) at (-1,-0.5) {};
          
          \node (0a) at (-2,0.5) {};
          \node (0c) at (-1.5,0.5) {}; 
                            
          \node[blackdot] (b) at (-1,-1) {};
          \node[blackdot] (c) at (0,0.5) {};
          \node (0b) at (-0.5,0.5) {};
          \node (0c') at (-1.5,1) {};
          \node[ls] (7) at (-0.5,0) {};       
          \node (10a) at (0.5,-1.5) {};       
          \draw[string, out=270, in =180] (0c.center) to (a);
          \draw[string, out=180, in =0] (7) to (a);
          \draw[string, out=0, in=0] (7) to (b);
          \draw[string, out=180, in=270] (b) to (0a.center);      
          \draw[string,out=90 ,in=180] (7) to (c);
          \draw[string,out=0,in=90] (c) to (10a); 
          \draw[string,out=90,in=270] (0a.center) to (0c');   
          \draw[string,out=90,in=270] (0c.center) to (0a');  
          \end{tikzpicture}\end{aligned}
 \end{equation}

\addcontentsline{toc}{chapter}{Bibliography}
\bibliography{aaThesis}        
\bibliographystyle{plain}  

\end{document}